\documentclass[reprint,amsmath,amssymb,aps,nofootinbib]{revtex4-2}
\usepackage{bm}
\usepackage{graphicx}
\usepackage{amsmath}
\usepackage{amssymb}
\usepackage{amsthm}
\usepackage{braket}
\usepackage{siunitx}
\usepackage{color}
\usepackage[utf8]{inputenc}
\usepackage{hyperref}
\usepackage{here}
\usepackage{ytableau} 
\usepackage{xcolor}
\usepackage{dsfont}
\usepackage{here}

\usepackage{tikz}
\usepackage{tikz-cd}
\usetikzlibrary{arrows}
\usetikzlibrary{intersections}
\usetikzlibrary{shapes.geometric}
\usetikzlibrary{decorations.pathmorphing, patterns,shapes}
\usetikzlibrary{decorations.markings}

\usepackage[most]{tcolorbox}
\usepackage[lmargin=.7in,rmargin=.7in,tmargin=.7in,bmargin=1in]{geometry}

\hypersetup{ 
setpagesize=false,
 bookmarksnumbered=true,%
 bookmarksopen=true,%
 colorlinks=true,%
 linkcolor=blue,
 citecolor=blue,
 urlcolor=blue
}

\theoremstyle{plain}
\newtheorem{thm}{Theorem}
\newtheorem{lm}{Lemma}
\newtheorem{cor}{Corollary}
\newtheorem*{thm*}{Theorem}
\theoremstyle{definition}
\newtheorem{dfn}{Definition}

\newtheorem{prop}[thm]{Proposition}
\renewcommand{\proofname}{Proof sketch}

\newcommand{\C}{\mathbb C}
\newcommand{\cD}{\mathcal D}

\newcommand{\cN}{\mathcal N}
\newcommand{\cS}{\mathcal S}
\newcommand{\cU}{\mathcal U}

\newcommand{\cE}{\mathcal{E}}
\newcommand{\cT}{\mathcal{T}}

\newcommand{\Path}[1]{\mathsf{Path}(#1)}

\newcommand{\mfS}{\mathfrak{S}}

\newcommand{\bfI}{\mathbf{I}}
\newcommand{\bfO}{\mathbf{O}}

\usepackage{physics}
\usepackage{mathtools}

\begin{document}
\title{Asymptotically optimal purification of noisy unitary channels in any dimension}

\author{Ryotaro Niwa}
\email{ryotaro.niwa@phys.s.u-tokyo.ac.jp}
\affiliation{Department of Physics, Graduate School of Science, The University of Tokyo, Hongo 7-3-1, Bunkyo-ku, Tokyo 113-0033 Japan}

\author{Satoshi Yoshida}
\email{satoshiyoshida.phys@gmail.com}
\affiliation{Department of Physics, Graduate School of Science, The University of Tokyo, Hongo 7-3-1, Bunkyo-ku, Tokyo 113-0033 Japan}

\author{Mio Murao}
\email{murao@phys.s.u-tokyo.ac.jp}
\affiliation{Department of Physics, Graduate School of Science, The University of Tokyo, Hongo 7-3-1, Bunkyo-ku, Tokyo 113-0033 Japan}

\date{\today}

\begin{abstract}
We consider the problem of \textit{noisy unitary purification}. Given access to an unknown $d$-dimensional unitary channel followed by depolarizing noise of strength $p$, we aim to construct a superchannel that universally purifies the noisy unitary back to the original unknown unitary. We optimize over arbitrary adaptive sequential strategies and analytically derive the optimal fidelity to the leading order in the noise strength and number of channel uses, while also providing a concrete $SU(d)$-covariant parallel strategy that attains the optimum. Our result implies the query complexity $\Theta(d^2p/\epsilon)$ for achieving leading-order infidelity $\epsilon$ in the low-noise regime, which scales better than the naive approach combining optimal state purification and storage-and-retrieval of quantum channels. We also consider the dual problem of \textit{noisy unitary conjugation}, where the goal is to obtain the best approximation of the complex conjugate of the original unknown unitary from access to noisy queries. We show that the optimal fidelity for this task coincides with that of noisy unitary purification to the leading-order in the low-noise and large-query limit. 
\end{abstract}

\maketitle

\textit{Introduction.---}
Unitary channels are an important class of quantum channels that describe the dynamics of closed quantum systems. They serve as a fundamental building block of quantum computation: A carefully designed sequence of unitary gates applied to an initial quantum state, followed by subsequent measurement allows one to efficiently obtain solutions to certain computational problems such as factorization~\cite{Shor_1997} and search~\cite{grover1996fastquantummechanicalalgorithm}. In practice, however, quantum systems are sensitive to noise from the environment. In order to perform reliable quantum computation, it is thus crucial to develop methods that recover the original unitary by removing the effect of noise. 

One way to achieve this is fault-tolerant quantum computation (FTQC)~\cite{PhysRevA.57.127, gottesman2014faulttolerantquantumcomputationconstant} with quantum error-correcting codes (QECC)~\cite{Gottesman1997,PhysRevA.54.1098, PhysRevA.54.4741, Kitaev_2003, Dennis_2002}. In FTQC, one first encodes the quantum state into a QECC, applies unitary gates that are potentially faulty, and finally performs decoding to recover  the intended logical output. However, FTQC assumes that the unitary being applied is specified and \textit{known}, which may not always be true in practice. Moreover, to implement unitary gates fault-tolerantly and \textit{universally}, one often needs sophisticated methods such as magic state distillation~\cite{PhysRevA.71.022316, Wills:2024wid, golowich2024asymptotically, golowich2024quantum, nguyen2024good, itogawa2025efficient} or magic state cultivation~\cite{gidney2024magicstatecultivationgrowing}, which necessitates various optimization efforts. 

A complementary alternative approach is \textit{state purification}~\cite{PhysRevLett.82.4344, Childs_2025, grier2025streamingquantumstatepurification,PhysRevX.9.031013, brahmachari2025optimalqubitpurificationunitary, li2025optimalquantumpurityamplification, li2026quantumpurityamplificationarbitrary, scharnhorst2026nonasymptoticboundsquantumpurity}, where one aims to obtain the best possible approximation to the original unknown pure state from access to multiple copies of noise-corrupted states. Compared to FTQC, where the goal is to reduce noise of a known operation using potentially noisy operations, state purification uses noiseless operations to reduce the noise of unknown input states. Notably, recent works have derived the optimal fidelity for this task~\cite{li2026quantumpurityamplificationarbitrary, li2025optimalquantumpurityamplification, scharnhorst2026nonasymptoticboundsquantumpurity} together with efficient quantum circuits that implement the optimal or near-optimal protocols~\cite{li2025optimalquantumpurityamplification, PhysRevX.9.031013, grier2025streamingquantumstatepurification, yang2024quantum}. However, state purification is not sufficient if we want a purified unitary channel that can later be applied to any quantum state at will. Such a noise reduction, for example, is potentially useful for learning unitary channels in a noisy setting, where FTQC cannot reduce the noise in an \emph{unknown} unitary channel to be learned~\cite{cotler2026noisy, kannan2026exponential}. Although storage-and-retrieval may offer a protocol for retrieving the unitary channel stored in a purified quantum state, it requires an exponential cost for the query complexity in terms of the number of qubits~\cite{ishizaka2008asymptotic, ishizaka2009quantum, bisio2010optimal, yang2020optimal, yoshida2026one} and may not be optimal among all possible strategies.

In this Letter, we consider the problem of \textit{noisy unitary purification}: Given $n$ uses of unknown $d$-dimensional unitary channels under depolarizing noise of strength $p$, we aim to construct a protocol that produces a single use of the best possible approximation to the original unitary. Unlike state purification, we have the freedom to apply quantum operations in various temporal orders, which gives rise to richer strategies such as adaptive protocols with feedback from past measurement outcomes. Among extensive degrees of freedom, we derive the optimal fidelity for this task to the leading order in the noise strength and the number of channel uses, while also finding a concrete $\mathrm{SU}(d)$-covariant parallel strategy that attains the optimum. Our result implies the query complexity $\Theta(d^2p/\epsilon)$ for achieving leading-order infidelity $\epsilon$ in the low-noise regime, which has a better $\epsilon$-dependence than the naive approach combining optimal state purification with storage-and-retrieval of quantum channels. Prior to this work, the optimal fidelity for the minimal $d=2, n=3$ case for the restricted family of \textit{parallel} protocols~\cite{zhao2026distillingunitaryoperationsnogo}, as well as a bound for the leading-order term for the $d=2$ case ~\cite{niwa2026scalingoptimalpurificationnoisyqubit} were known. This work improves upon these works and determines the asymptotically optimal fidelity for the arbitrary $d$-dimensional case under general sequential protocols. Note that our work is distinct from virtual channel purification~\cite{PRXQuantum.6.020325}, where the goal is to obtain expectation values of observables through classical post-processing, rather than to obtain the purified channel itself. 

From a broader perspective, our work initiates the study of the robustness of previously known black-box unitary transformation protocols~\cite{PhysRevResearch.1.013007, Ebler_2023, Grinko_2024, grinko2026sequentialquantumprocessesgroup, Yoshida_2023, Chen_2026} under 
noise. Indeed, noisy unitary purification corresponds to the simplest case of analyzing the robustness of the \textit{identity transformation} under depolarizing noise. Following this viewpoint, we also consider the task of \textit{noisy unitary conjugation}, where the goal is to obtain the best possible approximation to a single use of the conjugate unitary from access to noisy unitaries. Surprisingly, we find that the cost for achieving leading-order infidelity $\epsilon$ in the low-noise and large-query regime is the same for noisy unitary purification and noisy unitary conjugation, which can be contrasted with the noiseless case, where a single use of $\mathcal{U}$ obviously suffices to implement $\mathcal{U}$, while $d-1$ uses of $\mathcal{U}$ are necessary and sufficient to realize a single use of $\mathcal{U}^*$~\cite{PhysRevResearch.1.013007, Ebler_2023, quintino2019probabilistic}. 

\textit{Preliminaries.---}
Let $\mathcal{L}(\mathcal{H})$ denote the set of linear operators on a Hilbert space $\mathcal{H}$. Quantum superchannels~\cite{Chiribella_2008, 8678741} are linear maps that transform quantum channels to quantum channels. Mathematically, an $n$-slot superchannel $\Xi$ is a linear map
\begin{align}
    \Xi: \bigotimes_{i=1}^{n} [\mathcal{L}(\mathcal{I}_i) \to \mathcal{L}(\mathcal{O}_{i})] \to [\mathcal{L}(\mathcal{P}) \to \mathcal{L}(\mathcal{F})] 
\end{align}
such that for any set of quantum channels $\Phi_i: \mathcal{L}(\mathcal{I}_i \otimes \mathcal{A}_i) \to \mathcal{L}(\mathcal{O}_{i} \otimes \mathcal{A}_i)$ and auxiliary systems $\mathcal{A}_i$ for $i \in [n]$, the output $\Phi_{\textrm{out}}=(\Xi\otimes \mathrm{id}_{\mathcal{A}_1\ldots \mathcal{A}_n})(\Phi_1\otimes \Phi_2\otimes \cdots \otimes\Phi_n)$ is a quantum channel. We use $\mathrm{id}$ to denote the identity channel. In this work, we are concerned with the case where input channels $\Phi_1, \ldots, \Phi_n$ are identical. Then, the most general quantum superchannel that can be implemented using a quantum circuit has the form 
\begin{align}\label{eq:Seqdecomp}
    \Phi_{\textrm{out}}= \Lambda_{n} \circ (\Phi_n \otimes \mathrm{id}_{\mathcal{A}_n}) \circ \Lambda_{n-1} \circ \cdots \circ (\Phi_{1} \otimes \mathrm{id}_{\mathcal{A}_1}) \circ \Lambda_0 
\end{align}
using auxiliary Hilbert spaces $\mathcal{A}_1, \cdots, \mathcal{A}_n$ and channels $\Lambda_i: \mathcal{L}(\mathcal{O}_{i}\otimes \mathcal{A}_{i}) \to \mathcal{L}(\mathcal{I}_{i+1}\otimes \mathcal{A}_{i+1}),\, (i= 0, 1, \cdots n)$ and $\mathcal{O}_0 := \mathcal{P}$, $\mathcal{I}_{n+1} := \mathcal{F}$, $\mathcal{A}_0 = \mathcal{A}_{n+1} = \mathbb{C}$. Such a superchannel $\Xi$ is called a \textit{sequential superchannel}, a sequential strategy, or a quantum comb~\cite{chiribella2008quantum}. A restricted class of sequential superchannels is the \textit{parallel superchannel}, (or parallel strategy), which can be decomposed as 
\begin{align}
    \Phi_{\textrm{out}} = \mathcal{D} \circ \qty(\bigotimes_{i=1}^n \Phi_i \otimes \mathrm{id}_\mathcal{A}) \circ \mathcal{E}, 
\end{align}
using an encoding channel $\mathcal{E}: \mathcal{L}(\mathcal{P})\to \mathcal{L}(\mathbf{I} \otimes \mathcal{A})$ with $\mathbf{I} = \bigotimes_i \mathcal{I}_{i}$ and the decoding channel $\mathcal{D}:  \mathcal{L}(\mathbf{O} \otimes \mathcal{A}) \to \mathcal{L}(\mathcal{F})$ with $\mathbf{O} = \bigotimes_i \mathcal{O}_{i}$. Analogous to quantum channels, one can define the \textit{Choi matrix} of quantum superchannels~\cite{Chiribella_2008}, which is a positive semidefinite matrix satisfying certain linear constraints (see Appendix for details). 

\textit{Setup.---}
\begin{figure}[t]
  \centering
  \includegraphics[width=\linewidth]{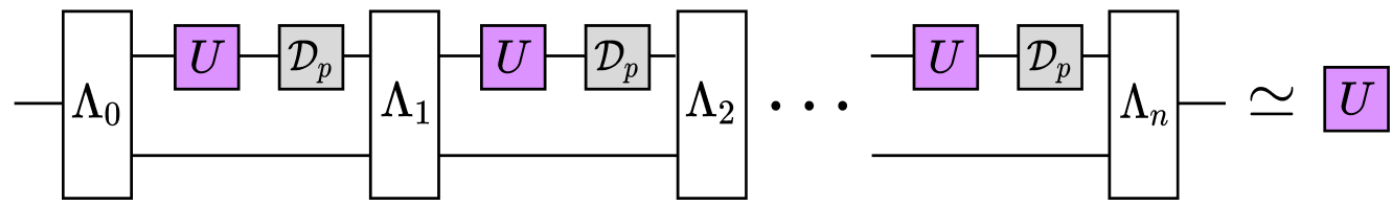}
  \caption{A quantum superchannel $\Xi$ implementing noisy unitary purification. Quantum channels $\Lambda_i,\, (i=0,1,\cdots n)$ are inserted between the noisy unitary channel $\mathcal{N}_{\mathcal{U}, p}$ to produce an approximation of a single use of $\mathcal{U}$. }
\end{figure}
We first describe the formulation of the noisy unitary purification problem~\cite{zhao2026distillingunitaryoperationsnogo, niwa2026scalingoptimalpurificationnoisyqubit}. Given access to a noisy unitary channel 
\begin{align}
    \mathcal{N}_{\mathcal{U},p} = \mathcal{D}_p\circ \mathcal{U}, 
\end{align}
where $\mathcal{D}_p(\rho) = (1-p)\rho+p\Tr(\rho)\frac{\mathbb{I}_d
}{d}$ is the depolarizing channel of strength $p$ ($\mathbb{I}_d$ denotes the identity operator) and $\mathcal{U}(\cdot) = U (\cdot) U^\dagger$ is an \textit{unknown} unitary channel, the goal is to construct a superchannel $\Xi$ such that 
\begin{align}
\Xi(\mathcal{N}_{\mathcal{U},p}^{\otimes n}) \simeq \mathcal{U}
\end{align}
for all $U\in \mathrm{SU}(d)$.
To quantify the performance of $\Xi$, we use the channel fidelity $f_{\mathrm{Choi}}$ as a figure of merit. It is given by~\cite{raginsky2001fidelity}
\begin{align}
    f_{\mathrm{Choi}}(\Phi_1, \Phi_2) := f\qty(\frac{\mathcal{J}_{\Phi_1}}{d}, \frac{\mathcal{J}_{\Phi_2}}{d}), 
\end{align}
where $f(\rho, \sigma):=( \Tr\sqrt{\rho^{1/2}\sigma \rho^{1/2}})^2$ is the squared fidelity and 
$\mathcal{J}_\Phi:= d(\mathrm{id} \otimes \Phi )|\Phi_d^+\rangle \langle \Phi_d^+|$ is the Choi matrix of the quantum channel $\Phi$, with $|\Phi_d^+\rangle = \frac{1}{\sqrt{d}}\sum_i |i\rangle \otimes |i\rangle $ denoting the $d$-dimensional maximally entangled state. Define the performance operator~\cite{chiribella2016optimal} by
\begin{align}
    \Omega_{n,p} = \frac{1}{d^2} \int \mathrm{d} U\, \mathcal{J}_{\mathcal{U}^*} \otimes (\mathcal{J}_{\mathcal{D}_p \circ \mathcal{U}})^{\otimes n}, 
\end{align}
where $(\cdot)^*$ represents the complex conjugate in the computational basis and $\mathrm{d}U$ is the Haar measure. The average channel fidelity between the output channel $\mathcal{Q}_{\mathcal{U}, p}^{(n)}:=\Xi(\mathcal{N}_{\mathcal{U},p}^{\otimes n})$ and the unitary channel $\mathcal{U}$ is given by
\begin{align}
    \int \mathrm{d}U\, f_{\mathrm{Choi}}(\mathcal{Q}_{\mathcal{U}, p}^{(n)}, \mathcal{U}) = \Tr(\Omega_{n,p}^T \mathcal{J}_\Xi), 
\end{align}
where $\mathcal{J}_\Xi$ is the Choi matrix of the superchannel $\Xi$ and $(\cdot)^T$ denotes the transpose in the computational basis. The noisy unitary purification problem can then be formulated as a semidefinite program (SDP): 
\begin{align}\label{eq:SDP}
    &\textrm{max} \Tr(\Omega_{n,p}^T \mathcal{J}_\Xi)\nonumber\\
    &\textrm{subject to ``$\Xi$ is a quantum comb"}. 
\end{align}
We denote the optimal value for the above SDP by $f_{n,d}(p)$. The goal of this paper is to determine the behavior of $f_{n,d}(p)$ in the small-$p$ and large-$n$ limit. Note that, since the performance operator $\Omega_{n,p}$ satisfies 
\begin{align}
    [A^*_P \otimes A^{\otimes n}_{\mathbf{I}} \otimes B_{F}^* \otimes B_{\mathbf{O}}^{\otimes n}, \Omega_{n,p}] = 0\quad (\forall A,B\in \mathrm{SU}(d)), 
\end{align}
one can twirl the Choi matrix $\mathcal{J}_\Xi$ and assume without loss of generality that $\mathcal{J}_\Xi$ also satisfies the same symmetry as $\Omega_{n,p}$. If we use this twirled protocol, the average fidelity is equal to the optimal fidelity for every $U$, and hence the protocol is universal~\cite{quintino2022deterministic}.

\begin{figure*}[t]
  \centering
  \includegraphics[width=0.75\linewidth]{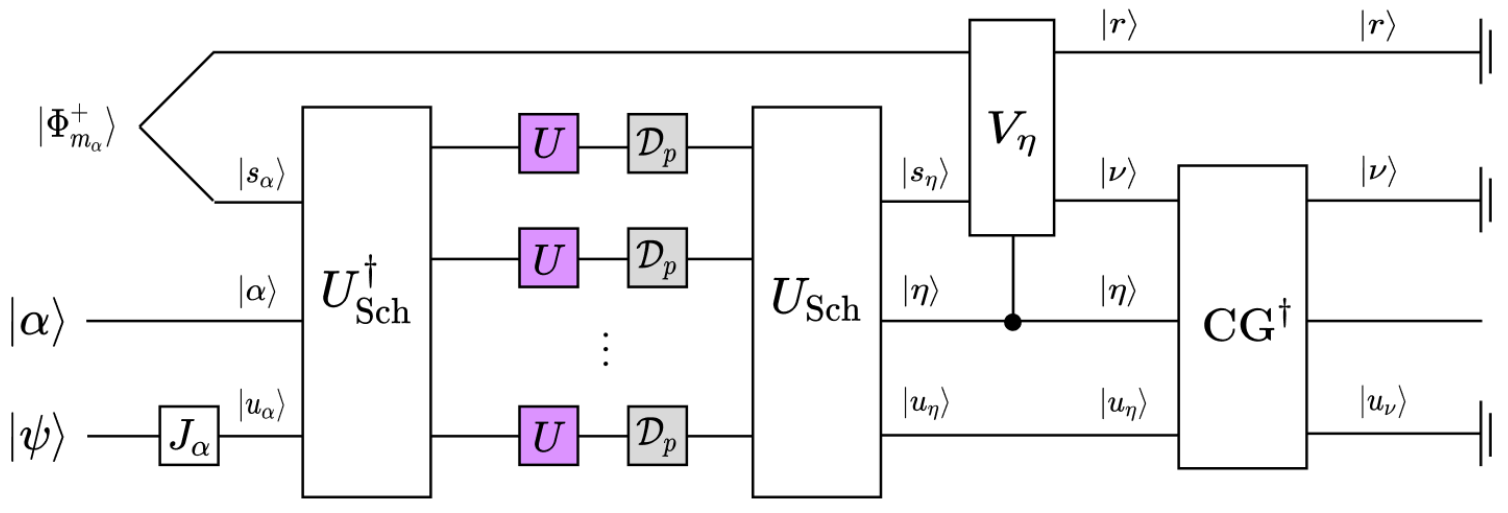}
  \caption{In the above figure, $J_\alpha$ is the basis transformation implementing the isomorphism $\mathbb{C}^d \simeq \mathcal{U}_\alpha$, $U_{\mathrm{Sch}}$ denotes the quantum Schur transform~\cite{bacon2005quantumschurtransformi, Kirby_2018, burchardt2025highdimensionalquantumschurtransforms}, $|\Phi_{m_\alpha}^+\rangle := \frac{1}{\sqrt{m_\alpha}}\sum_{s_\alpha=1}^{m_\alpha} |s_\alpha\rangle_{\mathcal{S}_\alpha} \otimes |s_\alpha \to \hat{\lambda}\rangle_E$ is a maximally entangled state between the Specht module $\mathcal{S}_\alpha$ and the environment, $V_\eta$ is an isometry, and $\mathrm{CG}$ denotes the Clebsch-Gordan transformation. See Appendix for details.}
  \label{fig:purificationcircuit}
\end{figure*}
\textit{Asymptotically optimal noisy unitary purification.---}
To rigorously handle the two limits $p\to 0$ and $n\to \infty$, we first define the first-order infidelity coefficient. 
\begin{dfn} The first-order infidelity coefficient $C_{n,d}$ for the $d$-dimensional noisy unitary purification problem is defined by 
\begin{align}
    C_{n,d} := \lim_{p\to +0}\frac{1-f_{n,d}(p)}{p}. 
\end{align}
\end{dfn}
\noindent One can show that the above limit indeed exists by using Danskin's theorem~\cite{danskin1966, bertsekas2003convex} in convex optimization (see Appendix for details). We now state the first part of our main result: 
\begin{thm}[Fundamental limits on the achievable fidelity]\label{thm1}
The first-order infidelity coefficient $C_{n,d}$ for any sequential protocol satisfies 
\begin{align}\label{eq:upperbound}
    C_{n,d} \geq C^{\star}_{n,d}, 
\end{align}
where 
\begin{align}\label{eq:Cstar}
C_{n,d}^{\star} =  
\frac{1}{d^2n}\qty[1+\frac{(d^2-2)(d^2 +1+2\sqrt{1+\frac{d^2 -1}{n}})}{\qty(1+ \sqrt{1+\frac{d^2 -1}{n}})^2}]
\end{align}
for all $n$.
\end{thm}
\begin{proof}
Let $\Xi_{\mathrm{Pur}}$ denote the Stinespring dilation of the twirled superchannel $\Xi$, which is obtained by dilating each channel $\Lambda_i$ interleaving the noisy unitary channel $\mathcal{N}_{\mathcal{U},p}$, and let $V_{\mathrm{Pur}}(U_1, \cdots U_n)$ denote the associated isometry $\Xi_{\mathrm{Pur}}(\mathcal{U}_1, \cdots \mathcal{U}_n) = V_{\mathrm{Pur}}(\cdot)V_{\mathrm{Pur}}^\dagger$. Let the environment of $\Xi_{\mathrm{Pur}}$ be accessible to a new recovery operation $\mathcal{R}$. Note that choosing $\mathcal{R}$ to be the traceout channel on the environment allows one to recover the original superchannel $\Xi$. Thus, this assumption is sufficient for proving an upper bound on the achievable fidelity. 

Using the Bény-Oreshkov condition~\cite{PhysRevLett.104.120501} together with the covariance of the twirled superchannel $\Xi$, one gets 
\begin{align}\label{eq:BenyOreshkov}
\Tr(\Omega_{n,p}^T \mathcal{J}_\Xi) &\leq 
\max_\mathcal{R} f_{\mathrm{Choi}}(\mathcal{R} \circ \Xi_{\mathrm{Pur}}(\mathcal{N}_{\mathcal{U},p}
^{\otimes n}), \mathcal{U}) \nonumber\\&= \max_\rho f_{\mathrm{Choi}}([\Xi_{\mathrm{Pur}}(\mathcal{N}_{\mathrm{id},p}
^{\otimes n})]^c, \mathcal{T}_\rho).  
\end{align}
Here, $\mathcal{N}^c(\cdot):=\Tr_{\mathcal{O}}[V(\cdot)V^\dagger]$ is the complementary channel of $\mathcal{N}$, $V:\mathcal{I} \to \mathcal{O} \otimes \mathcal{E}$ is the Stinespring dilation of $\mathcal{N}$ with $\mathcal{N}(\cdot):=\Tr_{\mathcal{E}}[V(\cdot)V^\dagger]$, and $\mathcal{T}_\rho(\cdot):= \Tr(\cdot)\rho$ is a constant (trace-and-replace) channel. Note that  $\rho$ is maximized over quantum states, while 
$\mathcal{R}$ is maximized over quantum channels.

Now, the depolarizing channel can be written as 
\begin{align}
    \mathcal{D}_p(\rho) = \qty(1-\frac{d^2-1}{d^2}p)\rho + \frac{p}{d^2}\sum_a T_a \rho T_a^\dagger, 
\end{align}
where $\{T_a\}_{a=1}^{d^2-1}$ represent Hermitian traceless operators satisfying 
\begin{align}
    &\Tr(T_aT_b) = d\delta_{ab}, \\
    &T_a T_b = \delta_{ab}\mathbb{I}_d + (d_{abc} + if_{abc}) T_c. 
\end{align}
Here, $\delta_{ab}$ is the Kronecker delta defined by $\delta_{aa} = 1$ and $\delta_{ab}=0$ for $a\neq b$, and $d_{abc}, f_{abc}$ are totally symmetric and antisymmetric tensors, respectively, defined by 
\begin{align}
\begin{dcases}
    d_{abc} = \frac{1}{2d} \Tr(\{T_a, T_b \}T_c)\\
    f_{abc} = \frac{1}{2id} \Tr([T_a, T_b] T_c).
\end{dcases}
\end{align}
To derive the leading-order term of the fidelity in $p$, it suffices to consider the no-error and single-error terms. Define $V_0:=V_{\mathrm{Pur}}(\mathbb{I}, \mathbb{I}, \cdots, \mathbb{I})$, and $V_{a,r}:=-i\partial_\theta V_{\mathrm{Pur}}(\mathbb{I},\cdots, e^{iT_a \theta}, \cdots \mathbb{I})|_{\theta=0}$. Here, $V_{a,r}$ is effectively obtained by substituting $T_a$ into the $r$-th slot of $V_0$. The covariance of $\Xi_{\mathrm{Pur}}$ forces the relations
\begin{align}\label{eq:tensorrule}
\begin{dcases}
    U [V_{0}^\dagger V_{a,r}] U^\dagger = \sum_b R_{ba}(U) [V_{0}^\dagger V_{b,r}] \\
    U [V_{a,r}^\dagger V_{b,s}] U^\dagger = \sum_{c,e} R_{ca}(U) R_{eb}(U)[V_{c,r}^\dagger V_{e,s}], 
\end{dcases}
\end{align}
where $R(U)$ is the adjoint representation given by
\begin{align}
UT_aU^\dagger = \sum_b R_{ba}(U) T_b. 
\end{align}
Due to the Wigner-Eckart theorem, the relation in Eq.~\eqref{eq:tensorrule} further implies 
\begin{align}\label{eq:tensorrelation}
\begin{dcases}
    V_{0}^\dagger V_{a,r} = c_r T_a \\
    V_{a,r}^\dagger V_{b,s} = K_{rs}\delta_{ab}\mathbb{I}_d + S_{rs}d_{abc}T_c+iL_{rs}f_{abc}T_c, 
\end{dcases}
\end{align}
where $c_r\in \mathbb{R}$ and $K,S,L$ are Hermitian matrices that satisfy 
\begin{align}\label{eq:constraintKSL}
K_{rr}=1, S_{rr}=L_{rr}=c_r. 
\end{align}
Choosing $U= e^{iT_a\theta}$ and differentiating the relation $V_{\mathrm{Pur}}(U, \cdots U) = U \otimes |\eta_\theta\rangle$ with respect to $\theta$ further gives 
\begin{align}\label{eq:constraintc}
    \sum_r c_r = 1. 
\end{align}
Finally, we introduce a symmetry-enhanced optimization problem, the solution to which gives an upper bound on Eq.~\eqref{eq:BenyOreshkov}. We solve this problem under the constraints in Eq.~\eqref{eq:tensorrelation}, ~\eqref{eq:constraintKSL}, ~\eqref{eq:constraintc}, utilizing a key decomposition of the tensor product representation $\mathrm{Adj} \otimes \bar{\Box}$. The details of the proof can be found in the Appendix. 
\end{proof}
We next prove the complementary lower bound for the optimal fidelity by constructing a concrete $\mathrm{SU}(d)$-covariant parallel protocol.  
\begin{thm}[Achievability with parallel strategies]\label{thm2}
The first-order infidelity coefficient $C_{n,d}$ satisfies 
\begin{align}
C_{n,d} \leq C_{n,d}^{\star}
\end{align}
for $n=kd+1, \, k\in \mathbb{N}$, where $C_{n,d}^\star$ is defined in Eq.~\eqref{eq:Cstar}. The upper bound is saturated by an $\mathrm{SU}(d)$-covariant parallel strategy. 
\end{thm}
\begin{proof}
Consider the following representation of the unitary group $U(d)$ and the symmetric group $\mathfrak{S}_n$,  
\begin{align}
    U^{\otimes n}|i_1\rangle \otimes  \cdots \otimes |i_n\rangle &= U|i_1\rangle \otimes \cdots \otimes U|i_n\rangle\\
    \pi(\sigma) |i_1\rangle \otimes  \cdots \otimes |i_n\rangle &= |i_{\sigma^{-1}(1)} \rangle \otimes \cdots \otimes |i_{\sigma^{-1}(n)} \rangle, 
\end{align}
for $U \in U(d)$ and $\sigma \in \mathfrak{S}_n$. 
The Schur-Weyl duality allows simultaneous decomposition of these representations as 
\begin{align}
    (\mathbb{C}^d)^{\otimes n} &\simeq \bigoplus_{\lambda \vdash_d n} \mathcal{U}_\lambda \otimes \mathcal{S}_\lambda,\nonumber\\
    U^{\otimes n} &\simeq \bigoplus_{\lambda \vdash_d n} f_\lambda(U)\otimes \mathbb{I}_{\mathcal{S}_\lambda},\nonumber\\
    \pi(\sigma) &\simeq \bigoplus_{\lambda \vdash_d n} \mathbb{I}_{\mathcal{U}_\lambda} \otimes g_\lambda(\sigma), 
\end{align}
where $\lambda \vdash_d n$ denotes a partition $(\lambda_1, \lambda_2, \cdots \lambda_d)$ satisfying $\lambda_1 + \lambda_ 2 + \cdots \lambda_d = n$, $\lambda_1 \geq \lambda_2 \geq  \cdots \geq \lambda_d \geq 0$, and $f_\lambda: U(d) \to \mathcal{U}_\lambda$ and $g_\lambda: \mathfrak{S}_n \to \mathcal{S}_\lambda$ represent the irreducible  representations of the unitary group $\mathrm{U}(d)$ and the symmetric group $\mathfrak{S}_n$ labeled by the Young diagram $\lambda$, respectively. Define $d_\lambda := \dim \mathcal{U}_\lambda, \, m_\lambda := \dim \mathcal{S}_\lambda $. 

For $n=kd+1$, one can choose the balanced partition $\alpha = (k+1, k, \cdots k)$. Then, $\dim \mathcal{U}_{\alpha} = d$ and
\begin{align}
    U_{\alpha}(U) = (\mathrm{det} U)^k U. 
\end{align}
Define the encoder 
\begin{align}
    E_{n,d}|\psi\rangle = \frac{1}{\sqrt{m_{\alpha}}}\sum_{j=1}^{m_{\alpha}}U_{\mathrm{Sch}}^\dagger(|j\rangle_{\mathcal{S}_{\alpha}} \otimes |\alpha\rangle \otimes |\psi\rangle_{\mathcal{U}_{\alpha}}) \otimes |j\rangle_E. 
\end{align}
This is an $\mathrm{SU}(d)$-covariant encoding: 
\begin{align}
    (U^{\otimes n} \otimes \mathbb{I}_E) E_{n,d} = E_{n,d} U, \quad \forall U \in \mathrm{SU}(d). 
\end{align}
With a suitable choice of the decoding operation in FIG.~\ref{fig:purificationcircuit}, one can show the desired upper bound on the first-order infidelity coefficient $C_{n,d}$. Note that the circuit is independent of $p$, which implies that we do not have to know the noise strength in advance. The details of the proof can be found in Appendix. 
\end{proof}
Theorems 1 and 2 tell us the exact value of $C_{n,d}$ for $n=kd+1$. For $n\in \mathbb{N}$ with $n\not\equiv 1 \mod d$, one can discard queries and choose $m=\lfloor \frac{n-1}{d}\rfloor d + 1$ to apply the $\mathrm{SU}(d)$-covariant parallel strategy in Theorem 2. Then, the same infidelity coefficient is achieved to the leading order in $n$, and the optimal fidelity $f_{n,d}(p)$ asymptotically satisfies 
\begin{align}\label{eq:limnC}
    \lim_{n\to \infty} nC_{n,d}  &= \lim_{n\to \infty} n\qty[\lim_{p\to +0}\frac{1-f_{n,d}(p)}{p}] \nonumber\\
    &= \frac{(d^2-1)(d^2+2)}{4d^2}.
\end{align}
Thus, adaptivity does not help in the asymptotically small-$p$ and large-$n$ regime. Note that substituting $d=2$ into Eq.~\eqref{eq:limnC} gives $\frac{9}{8}$, which matches the fidelity lower bound obtained in Ref.~\cite{niwa2026scalingoptimalpurificationnoisyqubit}. Our work not only determines the exact leading-order coefficient that was left open for $d=2$, but also applies to any dimension $d$. 

\textit{Comparison with storage-and-retrieval.---}
Let us compare our result with the strategy combining optimal state purification~\cite{li2025optimalquantumpurityamplification} with storage-and-retrieval of channels: One uses $n$ copies of noisy Choi states to produce a single purified Choi state, and subsequently uses $m$ copies of the purified Choi state to implement port-based teleportation (PBT)~\cite{Strelchuk_2023}. In this scenario, one needs $\frac{p}{n} + \frac{d^2}{m} = \epsilon$ to achieve infidelity $\epsilon$. Thus it suffices to choose $n = O(p/\epsilon)$ and $m=O(d^2/\epsilon)$ so we obtain the query upper bound $nm = O(d^2p/\epsilon^2)$. On the other hand, our protocol has query complexity $O(d^2p/\epsilon)$ for the leading-order infidelity $\epsilon$ in the low-noise regime, which indicates a better $\epsilon$ dependence.   

\textit{Asymptotically optimal noisy unitary conjugation.---}
Analogous to noisy unitary purification, one can define the problem of noisy unitary conjugation. This time, we are similarly given access to the noisy unitary channel $\mathcal{N}_{\mathcal{U},p} = \mathcal{D}_p\circ \mathcal{U}$, but aim to construct a superchannel $\Xi'$ that outputs the best possible approximation to the complex conjugate $\Xi'(\mathcal{N}_{\mathcal{U},p}^{\otimes n}) \simeq \mathcal{U}^*$. By an analogous argument, one can define the performance operator $\Omega_{n,p}'$ by
\begin{align}
    \Omega'_{n,p} = \frac{1}{d^2} \int \mathrm{d} U\, \mathcal{J}_{\mathcal{U}} \otimes (\mathcal{J}_{\mathcal{D}_p \circ \mathcal{U}})^{\otimes n} 
\end{align}
and formulate the universal noisy unitary conjugation problem as an SDP:
\begin{align}\label{eq:SDP}
    &\textrm{max} \Tr([\Omega_{n,p}']^T \mathcal{J}_\Xi')\nonumber\\
    &\textrm{subject to ``$\Xi'$ is a quantum comb"}. 
\end{align}
We use $g_{n,d}(p)$ to denote the optimal value of the above SDP, and define the leading-order coefficient 
\begin{align}
    C_{n,d}':= \lim_{p \to +0} \frac{1-g_{n,d}(p)}{p}
\end{align}
for $n\geq d-1$. Here, we restrict to $n\geq d-1$ because this ensures that $g_{n,d}(0) =1$~\cite{PhysRevResearch.1.013007}. One can then show the following theorem:

\begin{thm}[Asymptotically optimal fidelity for noisy unitary conjugation]\label{thm4}
The leading-order coefficient $C'_{n,d}$ satisfies 
\begin{align}
C'_{n,d} \geq  \frac{1}{d^2}\qty[n(d^2-1)-\frac{\mathfrak{B}_{n,d}^2}{d^2-1}], 
\end{align}
where 
\begin{align}
\mathfrak{B}_{n,d} = \sqrt{n+2-\frac{d^2-1}{n}} &+\frac{(d-1)(d+2)}{2}\sqrt{n-d+1}\nonumber\\
&+ \frac{(d-2)(d+1)}{2}\sqrt{n+d+1}
\end{align}
for all $n\geq d-1, \, n\in \mathbb{N}$. If $n=kd-1, k\in \mathbb{N}$, there is an $\mathrm{SU}(d)$-covariant parallel strategy that saturates the upper bound. Consequently, the optimal fidelity of the noisy unitary conjugation problem asymptotically satisfies: 
\begin{align}
    \lim_{n\to \infty} nC'_{n,d} &= \lim_{n\to \infty} n\qty[\lim_{p\to +0}\frac{1-g_{n,d}(p)}{p}] \nonumber\\
    &= \frac{(d^2-1)(d^2+2)}{4d^2}. 
\end{align}
\end{thm}
\noindent See appendix for the details of the proof. Theorem 3 implies the query complexity  $\Theta(d^2p/\epsilon)$ for achieving leading-order infidelity $\epsilon$ in the low-noise regime, with the leading-order constant exactly coinciding with that of noisy unitary purification. This is in contrast to the noiseless case, where a single use of $\mathcal{U}$ suffices to achieve $\mathcal{U}$, while $d-1$ queries to $\mathcal{U}$ suffice to obtain a single copy of $\mathcal{U}^*$. We note that the encoding used in the optimal parallel strategy resembles the noiseless case~\cite{PhysRevResearch.1.013007, yoshida2026optimal}, but uses an auxiliary environment space that can later be fed into the decoding operation. 

\textit{Relation to covariant QECC.---} In FTQC, transversal implementations of logical gates are desirable because they prevent errors from spreading within code blocks. Unfortunately, Eastin-Knill theorem~\cite{Eastin_2009} rules out the possibility of universal and exact implementation of logical transversal gates. Even if one relaxes the goal to \textit{approximate} error-correction, several upper bounds~\cite{Faist_2020, Zhou_2021, Kubica_2021} on the achievable fidelity are known to hold.
There, one typically considers the implementation of $U\in \mathrm{U}(d)$ in a transversal way given by $\rho_1(U) \otimes \cdots \otimes \rho_n(U)$ using unitary representations $\rho_i$ for $i\in \{1, \cdots, n\}$.
On the other hand, we restrict to $\rho_i(U) = U$ (noisy unitary purification) or $\rho_i(U) = U^*$ (noisy unitary conjugation) for all $i$, but allow arbitrary sequential protocols including entanglement assistance.
Therefore, it is also possible to interpret our result as an approximate Eastin-Knill-type theorem that prevents the implementation of certain types of transversal gates in sequential protocols.

\textit{Conclusion.---}
In this work, we considered the problem of purifying unknown unitary channels under depolarizing noise. We formulated the problem as an SDP and derived the optimal fidelity to the leading order in the noise strength for sequential protocols while also providing a noise-independent explicit $\mathrm{SU}(d)$-covariant parallel protocol that saturates the optimum. Thus, adaptivity does not help for channel purification in the low-noise and large-$n$ limit. We also considered the dual problem of noisy unitary complex conjugation under depolarizing noise and showed that the leading order coefficient converges to the same value as noisy unitary purification in the large-$n$ limit. Although we focused on deterministic protocols producing a single purified channel, it would be interesting to consider probabilistic protocols or $m$-input-to-$n$-output transformations as well. 

\textit{AI disclosure.---}
GPT-5.6 Sol Ultra on Codex supplied the proof ideas for the main results. The authors later independently reconstructed the proofs and supplied additional arguments to strengthen the results. The authors take full responsibility for the content. All texts and figures were produced by the authors.

\textit{Acknowledgments.---} 
We thank Koki Ono, Ryuji Takagi, Takeru Utsumi, Yuxiang Yang, and  Zhaoyi Li for collaboration on a related project. We also thank Debbie Leung, Isaac Chuang, and Theerapat Tansuwannont for discussions. This work was supported by MEXT Quantum Leap Flagship Program (MEXTQLEAP) JPMXS0118069605 and JPMXS0120351339; Japan Science and Technology Agency (JST) as part of Adopting Sustainable Partnerships for Innovative Research Ecosystem (ASPIRE), Grant Number JPMJAP25A3; JSTCREST, Grant Number JPMJCR25I5; JSTNEXUS, Grant Number JPMJNX26C9; JSPS KAKENHI Grant No.23K21643 and 26K25550; and IBM Quantum.

\let\oldaddcontentsline\addcontentsline
\renewcommand{\addcontentsline}[3]{}

\bibliography{main}
\let\addcontentsline\oldaddcontentsline

\clearpage
\newgeometry{hmargin=1.2in,vmargin=0.8in}
\widetext
\appendix
\onecolumngrid

\begin{center}
{\large \bf Appendices}
\end{center}

\tableofcontents
\renewcommand{\proofname}{Proof}
\setcounter{thm}{0}
\renewcommand{\thethm}{S.\arabic{thm}}
\renewcommand{\thelm}{S.\arabic{lm}}
\renewcommand{\thecor}{S.\arabic{cor}}
\setcounter{figure}{0}
\renewcommand{\thefigure}{S.\arabic{figure}}

\section{Preliminaries}\label{ap:preliminaries}
\subsection{Choi--Jamio\l{}kowski isomorphism}
The Choi--Jamio\l{}kowski isomorphism~\cite{CHOI1975285, JAMIOLKOWSKI1972275} is a convenient way to represent quantum channels and superchannels. Given a quantum channel $\Phi: \mathcal{L}(\mathcal{I}) \to \mathcal{L}(\mathcal{O})$, its Choi matrix is defined by 
\begin{align}
  J_\Phi:=\sum_{i,j=1}^{d_{\mathcal{I}}} \ketbra{i}{j}_{\mathcal{I}} \otimes \Phi(\ketbra{i}{j})_{\mathcal{O}} \in \mathcal{L}(\mathcal{I} \otimes \mathcal{O}), 
\end{align}
satisfying the condition $J_\Phi \geq 0, \Tr_\mathcal{O}(J_\Phi) = \mathbb{I}_\mathcal{I}$. Under the Choi--Jamio\l{}kowski isomorphism, the composition of quantum channels $\Phi_1: \mathcal{L}(\mathcal{X}) \to \mathcal{L}(\mathcal{Y})$ and $\Phi_2: \mathcal{L}(\mathcal{Y}) \to \mathcal{L}(\mathcal{Z})$ is represented by the link product~\cite{chiribella2008quantum} defined by
\begin{align}
  J_{\Phi_2 \circ \Phi_1} = J_{\Phi_2} \star J_{\Phi_1} \coloneqq \Tr_{\mathcal{Y}}\left[
    (\mathbb{I}_{\mathcal{X}} \otimes J_{\Phi_2}) (J_{\Phi_1}^{\mathsf{T}_{\mathcal{Y}}} \otimes \mathbb{I}_{\mathcal{Z}})\right],
\end{align}
where $\star$ denotes the link product and $(\cdot)^{\mathsf{T}_{\mathcal{Y}}}$ denotes the partial transpose with respect to the subsystem $\mathcal{Y}$. Analogously, the Choi matrix of a sequential superchannel $\Xi$ is defined by
\begin{align}
    \mathcal{J}_\Xi  \coloneqq \mathcal{J}_{\Lambda_{n}} \star \mathcal{J}_{\Lambda_{n-1}} \star \cdots \mathcal{J}_{\Lambda_0}, 
\end{align}
where $\Lambda_i$ represent quantum channels appearing in the decomposition, Eq.~\eqref{eq:Seqdecomp}. The Choi matrix of the output channel $\Phi_{\textrm{out}} = \Xi(\Phi_1, \Phi_2 \cdots \Phi_n)$ is then given by 
\begin{align}
    \mathcal{J}_{\Phi_\textrm{out}} = \mathcal{J}_\Xi \star \qty[\bigotimes_{i=1}^n \mathcal{J}_{\Phi_i}]. 
\end{align}
The Choi matrix of a sequential superchannel $\Xi$ is fully characterized by the conditions~\cite{Chiribella_2008,chiribella2008quantum,chiribella2009theoretical}
\begin{align}\label{eq:seqChoi}
\begin{dcases}
\mathcal{J}_\Xi \geq 0 \\
\Tr_{I_k, O_k, \cdots I_{n+1}} (\mathcal{J}_\Xi) = \Tr_{O_{k-1}, I_k, O_k, \cdots I_{n+1}} (\mathcal{J}_\Xi)  \otimes \frac{\mathbb{I}_{O_{k-1}}}{d_{O_{k-1}}}\\
\Tr(\mathcal{J}_\Xi) = d_Pd_\mathbf{O}. 
\end{dcases}
\end{align}
for all $k\in[n+1]$ where $\mathbf{O}\coloneqq \bigotimes_{i=1}^{n} O_i$ and $d_X$ represents the dimension of a Hilbert space $X$. For parallel superchannels, this condition reduces to: 
\begin{align}\label{eq:parChoi}
\begin{dcases}
\mathcal{J}_\Xi \geq 0 \\
\tr_F(\mathcal{J}_\Xi) = \tr_{\mathbf{O}F}(\mathcal{J}_\Xi) \otimes \frac{\mathbb{I}_\mathbf{O}}{d_{\mathbf{O}}} \\
\tr_{\mathbf{I}\mathbf{O}F}(\mathcal{J}_\Xi) = \tr_{P\mathbf{I}\mathbf{O}F}(\mathcal{J}_\Xi)\otimes \frac{\mathbb{I}_P}{d_{P}} \\
\Tr(\mathcal{J}_\Xi) = d_Pd_{\mathbf{O}}, 
\end{dcases}
\end{align}
where $\mathbf{I}\coloneqq \bigotimes_{i=1}^{n} I_i$.

\subsection{Schur-Weyl duality}
Consider the following representation of the unitary group $U(d)$ and the symmetric group $\mathfrak{S}_n$,  
\begin{align}
    U^{\otimes n}|i_1\rangle \otimes  \cdots \otimes |i_n\rangle &= U|i_1\rangle \otimes \cdots \otimes U|i_n\rangle\\ \label{eq:perm}
    \pi(\sigma) |i_1\rangle \otimes  \cdots \otimes |i_n\rangle &= |i_{\sigma^{-1}(1)} \rangle \otimes \cdots \otimes |i_{\sigma^{-1}(n)} \rangle, 
\end{align}
for $U \in U(d)$ and $\sigma \in \mathfrak{S}_n$. The Schur-Weyl duality~\cite{ceccherini2010representation} is the following decomposition of the Hilbert space: 
\begin{align}
    (\mathbb{C}^{d})^{\otimes n} &\simeq \bigoplus_{\lambda \vdash_d n} \mathcal{U}_\lambda \otimes \mathcal{S}_\lambda \\
    U^{\otimes n} &\simeq \bigoplus_{\lambda \vdash_d n}  f_\lambda(U) \otimes \mathbb{I}_{\mathcal{S}_\lambda}\\
    \pi(\sigma) &\simeq \bigoplus_{\lambda \vdash_d n}  \mathbb{I}_{\mathcal{U}_\lambda} \otimes g_\lambda(\sigma),
\end{align}
where $\lambda \vdash_d n$ denotes a partition $(\lambda_1, \lambda_2, \cdots \lambda_d)$ satisfying
$\lambda_1+ \lambda_2+ \cdots \lambda_d=n$, $\lambda_1 \geq \lambda_2 \geq \cdots \geq \lambda_d\geq 0$, while $f_\lambda: U(d) \to \mathcal{L}(\mathcal{U}_\lambda)$, $g_\lambda: \mathfrak{S}_n \to \mathcal{L}(\mathcal{S}_\lambda)$ denote the irreducible representation of $U(d)$ and $\mathfrak{S}_n$ labeled by the partition $\lambda$, respectively. We use $\ell(\lambda)$ to denote the height of the Young diagram $\lambda$, the number of nonzero entries in the partition. We also define the dimensions $d_\lambda := \dim \mathcal{U}_\lambda, m_\lambda:= \dim \mathcal{S}_\lambda$ and employ the Gelfand-Tsetlin basis $\{|u\rangle\}_{u\in [d_\lambda]}$~\cite{goodman2009symmetry} as an orthonormal basis for $\mathcal{U}_\lambda$ and the Young-Yamanouchi basis $\{|i\rangle\}_{i\in [m_\lambda]}$~\cite{james2006representation} for an orthonormal basis for $\mathcal{S}_\lambda$. We call the basis $\{|\lambda, u, i \rangle\}_{\lambda \vdash_d n}$ the Schur basis. The unitary transformation taking the computational basis to the Schur basis $U_{\textrm{Sch}}^{(n,d)}$ is called the quantum Schur transform. The Young-Yamanouchi basis is associated to the standard Young tableau (SYT). The matrix unit, defined by 
\begin{align}
    E^\lambda_{ij}:= \mathbb{I}_{d_\lambda} \otimes |i\rangle \langle j|, 
\end{align}
enjoys the following convenient properties~\cite{9815215} 
\begin{align}\label{eq:YYformula}
\begin{dcases}
    E^{\alpha}_{ab} \otimes \mathbb{I}_d = \sum_{\mu \in \alpha + \Box} E^{\mu}_{a_\mu b_\mu}\\
    \Tr_{n+1} E_{ij}^\mu  = \delta_{\alpha \beta} \frac{d_\mu}{d_\alpha} E^{\alpha}_{i', j'}
\end{dcases}
\end{align}
where $\alpha+\Box$ is the set of Young diagrams obtained by adding a box to $\alpha$, $\alpha_\mu, b_\mu$ is the SYT obtained by adding the box $n+1$ to the SYT $a,b$, and $i',j'$ are SYT with frame $\alpha, \beta$ that can be obtained from SYT $i,j$ by removing the box $n+1$. In the second equation, the partial trace vanishes if the two tableaux have different shapes after removing $n+1$. 
\subsection{Mixed Schur-Weyl duality}
Mixed Schur-Weyl duality~\cite{koike1989decomposition,benkart1994tensor,grinko2025mixed} is a representation-theoretic machinery for handling the mixed tensor representation $U^{\otimes n}\otimes (U^*)^{\otimes m}$ of the unitary group $\mathrm{U}(d)$. To introduce this tool, we first review the walled-Brauer algebra $\mathcal{B}_{n,m}^\delta$. 

Let $n,m>0$ be integers and let $\delta\in \mathbb{C}$. Walled-Brauer algebra consists of complex combinations of walled-Brauer diagrams, which is a diagram of two rows with $n+m$ nodes separated by a vertical ``wall" that separates $n$ nodes and $m$ nodes. The nodes are paired together under the following restriction: if both nodes are in the same row, they must be on different sides of the wall, while if they are in different rows, they must be on the same side of the wall. In plain language, a walled-Brauer diagram is a SWAP diagram that is partially transposed on the last $m$ nodes (See FIG~\ref{fig:WBdiagram}). To multiply two walled-Brauer diagrams, one simply concatenates the diagrams vertically. Whenever a loop forms under the concatenation, we remove it and multiply the diagram by $\delta$. 
\begin{figure*}[h]
  \centering
  \includegraphics[width=0.2\linewidth]{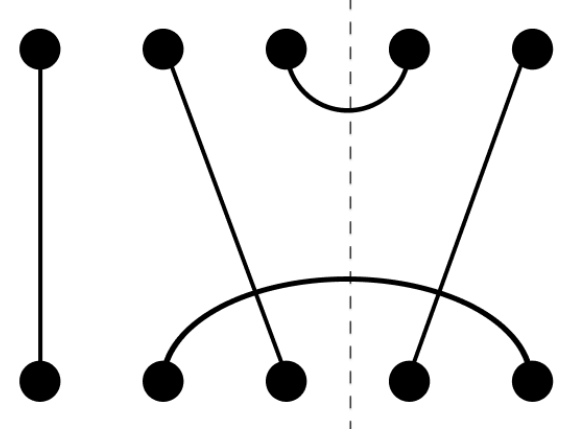}
  \caption{A walled-Brauer diagram $\sigma \in \mathcal{B}_{3,2}^\delta$, which is a ``partially transposed" SWAP diagram.}
  \label{fig:WBdiagram}
\end{figure*}

One can consider the following matrix representation of the walled-Brauer algebra: 
\begin{align}
    \pi_{n,m}^d(\sigma) (|i_1\rangle |i_2\rangle \cdots |i_{n+m}\rangle) = \sum_{j_1, \cdots  j_{n+m}=1}^d \sigma_{j_1, \cdots j_{n+m}}^{i_1, \cdots i_{n+m}} (|j_1\rangle |j_2\rangle \cdots |j_{n+m}\rangle), 
\end{align}
where $\sigma_{j_1, \cdots j_{n+m}}^{i_1, \cdots i_{n+m}} = \prod_{(t,b)\in \sigma} \delta_{i_t, j_b}$. This is a natural generalization of the unitary representation of symmetric group in Eq.~\eqref{eq:perm}. The algebra generated by this representation is called the partially transposed permutation matrix algebra, $\mathcal{A}_{n,m}^d$.

The mixed Schur-Weyl duality is the statement that the mixed tensor representation $U^{\otimes n}\otimes (U^*)^{\otimes m}$ and the partially transposed permutation matrix $\pi_{n,m}^d(\sigma)$ can be simultaneously block-diagonalized in the following way: 
\begin{align}
    (\mathbb{C}^d)^{\otimes n}\otimes (\bar{\mathbb{C}}^d)^{\otimes m} &\simeq \bigoplus_{\gamma:\mathrm{irrep}} \mathcal{U}_\gamma\otimes S_{\gamma},\\
    U^{\otimes n}\otimes (U^*)^{\otimes m} &= \bigoplus_{\gamma:\mathrm{irrep}} f_\gamma(U) \otimes \mathbb{I}_{S_\gamma}, \quad U \in \mathrm{U}(d),\\
    \pi_{n,m}^d(\sigma) &= \bigoplus_{\gamma:\mathrm{irrep}} \mathbb{I}_{\mathcal{U}_\gamma} \otimes g_\gamma(\sigma) ,\quad \sigma \in \mathcal{B}_{n,m}^d.
\end{align}
Here, the irreps $\gamma$ are labeled by a pair of Young diagrams, and $\mathcal{U}_\gamma, \mathcal{S}_\gamma$ are called the Weyl module and the Specht module, respectively. To obtain the concrete form of the irrep labels, one first constructs the \textit{Bratteli diagram} of the walled-Brauer algebra $\mathcal{B}_{n,m}^d$ according to the following rule~\cite{bulgakova2020fusion}:
\begin{itemize}
    \item Start from a pair of null diagrams $(\phi, \phi)$. 
    \item To the left of the wall (the first $n$ levels), add a single box to the left diagram. 
    \item To the right of the wall (the last $m$ levels), remove a single box from the left diagram, or add a single box to the right diagram. 
\end{itemize}
\begin{figure*}[t]
  \centering
  \includegraphics[width=0.75\linewidth]{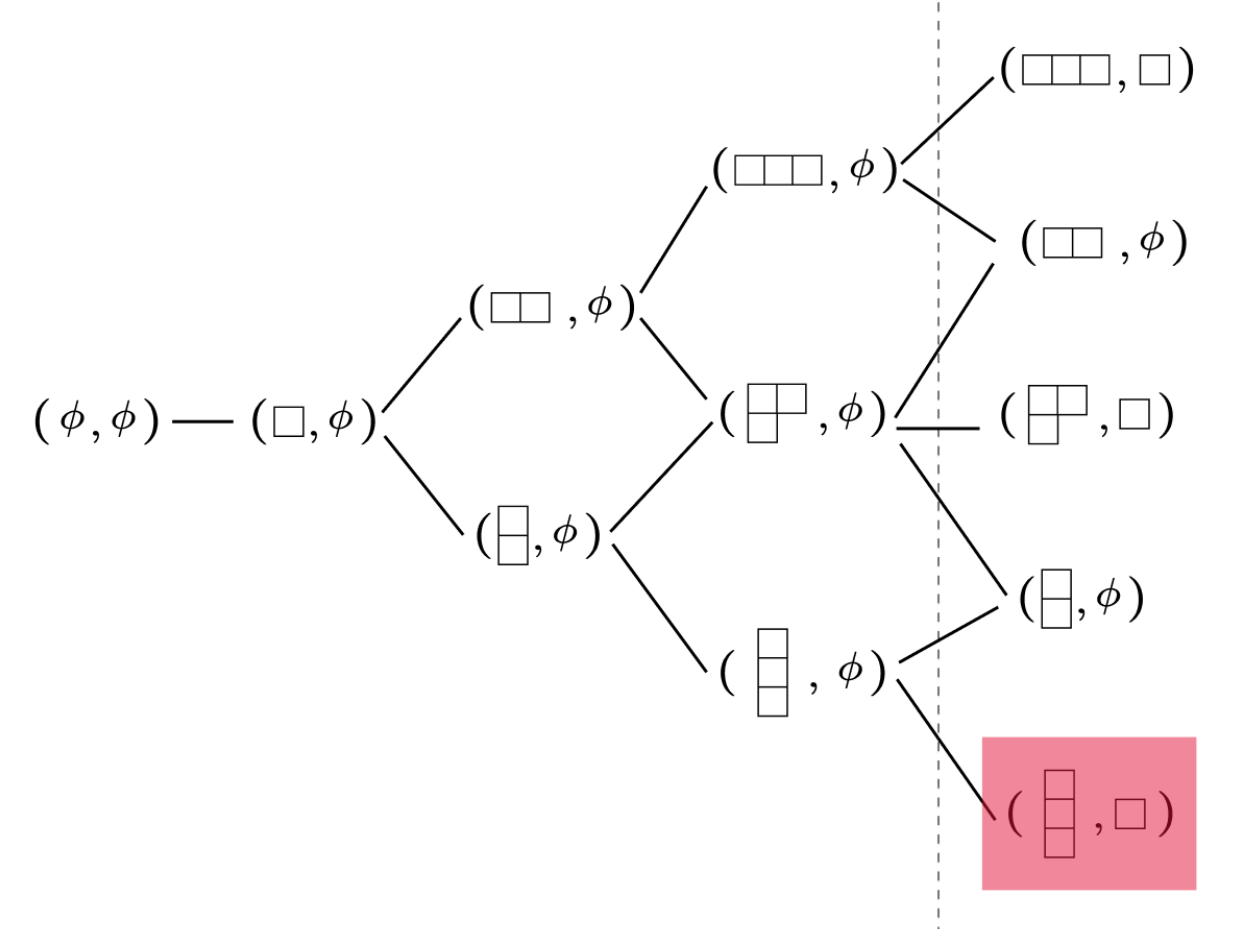}
  \caption{The Bratteli diagram for $\mathcal{A}_{3,1}^d$. The dashed line represents the ``wall". Before crossing the wall, a box is added to the left diagram. After crossing the wall, a box is either removed from the left diagram, or added to the right diagram. For $d=3$, the edge with the red shade is removed because $\ell(\lambda) + \ell(\mu) > d$.}
\end{figure*}
After this procedure, one removes the edges $(\lambda, \mu)$ from the Bratteli diagram if $\ell(\lambda) + \ell(\mu) > d$. This is the Bratteli diagram of the partially transposed permutation matrix algebra $\mathcal{A}_{n,m}^d$, and the resulting vertices on the final level represent the irreps $\gamma$ that appear in the mixed Schur-Weyl duality. We use the notation $\ell(\gamma):= \ell(\lambda) + \ell(\mu)$ for $\gamma=(\lambda, \mu)$. The basis of $S_\gamma$ can be identified with a path in the Bratteli diagram of $\mathcal{A}_{n,m}^d$. This construction includes the Young-Yamanouchi basis as a special case, $m=0$: A standard Young tableau (SYT) corresponds to a path in the Bratteli diagram with no ``walls". 

The basis $\{|u_\gamma\rangle |\gamma\rangle |s_\gamma\rangle\}$ is called the mixed Schur basis, and the unitary transformation taking the computational basis to the mixed Schur basis is called the mixed Schur transform~\cite{grinko2024efficientquantumcircuitsportbased, nguyen2023mixedschurtransformefficient, fei2023efficientquantumalgorithmportbased}, which we denote by $U_{\mathrm{mSch}}^{(n,m)}$. 

\subsection{Clebsch-Gordan transforms}
Let $\rho_\lambda, \rho_\mu$ be irreducible representations of the unitary group $\mathrm{U}(d)$, with $\lambda \vdash_d n, \mu \vdash_d m$, and let $N=n+m$. Although $\rho_\lambda, \rho_\mu$ are both irreducible, its tensor-product representation $\rho_\lambda \otimes \rho_\mu$ is generally reducible: 
\begin{align}\label{eq:CGdecomposition}
    \mathcal{U}_\lambda \otimes \mathcal{U}_\mu \simeq \bigoplus_{\nu \vdash_d{N}} \mathcal{U}_\nu \otimes \mathbb{C}^{c_{\lambda \mu}^\nu}, 
\end{align}
where $c_{\lambda \mu}^\nu$ is the Littlewood-Richardson coefficient. Generally, the Clebsch-Gordan transform is the basis transformation that establishes the above isomorphism. In this paper, we are specifically interested in the case where one of the irreps in Eq.~\eqref{eq:CGdecomposition} is the fundamental representation $\Box$: 
\begin{align}
    \mathcal{U}_\lambda \otimes \mathcal{U}_\Box \simeq \bigoplus_{\substack{\nu \vdash_d{N}\\
    \nu = \lambda + \Box}} \mathcal{U}_\nu. 
\end{align}
In this case, the decomposition is multiplicity-free due to the Pieri rule~\cite{macdonald1995symmetric}. We use $\mathrm{CG}$ to denote the Clebsch-Gordan transform, which is the basis transformation $|u_\lambda\rangle \otimes |u_\Box\rangle \to |\nu\rangle |u_\nu\rangle$~\cite{bacon2005quantumschurtransformi,harrow2005applications}. Similarly, the tensor product of Weyl modules in the mixed Schur-Weyl duality decomposes into irreducible representations. When one of the irreps is the anti-fundamental representation $\bar{\Box}$, 
\begin{align}
    \mathcal{U}_\gamma \otimes \mathcal{U}_{(\phi, \Box)} \simeq \bigoplus_{\substack{\nu \in \gamma-\Box\\
    \ell(\nu) \leq d}} \mathcal{U}_\nu. 
\end{align}
The dual Clebsch-Gordan transform is the basis transformation $|u_\gamma\rangle \otimes |u_{\bar{\Box}}\rangle \to |\nu\rangle |u_\nu\rangle$, which we denote by $\mathrm{dCG}$~\cite{grinko2024efficientquantumcircuitsportbased, nguyen2023mixedschurtransformefficient, fei2023efficientquantumalgorithmportbased}.

\subsection{Quadratic Casimir operator}
Let $G$ be a Lie group and let $\mathfrak{g}$ be its Lie algebra. Let $\rho: G \to \mathrm{GL}(W)$ be a representation of $G$ on a vector space $W$. The representation $\rho$ induces a Lie algebra representation $\mathrm{d}\rho: \mathfrak{g} \to \mathcal{L}(W)$ by differentiation at the identity: 
\begin{align}
    \mathrm{d}\rho(X) =  \frac{1}{i}\left.\dv{\theta} \rho(\mathrm{exp}(i\theta X))\right|_{\theta =0}. 
\end{align}
If $\{X_1, \ldots, X_{\dim \mathfrak{g}}\}$ is the basis of the Lie algebra, $\mathrm{d}\rho(X_1), \cdots \mathrm{d}\rho(X_{\dim \mathfrak{g}})$ are called the infinitesimal generators of the representation on $W$. This preserves the Lie brackets:
\begin{align}
    d\rho(i[X,Y]) = i[\mathrm{d}\rho(X), \mathrm{d}\rho(Y)]. 
\end{align}

\noindent \underline{Example 1: Anti-fundamental representation}\\
The anti-fundamental representation is given by 
\begin{align}
    \rho_{\mathrm{conj}}(U) = U^*. 
\end{align}
Let $\{T_a\}$ be the Hermitian traceless basis for the Lie algebra. Then, 
\begin{align*}
    \rho_{\mathrm{conj}}(U) = e^{-i\theta T_a^T}, 
\end{align*}
so the generators are $-T_a^T$. \\

\noindent \underline{Example 2: Adjoint representation}\\
The adjoint representation acts on traceless matrices by conjugation: 
\begin{align}
    \rho_{\mathrm{Adj}}(U)(X) = UXU^\dagger.  
\end{align}
Differentiating gives 
\begin{align}
     \frac{1}{i}\left.\dv{\theta} (e^{iT_a\theta}Xe^{-iT_a\theta})\right|_{\theta =0} = [T_a, X]. 
\end{align}
Thus, the generator corresponding to $T_a$ is given by $F_a(\cdot) := [T_a, \cdot]$. Thus, 
\begin{align}
    [F_a]_{cb} = \frac{1}{d} \Tr([T_a, T_b] T_c)= 2if_{abc}. 
\end{align}
Now, let $T_a$ be the basis of a Lie algebra and let $G_a$ be the corresponding generators in a representation $W$. The quadratic Casimir operator on $W$ is given by 
\begin{align}
    C = \sum_a G_a^2. 
\end{align}
This operator commutes with the generators: 
\begin{align}
    [C, G_b] = 0.  
\end{align}
Now, consider two representations $\rho_\lambda(U), \rho_\mu(U)$ and the tensor-product representation 
\begin{align}
    \rho_{\lambda \otimes \mu}(U) = \rho_\lambda(U) \otimes \rho_\mu(U). 
\end{align}
Let $G_a^{(\lambda)}, G_a^{(\mu)}$ denote the generators corresponding to the basis $T_a$ in the representation $\rho_\lambda, \rho_\mu$. Then, the generator corresponding to the basis $T_a$ in the representation $\lambda \otimes \mu$ is 
\begin{align}
    G_a^{\mathrm{tot}} = G_a^{(\lambda)} \otimes \mathbb{I} + \mathbb{I}\otimes G_a^{(\mu)}, 
\end{align}
and the total Casimir operator is 
\begin{align}
    C_{\mathrm{tot}} = \sum_a (G_a^{(\lambda)} \otimes \mathbb{I} + \mathbb{I}\otimes G_a^{(\mu)})^2 = C_\lambda \otimes \mathbb{I} + \mathbb{I}\otimes C_\mu + 2\sum_{a} G_a^{(\lambda)} \otimes G_a^{(\mu)}. 
\end{align}
If an irrep is represented by a Young diagram $\lambda=(\lambda_1, \lambda_2, \cdots \lambda_d)$ and $\sum_i \lambda_i = m$, one has~\cite{Piddock:2018xyi} 
\begin{align}\label{eq:Piddockequation}
    C(\lambda) = d^2m + d\sum_{i=1}^{d} \lambda_i(\lambda_i+1-2i)-m^2. 
\end{align}

\section{Existence of the leading-order coefficient}
\noindent Here we justify Definition 1 in the main text by showing that the limit indeed exists. We use an argument that is similar to Danskin's theorem~\cite{danskin1966, bertsekas2003convex} in convex optimization. 
\begin{lm}
The set of the Choi matrices of deterministic sequential superchannels is \textit{compact}. 
\end{lm}
\begin{proof}
Since Choi matrices of superchannels are Hermitian and thus can be considered as a real vector space, Heine-Borel theorem~\cite{rudin1976principles} implies that it suffices to show that the set is closed and bounded. The closedness follows because Choi matrices are positive semidefinite and constrained by partial trace. Boundedness follows because the trace is a constant. 
\end{proof}

\begin{lm}\label{lm:Danskin}
Let $\mathcal{K}_{n,d}$ be the set of Choi matrices of $n$-slot deterministic sequential superchannels on a $d$-dimensional system, and let $\Omega_{n,p}$ denote the performance operator for noisy unitary purification. We denote the fidelity achieved by a strategy $J$ by $h_{n,d}(p,J):=\Tr[\Omega_{n,p}^TJ]$. Then, 
\begin{align}
f_{n,d}(p) = \max_{J\in  \mathcal{K}_{n,d}} h_{n,d}(p,J). 
\end{align}
Let $\mathcal{M}_{n,d}$ be a set 
\begin{align}
    \mathcal{M}_{n,d} := \{ J \in \mathcal{K}_{n,d}| h_{n,d}(0,J) = f_{n,d}(0) = 1\}. 
\end{align}
The coefficient 
\begin{align}
    C_{n,d}:= \lim_{p\to 0} \frac{1-f_{n,d}(p)}{p}
\end{align}
exists, and its value is 
\begin{align}
    C_{n,d} = -\max_{J \in \mathcal{M}_{n,d}} \Tr(\dot{\Omega}_{n,0}^TJ). 
\end{align}
\end{lm}

\begin{proof}
We show the claim 
\begin{align}
    \left.\dv{p} \max_{\mathcal{J} \in \mathcal{K}_{n,d}} \Tr(\Omega_{n,p}^T \mathcal{J})\right|_{p=0} = \max_{\mathcal{J} \in \mathcal{M}_{n,d}} \partial_p h_{n,d}(0, \mathcal{J}). 
\end{align}
First, observe that $\Omega_{n,p}$ is a polynomial of $p$, so we can expand it as 
\begin{align}
    \Omega_{n,p} = \Omega_{n,0}+p\dot{\Omega}_{n,0}+p^2R_{n,d}(p). 
\end{align}
Thus, for the fidelity $h_{n,d}(p,J)$ achieved by a strategy $J$, 
\begin{align}
    h_{n,d}(p,J) = h_{n,d}(0,J) + p\Tr(\dot{\Omega}_{n,0}^TJ) + p^2\Tr(R_{n,d}(p)^TJ). 
\end{align}
For a sufficiently small $p$ in the interval $0\leq p \leq p_0$, one can make $\lVert R_{n,d}(p) \rVert_\infty$ finite. Using Hölder's inequality, we have 
\begin{align*}
    |\Tr(R_{n,d}(p)^TJ)| \leq \Tr|R_{n,d}(p)^TJ| \leq \lVert R_{n,d}(p)\rVert_\infty \lVert J \rVert_1
\end{align*}
so 
\begin{align}\label{eq:huniform}
    h_{n,d}(p,J) = h_{n,d}(0,J) + p\Tr(\dot{\Omega}_{n,0}^T J) + O_{n,d}(p^2). 
\end{align}
\noindent Here, note that $J$ can in principle depend on $p$. Next, let $J_\star \in \mathcal{M}_{n,d}$ maximize $\Tr(\dot{\Omega}_{n,0}^T J)$ over $J \in \mathcal{M}_{n,d}$. From Eq.~\eqref{eq:huniform}
\begin{align}
    h_{n,d}(p, \mathcal{J}_\star) = h_{n,d}(0, \mathcal{J}_\star) + p\Tr(\dot{\Omega}_{n,0}^T J_\star) + O_{n,d}(p^2). 
\end{align}
Since $f_{n,d}(p) \geq h_{n,d}(p,J_\star)$ and $f_{n,d}(0) = h_{n,d}(0,\mathcal{J}_\star)$, we have 
\begin{align}
    \liminf_{p\to 0} \frac{f_{n,d}(p)-f_{n,d}(0)}{p} \geq \liminf_{p\to 0} \frac{h_{n,d}(p, J_\star)-h_{n,d}(0, J_\star)}{p}  = \max_{J \in \mathcal{M}_{n,d}}\Tr(\dot{\Omega}_{n,0}^T J). 
\end{align}
Finally, for each $p>0$, choose the optimizer $J_p$. Also choose any operator $J_0 \in \mathcal{M}_{n,d}$. Then,
\begin{align}
    h_{n,d}(p, J_p) \geq h_{n,d}(p,J_0) \to f_{n,d}(0) \quad (p\to 0). 
\end{align}
Since $h_{n,d}(p,J_p)\to h_{n,d}(0, J_p)$ converges uniformly on $\mathcal{K}_{n,d}$, $h_{n,d}(0, J_p) \to f_{n,d}(0)$. Indeed, using 
\begin{align}
    \epsilon_p := \sup_{J \in \mathcal{K}_{n,d}}|h_{n,d}(p, J)-h_{n,d}(0, J)|, 
\end{align}
uniform convergence gives $h_{n,d}(0, J_p) \geq h_{n,d}(p, J_p) -\epsilon_p$ and $h_{n,d}(p,J_0) \geq h_{n,d}(0,J_0) -\epsilon_p =f_{n,d}(0)-\epsilon_p$. Combined with $h_{n,d}(p, J_p) \geq h_{n,d}(p, J_0)$, this yields $f_{n,d}(0) \geq h_{n,d}(0, J_p) \geq f_{n,d}(0)-2\epsilon_p$. Now, Lemma 1 tells us that $\mathcal{K}_{n,d}$ is compact, so any sequence $p_k\to  0$ has a subsequence such that $J_{p_k} \to\bar{J}$ and continuity gives $h_{n,d}(0, \bar{J}) = f_{n,d}(0)$. Therefore, every limiting optimizer satisfies $\bar{J} \in \mathcal{M}_{n,d}$. Now,
\begin{align}
    \frac{f_{n,d}(p) - f_{n,d}(0)}{p} &= \frac{h_{n,d}(p, J_p)-h_{n,d}(0,J_p)}{p}+\frac{h_{n,d}(0, J_p)-f_{n,d}(0)}{p}\nonumber\\
    & \leq \frac{h_{n,d}(p, J_p)-h_{n,d}(0,J_p)}{p} \nonumber\\
    &= \Tr(\dot{\Omega}_{n,0}^T J_p) + O_{n,d}(p)
\end{align}
so take any sequence $p_k$ such that its subsequence converges to $\bar{J}\in \mathcal{M}_{n,d}$. This yields 
\begin{align}
    \limsup_{p\to 0} \frac{f_{n,d}(p)-f_{n,d}(0)}{p} \leq \Tr(\dot{\Omega}_{n,0}^T \bar{J}) \leq \max_{J \in \mathcal{M}_{n,d}} \Tr(\dot{\Omega}_{n,0}^T J). 
\end{align}
\end{proof}

\begin{cor}\label{cor:noiseless}
The set $\mathcal{M}_{n,d}$ in Lemma~\ref{lm:Danskin} is characterized as the noiseless face $\mathcal{M}_{n,d} = \{\mathcal{J}_\Xi \in \mathcal{K}_{n,d}| \Xi(\mathcal{U}^{\otimes n}) = \mathcal{U}\}$. Thus, it is sufficient to optimize $h_{n,d}(p,J)$ over the noiseless face $\mathcal{M}_{n,d}$, instead of $\mathcal{K}_{n,d}$. 
\end{cor}
\begin{proof}
$f_{n,d}(0) = 1$. Thus, $\mathcal{J}_\Xi \in \mathcal{M}_{n,d} \implies \int \mathrm{d}U f_{\mathrm{Choi}}(\Xi(\mathcal{U}^{\otimes n}), \mathcal{U}) = 1$. This further implies $f_{\mathrm{Choi}}(\Xi(\mathcal{U}^{\otimes n}), \mathcal{U})=1$ for all $\mathcal{U}$, and therefore $\Xi(\mathcal{U}^{\otimes n}) = \mathcal{U}$. The converse, $\Xi(\mathcal{U}^{\otimes n}) = \mathcal{U} \implies J_\Xi \in \mathcal{M}_{n,d}$ is obvious. 
\end{proof}

\begin{cor}\label{cor:noiselessconjugate}
The first-order infidelity coefficient $C'_{n,d}$ for the noisy conjugation problem also exists for $n\geq d-1$. Furthermore, it suffices to optimize over the noiseless face 
$\{\mathcal{J}_\Xi \in \mathcal{K}_{n,d}|\Xi(\mathcal{U}^{\otimes n}) = \mathcal{U}^*\}$. 
\end{cor}
\begin{proof}
In the proof of Lemma~\ref{lm:Danskin} and Corollary~\ref{cor:noiseless}, we have only used the fact that $\Omega_{n,p}$ can be expanded as a polynomial in $p$, and that $f_{n,d}(0) =1$. These features survive in the noisy conjugation problem. That is to say, $\Omega'_{n,p}$ is a polynomial in $p$, and $g_{n,d}(0) =1$ for $n\geq d-1$~\cite{PhysRevResearch.1.013007}. 
\end{proof}

\section{Proof of Theorem 1}
Here we prove our main theorem, the upper bound on the achievable fidelity. The proof can be outlined as follows: We first dilate the superchannel and use covariance (Lemma~\ref{lm:covariance}) to simplify the problem. We then translate the problem to a dual optimization problem using the Bény-Oreshkov condition (Lemma~\ref{lm:truncation}). We further simplify the optimization using tools from representation theory (Lemmas~\ref{lm:rhoERintro}-\ref{lm:rhoER}) and derive a symmetry-enhanced optimization problem that gives an upper bound to the original problem (Lemma~\ref{lm:symmetry-enhancedLP}). Finally, we solve this problem (Lemma~\ref{lm:LPsolution}), which gives the desired upper bound. 
\begin{thm}[Fundamental limits on the achievable fidelity, Theorem 1 in the main text]
\label{thm:main}
Let \(\Delta=d^2-1\).  Every adaptive sequential strategy obeys
\begin{equation}
C_{n,d} \geq C_{n,d}^{\star},
\qquad
C_{n,d}^{\star}:=
\frac{1}{d^2n}\qty[1+\frac{(\Delta-1)(\Delta +2+2\sqrt{1+\frac{\Delta}{n}})}{\qty(1+ \sqrt{1+\frac{\Delta}{n}})^2}]. 
\end{equation}
\end{thm}
\begin{proof}
Choose a Hermitian traceless basis $\{T_a\}$ of $\mathfrak{su}(d)$ satisfying
\begin{align}
    &\Tr(T_aT_b) = d\delta_{ab}\\
    &T_aT_b = \delta_{ab} \mathbb{I}_d + (d_{abc}+if_{abc})T_c,
\end{align}
where $d_{abc}$ and $f_{abc}$ are totally symmetric and totally antisymmetric tensors, defined by 
\begin{align}
\begin{dcases}
    d_{abc} = \frac{1}{2d} \Tr(\{T_a, T_b \}T_c)\\
    f_{abc} = \frac{1}{2id} \Tr([T_a, T_b] T_c), 
\end{dcases}
\end{align}
respectively. The $d$-dimensional depolarizing channel can then be expressed as 
\begin{align}
    \mathcal{D}_p(\rho) = \qty(1-\frac{\Delta}{d^2} p)\rho + \frac{p}{d^2}\sum_{a=1}^\Delta T_a\rho T_a^\dagger. 
\end{align}
The basis $\{T_a\}$ transforms under the \textit{adjoint representation} $U\mapsto R(U)$ given by
\begin{align}
    UT_aU^\dagger = \sum_b R_{ba}(U)T_b, \quad U \in \mathrm{U}(d). 
\end{align}

Now, let $\Xi$ be the sequential superchannel that implements noisy unitary purification under depolarizing noise. We know from the proof of Lemma~\ref{lm:Danskin} that it suffices to consider $\Xi$ satisfying $\Xi(\mathcal{U}^{\otimes n})=\mathcal{U}$. Since the performance operator $\Omega_{n,p}$ satisfies the symmetry
\begin{align}
    [A_P^* \otimes A_\mathbf{I}^{\otimes n} \otimes B_F^* \otimes B_\mathbf{O}^{\otimes n}, \Omega_{n,p}] = 0, \quad \forall A,B\in \mathrm{SU}(d),
\end{align}
we can also assume without loss of generality that the purification superchannel $\Xi$ is twirled: 
\begin{align}
    [A_P^* \otimes A_\mathbf{I}^{\otimes n} \otimes B_F^* \otimes B_\mathbf{O}^{\otimes n}, \mathcal{J}_\Xi] = 0,  \quad \forall A,B\in \mathrm{SU}(d). 
\end{align}

\begin{figure*}[b]
  \centering
  \includegraphics[width=0.8\linewidth]{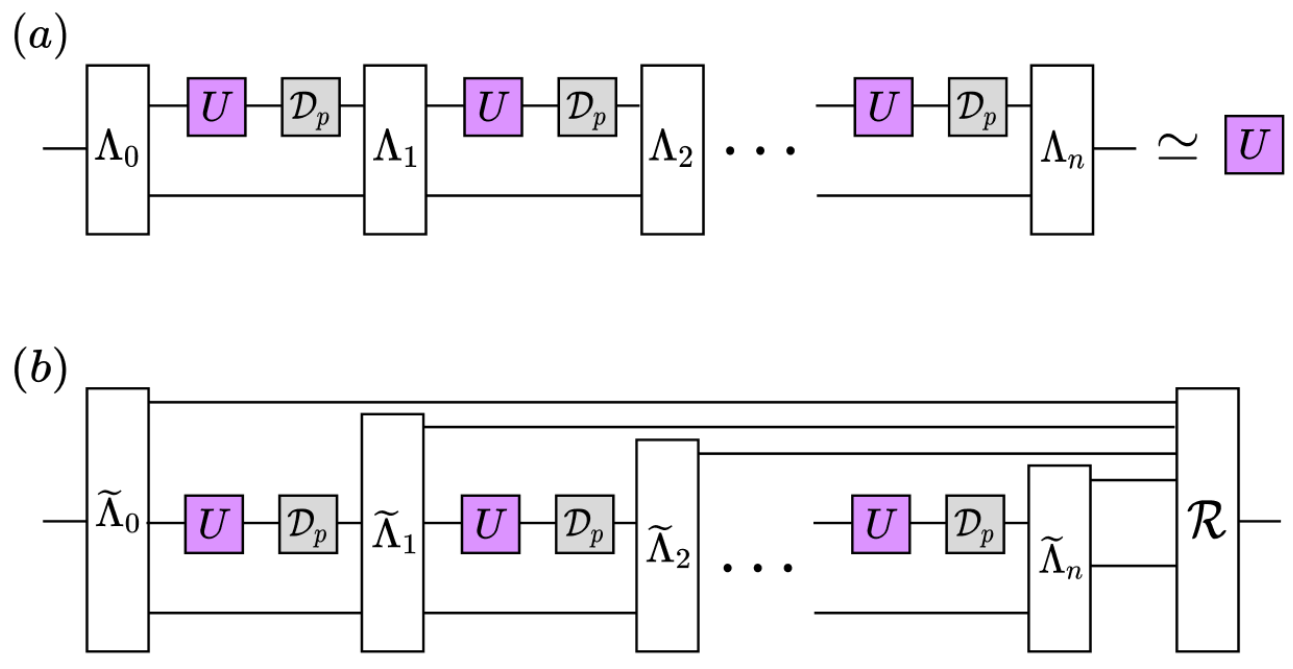}
  \caption{(a) The purification superchannel $\Xi$. (b) The dilated superchannel $\Xi_{\mathrm{Pur}}$, together with the new recovery operation $\mathcal{R}$.}
  \label{fig:dilation}
\end{figure*}

To begin with, we consider a dilation of the twirled superchannel $\Xi$, which is obtained by representing all operations in the superchannel $\Xi$ in its Stinespring dilation form~\cite{chiribella2009theoretical}, and introduce a new recovery operation $\mathcal{R}$ (see FIG.~\ref{fig:dilation}). Choosing $\mathcal{R}$ to be the traceout on the environment recovers the original twirled superchannel $\Xi$, so this assumption is sufficient to prove an upper bound on the achievable fidelity. Let $\Xi_{\mathrm{Pur}}$ denote the dilated sequential superchannel, and let $V_{\mathrm{Pur}}(U_1,\cdots U_n)$ denote the associated isometry $\Xi_{\mathrm{Pur}}(\mathcal{U}_1, \mathcal{U}_2 \cdots \mathcal{U}_n) = V_{\mathrm{Pur}}(\cdot)V_{\mathrm{Pur}}^\dagger$. Define the isometry $V_0$ by $V_0:= V_{\mathrm{Pur}}(\mathbb{I},\mathbb{I}\cdots \mathbb{I})$. For the Hermitian traceless basis $\{T_a\}$, we use $V_{a,r}$ to denote the operation where $T_a$ is inserted in the $r$-th slot of $V_0$: $V_{a,r}:=-i\partial_\theta V_{\mathrm{Pur}}(\mathbb{I},\mathbb{I},\cdots e^{iT_a \theta}, \cdots \mathbb{I})|_{\theta =0}$\footnote{Note that this differention technique is reminiscent of Refs.~\cite{odake2025analytical,bavaresco2025simulating}, which uses differentiation to bound the query complexity of \emph{exact} transformation. Instead, this work employs differentiation together with the inherent symmetry in the task to bound the fidelity in \emph{approximate} transformation as shown below.}.
Note that $V_0$ is an isometry but $V_{a,r}$ is not necessarily an isometry since $T_a$ is generally not a unitary operator. Due to the covariance of the twirled superchannel $\Xi_{\mathrm{Pur}}$, we get 
\begin{align}\label{eq:Cequivariance}
(C_F\otimes S_C)V_0 C_P^\dagger 
&=V_0 \quad \forall C \in \mathrm{SU}(d) \\\label{eq:overlap-equivariance}
(C_F\otimes S_C)V_{a,r} C_P^\dagger
&=\sum_{b}R_{ba}(C)V_{b,r},  \quad \forall C \in \mathrm{SU}(d) 
\end{align}
for some unitary $S_C$ acting on the environment. This further implies 
\begin{align}\label{eq:v0covariance}
C(V_0^\dagger V_{a,r})C^\dagger
&=\sum_bR_{ba}(C)V_0^\dagger V_{b,r}, \quad \forall C \in \mathrm{SU}(d)\\\label{eq:vacovariance}
C(V_{a,r}^\dagger V_{b,s})C^\dagger
&=\sum_{c,e}R_{ca}(C)R_{eb}(C)V_{c,r}^\dagger V_{e,s},  \quad \forall C \in \mathrm{SU}(d). 
\end{align}
The above covariance relations restrict the form of these tensors: 
\begin{lm}\label{lm:covariance}
The tensors in Eq.~\eqref{eq:v0covariance}, ~\eqref{eq:vacovariance} are of the form 
\begin{align}
    &V_0^\dagger V_{a,r} = c_r T_a.\\
    &V_{a,r}^\dagger V_{b,s} = K_{rs}\delta_{ab}\mathbb{I}_d + S_{rs}d_{abc}T_c+iL_{rs}f_{abc}T_c, 
\end{align}
where $c_r \in \mathbb{R}$ satisfy 
\begin{align}
    &\sum_{r=1}^n c_r = 1
\end{align}
and $K,S,L$ are Hermitian matrices satisfying 
\begin{align}
    K_{rr} = 1, \quad S_{rr}=L_{rr}=c_r. 
\end{align}
\end{lm}
\begin{proof}
Let $\mathrm{Adj} := \{X\in \mathcal{L}(\mathbb{C}^d)| \Tr(X) = 0\}$, and let $\mathrm{Hom}_G(V_1, V_2)$ denote the set of $\mathrm{SU}(d)$-covariant linear maps from $V_1$ to $V_2$. Define $\Phi_r: \mathrm{Adj} \to \mathcal{L}(\mathbb{C}^d)$ as 
\begin{align}
\Phi_r(e_a)=V_0^\dagger V_{a,r}. 
\end{align}
The covariance relations of~\eqref{eq:Cequivariance} tell us that $\Phi_r$ is an $\mathrm{SU}(d)$-intertwiner. Now, $\mathcal{L}(\mathbb{C}^d) \simeq \boldsymbol{1} \oplus \mathrm{Adj}$ so $\dim \mathrm{Hom}_G(\mathrm{Adj}, \mathcal{L}(\mathbb{C}^d)) = 1$. Thus, there exists $c_r\in \mathbb{C}$ such that 
\begin{align}\label{eq:Tatransform}
    V_0^\dagger V_{a,r} = c_r T_a.
\end{align}
$V_{a,r}$ inserts a Hermitian operator $T_a$ into the $r$-th slot of $V_0$, and thus we conclude that $V_0^\dagger V_{a,r}$ is Hermitian and thus $c_r\in \mathbb{R}$. Similarly, define $\Psi_{r,s}: \mathrm{Adj} \otimes \mathrm{Adj} \to \mathcal{L}(\mathbb{C}^d)$ by 
\begin{align}
    \Psi_{r,s}(e_a \otimes e_b) = V_{a,r}^\dagger V_{b,s}. 
\end{align}
The covariance relation in~\eqref{eq:overlap-equivariance} tells us that $\Psi_{r,s}$ is an $\mathrm{SU}(d)$ intertwiner. For $d\geq 3$, we have 
\begin{align}
    \mathrm{Adj} \otimes \mathrm{Adj}  \simeq \boldsymbol{1} \oplus 2\mathrm{Adj} \oplus (\mathrm{else}), 
\end{align}
so $\mathrm{Hom}_G(\mathrm{Adj}\otimes \mathrm{Adj}, \mathcal{L}(\mathbb{C}^d)) = \mathrm{Hom}_G(\mathrm{Adj}\otimes \mathrm{Adj}, \boldsymbol{1})\oplus \mathrm{Hom}_G(\mathrm{Adj}\otimes \mathrm{Adj}, \mathrm{Adj})$ and thus
\begin{align}
    V_{a,r}^\dagger V_{b,s} = K_{rs}\delta_{ab}\mathbb{I}_d + S_{rs}d_{abc}T_c+iL_{rs}f_{abc}T_c. 
\end{align}
Since $(V_{a,r}^\dagger V_{b,s})^\dagger = V_{b,s}^\dagger V_{a,r}$, we have $K^*_{rs} = K_{sr}, \, S_{rs}^*=S_{sr}, \,L_{sr}^* = L_{rs}$ so $K,S,L$ are Hermitian. We note that the above relations also hold for $d=2$, since in this case $d_{abc}=0$ and so $S$ may be chosen arbitrarily. 

Next, since $\Xi(\mathcal{U}^{\otimes n}) = \mathcal{U}$, we have $V_{\mathrm{Pur}}(U, \cdots U) = U \otimes |\eta_\theta\rangle$. Choosing $U = e^{iT_a\theta}$ and differentiating with respect to $\theta$ gives 
\begin{align*}
    \sum_{r} V_{a,r} = -i\partial_\theta V_{\mathrm{Pur}}(U, \cdots U)|_{\theta=0} = T_a \otimes |\eta_0\rangle + I\otimes -i\partial_\theta |\eta_\theta\rangle|_{\theta =0}, 
\end{align*}
which implies 
\begin{align}
    \sum_r V_0^\dagger V_{a,r} = T_a + \lambda_a \mathbb{I}_d. 
\end{align}
Due to the covariance relation in Eq.~\eqref{eq:Tatransform}, we have $\lambda_a =0$ and 
\begin{align}
    \sum_{r=1}^n c_r = 1. 
\end{align}
Furthermore, a direct calculation gives 
\begin{align*}
    V_{a,r}^\dagger V_{b,r} = \delta_{ab} + (d_{abc}+if_{abc})V_0^\dagger V_{c,r} = \delta_{ab} + (d_{abc}+if_{abc})(c_r T_c). 
\end{align*}
Therefore, we have 
\begin{align}
    K_{rr} = 1, \quad S_{rr}=L_{rr}=c_r. 
\end{align}
\end{proof}
Let us now derive the upper bound on the achievable fidelity using the Bény-Oreshkov condition. 
\begin{lm}\label{lm:truncation}
Let $Q:=(1-\Delta t)^n, \epsilon:= t(1-\Delta t)^{n-1}$, where $\Delta:=d^2-1$, $t:=\frac{p}{d^2}$. The optimal fidelity for the noisy unitary purification problem satisfies
\begin{align}
    \max_\mathcal{R}f_{\mathrm{Choi}}(\mathcal{R}\circ \Xi(\mathcal{N}_{\mathcal{U},p}^{\otimes n}), \mathcal{U})  \leq \max_{\mathcal{\sigma_E}} f\qty(\rho_{ER}, \sigma_E\otimes \frac{\mathbb{I}_R}{d}) + O_{n,d}(p^2), 
\end{align}
where
\begin{align}
    \rho_{ER} &:= \frac{1}{d}\left[Q|0\rangle \langle 0|_E\otimes \mathbb{I}_R+\sum_{a,r}c_r\sqrt{\epsilon Q}(|0\rangle \langle a,r|+|a,r \rangle \langle 0|)\otimes T_a^T +\sum_{a,r,b,s}\epsilon |a,r\rangle \langle b,s|\otimes (V_{b,s}^\dagger V_{a,r})^T\right]
\end{align}
is the zero- and one-error sector of the Choi state of the complementary channel of the output quantum channel $\Xi_{\mathrm{Pur}}(\mathcal{N}_{\mathcal{U},p}^{\otimes n})$. 
\end{lm}
\begin{proof}
The Bény-Oreshkov condition~\cite{PhysRevLett.104.120501, Faist_2020} states that the optimal recovery fidelity can be captured by how well the total complementary channel is close to a constant channel $\mathcal{T}_\sigma := \Tr(\cdot) \sigma$. Using this condition, 
\begin{align}
    \max_\mathcal{R}f_{\mathrm{Choi}}(\mathcal{R}\circ \Xi(\mathcal{N}_{\mathcal{U},p}^{\otimes n}), \mathcal{U}) &\leq \max_{\mathcal{R}'} f_{\mathrm{Choi}}(\mathcal{R}'\circ \Xi_{\mathrm{Pur}}(\mathcal{N}_{\mathcal{U},p}^{\otimes n}), \mathcal{U}) \nonumber\\
    &= \max_{\mathcal{\sigma_E}} f_{\mathrm{Choi}}( [\Xi_{\mathrm{Pur}}(\mathcal{N}_{\mathrm{id},p}^{\otimes n})]^c, \mathcal{T}_{\sigma_E}) \nonumber\\
    &= \max_{\mathcal{\sigma_E}} f\qty(\frac{1}{d}\sum_{\alpha, \beta}\sqrt{w_\alpha w_\beta}|\alpha\rangle \langle \beta|_E\otimes (V_\beta^\dagger V_\alpha)^T_R, \sigma_E\otimes \frac{\mathbb{I}_R}{d}) \nonumber\\
    &\leq \max_{\mathcal{\sigma_E}} f\qty(\frac{1}{d}\sum_{\alpha, \beta}\sqrt{w_\alpha w_\beta}\mathcal{P}(|\alpha\rangle \langle \beta|)_E\otimes (V_\beta^\dagger V_\alpha)^T_R, \mathcal{P}(\sigma_E)\otimes \frac{\mathbb{I}_R}{d}) \nonumber\\
    &= \max_{0 \leq q \leq 1} \left[\sqrt{q\max_{\sigma_{\leq 1}} f\qty(\frac{1}{d}\sum_{\alpha, \beta}\sqrt{w_\alpha w_\beta}(\Pi|\alpha\rangle \langle \beta|\Pi)_E\otimes (V_\beta^\dagger V_\alpha)^T_R, \sigma_{\leq 1}\otimes \frac{\mathbb{I}_R}{d})}\right.\nonumber\\
    &\left. + \sqrt{(1-q)\max_{\sigma_{\geq 2}} f\qty(\frac{1}{d}\sum_{\alpha, \beta}\sqrt{w_\alpha w_\beta}(\mathbb{I}-\Pi)|\alpha\rangle \langle \beta|(\mathbb{I}-\Pi)_E\otimes (V_\beta^\dagger V_\alpha)^T_R, \sigma_{\geq 2} \otimes \frac{\mathbb{I}_R}{d})}\right]^2 \nonumber\\
    &\leq \max_{\sigma_{\leq 1}} f\qty(\frac{1}{d}\sum_{\alpha, \beta}\sqrt{w_\alpha w_\beta}(\Pi|\alpha\rangle \langle \beta|\Pi)_E\otimes (V_\beta^\dagger V_\alpha)^T_R, \sigma_{\leq 1}\otimes \frac{\mathbb{I}_R}{d}) 
    \nonumber\\
    &+ \Tr\qty[\frac{1}{d}\sum_{\alpha, \beta}\sqrt{w_\alpha w_\beta}(\mathbb{I}-\Pi)|\alpha\rangle \langle \beta|(\mathbb{I}-\Pi)_E\otimes (V_\beta^\dagger V_\alpha)^T_R]\nonumber\\
    &= \max_{\sigma_E} f\qty(\frac{1}{d}\sum_{\alpha, \beta}\sqrt{w_\alpha w_\beta}(\Pi|\alpha\rangle \langle \beta|\Pi)_E\otimes (V_\beta^\dagger V_\alpha)^T_R, \sigma_E\otimes \frac{\mathbb{I}_R}{d}) + O_{n,d}(p^2)
\end{align}
where $\alpha, \beta$ denote the error labels, $w_\alpha, w_\beta$ are the corresponding weights, $\sigma_{\leq 1} := \frac{\Pi(\sigma_E)\Pi}{\Tr [\Pi(\sigma_E)]}, \sigma_{\geq 2} := \frac{(\mathbb{I}-\Pi)(\sigma_E)(\mathbb{I}-\Pi)}{{\Tr [(\mathbb{I}-\Pi)(\sigma_E)]}}, q= \Tr [\Pi(\sigma_E)]$, $R$ denotes the reference system, $(\cdot)^c$ denotes the complementary channel, and $\mathcal{P}$ is the pinching channel $\mathcal{P}(\cdot) = \Pi(\cdot) \Pi+(\mathbb{I}-\Pi)(\cdot)(\mathbb{I}-\Pi)$ with $\Pi$ being the projector onto the zero- and one-error sector. 
\end{proof}
Thus, the remaining problem is to optimize the quantum state $\sigma_E$ in the $\leq 1$ sector. To solve this optimization, let us first analyze the structure of $\rho_{ER}$. 
\begin{lm}\label{lm:rhoERintro}
Let $Q:=(1-\Delta t)^n, \epsilon:= t(1-\Delta t)^{n-1}$, where $\Delta:=d^2-1$, $t:=\frac{p}{d^2}$. Define $c=[c_1, \cdots c_n]^T$ as well as the embedding $J|\bar\psi\rangle = \sum_a |a \rangle \otimes T_a^T|\bar{\psi}\rangle $. The operator $\rho_{ER}$ has the form 
\begin{align}
    \rho_{ER} = \frac{1}{d}\left[Q|0\rangle \langle 0|\otimes \mathbb{I}_R+\sqrt{\epsilon Q}(c\otimes J + c^\dagger \otimes J^\dagger)+\epsilon cc^\dagger \otimes JJ^\dagger+ \epsilon[(K')^T \otimes I + (S')^T\otimes X_d + (L')^T\otimes X_f]\right], 
\end{align}
where 
\begin{align}
\begin{dcases}\label{eq:Xd}
    X_d := \sum_{a,b,c}d_{abc}|a\rangle \langle b| \otimes T_c^T\\
    X_f := -i\sum_{a,b,c}f_{abc}|a\rangle \langle b| \otimes T_c^T. 
\end{dcases}
\end{align}
and 
\begin{align}
    K' = K -cc^T, \, S'=S-cc^T, \, L'= L -cc^T. 
\end{align}
\end{lm}
\begin{proof}
We have $J^\dagger J = \Delta \mathbb{I}$ so $J/\sqrt{\Delta}$ is an isometry. Conversely, $JJ^\dagger/\Delta$ is a projector. Thus, 
\begin{align}
    \rho_{ER} &=\frac{1}{d}\left[Q|0\rangle \langle 0|_E\otimes \mathbb{I}_R+\sum_{a,r}c_r\sqrt{\epsilon Q}(|0\rangle \langle a,r|+|a,r \rangle \langle 0|)\otimes T_a^T  +\sum_{a,r,b,s}\epsilon |a,r\rangle \langle b,s|\otimes (V_{b,s}^\dagger V_{a,r})^T\right] \nonumber\\
    &= \frac{1}{d}\left[Q|0\rangle \langle 0|\otimes \mathbb{I}_R+\sqrt{\epsilon Q}(c\otimes J + c^\dagger \otimes J^\dagger)+\sum_{a,r,b,s}\epsilon |a,r\rangle \langle b,s|\otimes (V_{b,s}^\dagger V_{a,r})^T\right]. 
\end{align}
Now, define 
\begin{align}
    W_{a,r}:= V_{a,r}-c_rV_0 T_a. 
\end{align}
This gives 
\begin{align}
\begin{dcases}
    V_0^\dagger W_{a,r} = 0\\
    V_{a,r}^\dagger V_{b,s} = (W_{a,r}+c_rV_0 T_a)^\dagger (W_{b,s}+c_sV_0 T_b) = W_{a,r}^\dagger W_{b,s} + c_rc_sT_aT_b. 
\end{dcases}
\end{align}
Therefore, 
\begin{align}\label{eq:KSLprime}
    W_{a,r}^\dagger W_{b,s} = K'_{rs}\delta_{ab} + (S'_{rs}d_{abc}+iL'_{rs}f_{abc})T_c, 
\end{align}
where $K',S',L'$ are Hermitian operators defined by 
\begin{align}
    K' = K -cc^T, \, S'=S-cc^T, \, L'= L -cc^T 
\end{align}
and $c=[c_1, \cdots c_n]^T$. This allows us to write $\rho_{ER}$ as 
\begin{align}
    \rho_{ER} &= \frac{1}{d}\left[Q|0\rangle \langle 0|\otimes \mathbb{I}_R+\sqrt{\epsilon Q}(c\otimes J + c^\dagger \otimes J^\dagger)+\sum_{a,r,b,s}\epsilon |a,r\rangle \langle b,s|\otimes (V_{b,s}^\dagger V_{a,r})^T\right] \nonumber\\
    &= \frac{1}{d}\left[Q|0\rangle \langle 0|\otimes \mathbb{I}_R+\sqrt{\epsilon Q}(c\otimes J + c^\dagger \otimes J^\dagger)+\epsilon cc^\dagger \otimes JJ^\dagger + \epsilon[(K')^T \otimes I + (S')^T\otimes X_d + (L')^T\otimes X_f]\right]. 
\end{align}
\end{proof}
To further analyze the structure of $\rho_{ER}$, let us simultaneously diagonalize the operators $X_d, X_f$ using the following lemmas~\ref{lm:XdXf} and~\ref{lm:irrepdecomp}: 
\begin{lm}\label{lm:XdXf}
Let $X_d, X_f$ be operators defined in Eq.~\eqref{eq:Xd}. They are invariant under the tensor product of the adjoint representation and the anti-fundamental representation $R(U) \otimes U^*$. 
\end{lm}
\begin{proof}
First, observe the relations 
\begin{align}\label{eq:dabc}
    d_{\alpha \beta \gamma} = \sum_{a,b,c}d_{abc}[R(U)]_{\alpha a}[R(U)]_{\beta b}[R(U)]_{\gamma c}\\\label{eq:fabc}
    f_{\alpha \beta \gamma} = \sum_{a,b,c}f_{abc}[R(U)]_{\alpha a}[R(U)]_{\beta b}[R(U)]_{\gamma c}. 
\end{align}
To see this, note that $d_{abc}$ has the unitarily invariant expression
\begin{align}
    d_{abc} = \frac{1}{2d}\Tr(\{T_a, T_b\}T_c) = \frac{1}{2d}\Tr(\{UT_aU^\dagger, UT_bU^\dagger\} UT_cU^\dagger). 
\end{align}
From this we indeed obtain 
\begin{align}
    d_{abc} = \frac{1}{2d}\Tr(\{UT_aU^\dagger, UT_bU^\dagger\} UT_cU^\dagger) = \frac{1}{2d}\sum_{\alpha, \beta \gamma}[R(U)]_{\alpha a}[R(U)]_{\beta b}[R(U)]_{\gamma c} d_{\alpha \beta \gamma}, 
\end{align}
which is Eq.~\eqref{eq:dabc}. Similarly, $f_{abc}$ has the unitarily invariant expression 
\begin{align}
    f_{abc} = \frac{1}{2id}\Tr([T_a, T_b]T_c) = \frac{1}{2id}\Tr([UT_aU^\dagger, UT_bU^\dagger] UT_cU^\dagger). 
\end{align}
Thus,  
\begin{align}
    f_{abc} = \frac{1}{2id}\Tr([UT_aU^\dagger, UT_bU^\dagger] UT_cU^\dagger) = \frac{1}{2id}\sum_{\alpha, \beta \gamma}[R(U)]_{\alpha a}[R(U)]_{\beta b}[R(U)]_{\gamma c} f_{\alpha \beta \gamma}, 
\end{align}
which is Eq.~\eqref{eq:fabc}. Using these relations, we can show that 
\begin{align}
    (R(U)\otimes U^*)X_d(R(U)\otimes U^*)^\dagger &= \sum_{a,b,c}d_{abc}R(U)|a\rangle \langle b|R(U)^\dagger \otimes U^*T_c^TU^T\nonumber\\
    &= \sum_{\alpha, \beta, \gamma} \qty(\sum_{a,b,c}d_{abc}[R(U)]_{\alpha a}[R(U)]_{\beta b}[R(U)]_{\gamma c})|\alpha\rangle \langle \beta| \otimes T_\gamma^T \nonumber\\
    &= \sum_{\alpha, \beta, \gamma} d_{\alpha \beta \gamma}|\alpha\rangle \langle \beta| \otimes T_\gamma^T
\end{align}
and similarly, 
\begin{align}
    (R(U)\otimes U^*)X_f(R(U)\otimes U^*)^\dagger &=-i \sum_{a,b,c}f_{abc}R(U)|a\rangle \langle b|R(U)^\dagger \otimes U^*T_c^TU^T\nonumber\\
    &= -i\sum_{\alpha, \beta, \gamma} \qty(\sum_{a,b,c}f_{abc}[R(U)]_{\alpha a}[R(U)]_{\beta b}[R(U)]_{\gamma c})|\alpha\rangle \langle \beta| \otimes T_\gamma^T \nonumber\\
    &= -i\sum_{\alpha, \beta, \gamma} f_{\alpha  \beta \gamma}|\alpha\rangle \langle \beta| \otimes T_\gamma^T. 
\end{align}
\end{proof}

\begin{lm}\label{lm:irrepdecomp}
Consider the irreducible decomposition of the tensor product of the adjoint representation and the anti-fundamental representation, 
\begin{align}
    \mathrm{Adj}\otimes \bar{\Box} \simeq \bar{\Box}\oplus R_+ \oplus R_-, 
\end{align}
where $R_+ \simeq \mathcal{U}_{\gamma^+}$, $R_- \simeq \mathcal{U}_{\gamma^-}$, and  $\mathcal{U}_{\gamma^{\pm}}$ are Weyl modules of the mixed Schur-Weyl decomposition for the irrep labels $\gamma_+ := (\raisebox{0em}{\scalebox{0.4}{\ydiagram{1}}}, \raisebox{0em}{\scalebox{0.4}{\ydiagram{2}}}))$, and $\gamma_- :=  (\raisebox{0em}{\scalebox{0.4}{\ydiagram{1}}}, \raisebox{0.3em}{\scalebox{0.4}{\ydiagram{1,1}}})$, with $\dim R_+ = \frac{(d-1)d(d+2)}{2}, \dim R_- = \frac{(d-2)d(d+1)}{2}$. The operators $X_d, X_f$ both act as constants on each irrep $\bar{\Box}, R_+ ,R_-$. The eigenvalues of $X_f$ can be computed as 
\begin{align}
\begin{dcases}
\frac{d^2}{2} \quad (\bar{\Box}, \mathrm{multiplicity:}d)\\
-\frac{d}{2} \quad (R_+, \mathrm{multiplicity:}\frac{(d-1)d(d+2)}{2})\\
+\frac{d}{2} \quad (R_-, \mathrm{multiplicity:}\frac{(d-2)d(d+1)}{2})
\end{dcases}
\end{align}
while the eigenvalues of $X_d$ can be computed as 
\begin{align}
\begin{dcases}
\frac{d^2}{2}-2 \quad (\bar{\Box}, \mathrm{multiplicity:}d)\\
\frac{d}{2}-1 \quad (R_+, \mathrm{multiplicity:}\frac{(d-1)d(d+2)}{2})\\
-\frac{d}{2}-1 \quad (R_-, \mathrm{multiplicity:}\frac{(d-2)d(d+1)}{2}). 
\end{dcases}
\end{align}
\end{lm}
\begin{proof}
In terms of the irrep labels, the decomposition can be understood as 
\begin{align}
    \mathrm{Adj} \otimes \bar{\Box} \simeq (\raisebox{0em}{\scalebox{0.4}{\ydiagram{1}}}, \raisebox{0em}{\scalebox{0.4}{\ydiagram{1}}}) \otimes \overline{\scalebox{0.4}{\ydiagram{1}}} \simeq (\phi, \raisebox{0em}{\scalebox{0.4}{\ydiagram{1}}})\oplus (\raisebox{0em}{\scalebox{0.4}{\ydiagram{1}}}, \raisebox{0em}{\scalebox{0.4}{\ydiagram{2}}})\oplus (\raisebox{0em}{\scalebox{0.4}{\ydiagram{1}}}, \raisebox{0.3em}{\scalebox{0.4}{\ydiagram{1,1}}}).  
\end{align}
Thus, one can use the Weyl dimension formula~\cite{grinko2025mixed}
\begin{align}
    \dim \mathcal{U}_\lambda = \prod_{i<j} \frac{\lambda_i-\lambda_j + j-i}{j-i}
\end{align}
to confirm that  $\dim R_+ = \frac{(d-1)d(d+2)}{2}, \dim R_- = \frac{(d-2)d(d+1)}{2}$. Since we know from Lemma~\ref{lm:XdXf} that the operators $X_d, X_f$ are invariant under the tensor products of the adjoint representation and the anti-fundamental representation, Schur's lemma tells us that they act as constants on the irreducible spaces. 

To compute the eigenvalues of $X_d, X_f$, let us define the Casimir operator (see Appendix~\ref{ap:preliminaries}). In our case, the total Casimir operator has the form 
\begin{align}
    C_{\mathrm{tot}} = C_{\mathrm{Adj}} \otimes \mathbb{I} + \mathbb{I}\otimes C_{\bar{\Box}} -2\sum_{a} F_a \otimes T_a^T, 
\end{align}
where $F_a$ is the generator corresponding to the basis $T_a$ in the adjoint representation. We have 
\begin{align}
    X_f = -i\sum_{a,b,c} f_{cab}|a\rangle \langle b| \otimes T_c^T = \frac{1}{2}\sum_{a,b,c} [F_c]_{ab}|a\rangle \langle b| \otimes T_c^T = \frac{1}{2}\sum_a F_a \otimes T_a^T.
\end{align}
Thus, the eigenvalues of $X_f$ can be obtained by computing the eigenvalues of $C_{\mathrm{tot}}, C_{\mathrm{Adj}}, C_{\bar{\Box}}$. First,  
\begin{align}
 [C_{\mathrm{Adj}}]_{db}  = \sum_{a,c} [F_a]_{dc}[F_a]_{cb} = \sum_{a,c} (2if_{acd})(2if_{abc}) = 4\sum_{a,c}f_{acb}f_{acd}. 
\end{align}
We also have 
\begin{align*}
    \sum_a [T_a, [T_a, T_b]] = \sum_{a,c} [T_a, 2if_{abc} T_c] = -4\sum_{a,c}f_{abc}f_{acd} T_d
\end{align*}
and 
\begin{align*}
    \sum_a [T_a, [T_a, T_b]] =\sum_a (T_a^2 T_b+T_bT_a^2-2T_aT_bT_a) = 2(d^2-1)T_b + 2T_b = 2d^2 T_b. 
\end{align*}
The above two lines show 
\begin{align*}
    \sum_{a,c}f_{acb}f_{acd} = \frac{d^2}{2}\delta_{bd}
\end{align*}
so we get 
\begin{align}
    C_{\mathrm{Adj}} = 2d^2 \mathbb{I}. 
\end{align}
We also have 
\begin{align}
    C_{\bar{\Box}} = \sum_a(-T_a^T)^2 = (d^2-1)\mathbb{I}. 
\end{align}
Now, if an irrep is represented by a Young diagram $\lambda=(\lambda_1, \lambda_2, \cdots \lambda_d)$ and $\sum_i \lambda_i = m$, one has~\cite{Piddock:2018xyi} 
\begin{align}
    C(\lambda) = d^2m + d\sum_{i=1}^{d} \lambda_i(\lambda_i+1-2i)-m^2.
\end{align}
Thus, 
\begin{align}
\begin{dcases}
    C_{\mathrm{tot}}(\bar{\Box}) = d^2(d-1)-d(d-1)(d-2)-(d-1)^2 = d^2-1\\
    C_{\mathrm{tot}}(R_+) = d^2(d+1)+d(6-(d-1)(d-2))-(d+1)^2 = 3d^2+2d-1\\
    C_{\mathrm{tot}}(R_-) = d^2(d-1)+d(-d^2+5d-4)-(d-1)^2 = 3d^2-2d-1. 
\end{dcases}
\end{align}
This leads us to conclude that the eigenvalues of $X_f$ are  
\begin{align}
\begin{dcases}
\frac{d^2}{2} \quad (\bar{\Box}, \textrm{multiplicity:}d)\\
-\frac{d}{2} \quad (R_+, \textrm{multiplicity:}\frac{(d-1)d(d+2)}{2})\\
+\frac{d}{2} \quad (R_-, \textrm{multiplicity:}\frac{(d-2)d(d+1)}{2}). 
\end{dcases}
\end{align}
\noindent Next, define $J: \bar{\Box} \to \mathrm{Adj} \otimes \bar{\Box}$, where 
\begin{align}
    J|\bar{\psi}\rangle \to \sum_a |a\rangle \otimes T_a^T|\bar{\psi}\rangle. 
\end{align}
We have $J^\dagger J = \Delta \mathbb{I}$ so $J/\sqrt{\Delta}$ is an isometric embedding of $\bar{\Box}$ into $\mathrm{Adj} \otimes \bar{\Box}$. Conversely, $JJ^\dagger/\Delta$ is a projector onto the $\bar{\Box}$ component in $\mathrm{Adj} \otimes \bar{\Box}$. Since 
\begin{align}
    JJ^\dagger = \sum_{a,b} |a\rangle \langle b| \otimes (T_a^T T_b^T) = \sum_{a,b} |a\rangle \langle b| \otimes (\delta_{ab} \mathbb{I}+ (d_{abc}-if_{abc})T_c^T), 
\end{align}
we have the constraint $1+x_d+x_f=d^2-1$ on $\bar{\Box}$ and $1+x_d+x_f=0$ on $R_+, R_-$, where $x_d, x_f$ are eigenvalues of the operators $X_d, X_f$, respectively. Thus, the eigenvalues of $X_d$ can be evaluated as 
\begin{align}
\begin{dcases}
\frac{d^2}{2}-2 \quad (\bar{\Box}, \textrm{multiplicity:}d)\\
\frac{d}{2}-1 \quad (R_+, \textrm{multiplicity:}\frac{(d-1)d(d+2)}{2})\\
-\frac{d}{2}-1 \quad (R_-, \textrm{multiplicity:}\frac{(d-2)d(d+1)}{2}). 
\end{dcases}
\end{align}
\end{proof}
\noindent Putting lemmas~\ref{lm:rhoERintro}, \ref{lm:XdXf}, \ref{lm:irrepdecomp} together, we know that $\rho_{ER}$ has the following structure. 
\begin{lm}\label{lm:rhoER}
Let $Q:=(1-\Delta t)^n, \epsilon:= t(1-\Delta t)^{n-1}$, where $\Delta:=d^2-1$, $t:=\frac{p}{d^2}$. The operator $\rho_{ER}$ has the form 
\begin{align}
    \rho_{ER} \simeq \rho_0 \oplus \rho_+ \oplus \rho_-, 
\end{align}
where 
\begin{align}
\begin{dcases}
    \rho_0 = \frac{1}{d}\mqty(Q & \sqrt{\Delta \epsilon Q} c^\dagger\\ \sqrt{\Delta \epsilon Q} c & \epsilon(\Delta cc^\dagger + A_0^T))\otimes I_d,\\
    \rho_+ = \frac{\epsilon}{d} A_+^T\otimes I_{d_+}\\
    \rho_- = \frac{\epsilon}{d}A_-^T\otimes I_{d_-}.
\end{dcases}
\end{align}
Here, $c:=[c_1, \cdots c_n]^T$, and 
\begin{align}
\begin{dcases}
    A_0 = K' + \qty(\frac{d^2}{2}-2) S' + \frac{d^2}{2} L'\\
    A_+ = K' + \qty(\frac{d}{2}-1) S' - \frac{d}{2} L'\\
    A_- = K' - \qty(\frac{d}{2}+1) S' + \frac{d}{2} L'. 
\end{dcases}
\end{align}
where $K',S',L'$ are Hermitian operators defined by 
\begin{align}
    K' := K -cc^T, \, S':=S-cc^T, \, L':= L -cc^T, 
\end{align}
\end{lm}
\begin{proof}
By lemma~\ref{lm:irrepdecomp}, we have 
\begin{align}\label{eq:KSLprimeisomorphism}
    (K')^T \otimes I + (S')^T\otimes X_d + (L')^T\otimes X_f \simeq (A_0)^T \otimes I_d \oplus (A_+)^T \otimes I_{d_+} \oplus (A_-)^T \otimes I_{d_-}, 
\end{align}
where $d_+=\dim R_+, \,d_- = \dim R_-$ and 
\begin{align}
\begin{dcases}
    A_0 = K' + \qty(\frac{d^2}{2}-2) S' + \frac{d^2}{2} L'\\
    A_+ = K' + \qty(\frac{d}{2}-1) S' - \frac{d}{2} L'\\
    A_- = K' - \qty(\frac{d}{2}+1) S' + \frac{d}{2} L'. 
\end{dcases}
\end{align}
Using the form of $\rho_{ER}$ in Lemma~\ref{lm:rhoERintro}, and noting that $J/\sqrt{\Delta}$ is a unique isometric embedding that maps $\bar{\Box}$ into $\mathrm{Adj} \otimes \bar{\Box} \simeq \bar{\Box} \oplus R_+ \oplus R_-$, we can write 
\begin{align}
    \rho_{ER} &\simeq \rho_0 \oplus \rho_+ \oplus \rho_-, 
\end{align}
where 
\begin{align}
\begin{dcases}
    \rho_0 = \frac{1}{d}\mqty(Q & \sqrt{\Delta \epsilon Q} c^\dagger\\ \sqrt{\Delta \epsilon Q} c & \epsilon(\Delta cc^\dagger + A_0^T))\otimes I_d,\\
    \rho_+ = \frac{\epsilon}{d} A_+^T\otimes I_{d_+}\\
    \rho_- = \frac{\epsilon}{d}A_-^T\otimes I_{d_-}. 
\end{dcases}
\end{align}
\end{proof}
Using Lemmas~\ref{lm:rhoERintro}-\ref{lm:rhoER}, we now go on to optimize the fidelity in Lemma~\ref{lm:truncation}. However, a  direct optimization of $\sigma_E$ is difficult, and so we invoke the following lemma: 
\begin{lm}\label{lm:symmetry-enhancedLP}
Let $t:=\frac{p}{d^2}$, and 
\begin{align}
    \overline{\sigma}_E = (1-tz)|0\rangle \langle 0|+ \frac{t\mathbb{I}_{\mathrm{Adj}}}{\Delta}\otimes R, \quad R\geq 0, \, \Tr R=z, \, 0\leq tz \leq 1. 
\end{align}
It holds that 
\begin{align}
    \max_R f_{\mathrm{root}}(\rho_{ER}, \overline{\sigma}_E \otimes \frac{
    \mathbb{I}_R}{d}) &\geq \max_{\sigma_E} f_{\mathrm{root}}(\rho_{ER}, \sigma_E \otimes \frac{\mathbb{I}_R}{d}), 
\end{align}
where $f_{\mathrm{root}}(\rho, \sigma) := \Tr(\sqrt{\rho^{1/2}\sigma \rho^{1/2}})$ is the square root fidelity. Note that $f_{\mathrm{root}}(\rho, \sigma)^2 = f(\rho, \sigma)$. 
\end{lm}
\begin{proof}
We first perform the transformation $\Xi_{\mathrm{Pur}}(C^\dagger C,\cdots C^\dagger T_aC, \cdots C^\dagger C)$ and observe that 
\begin{align}
    \rho_{ER} = ([\boldsymbol{1}\oplus R(C)\otimes \mathbb{I}_n]\otimes C^*)\rho_{ER}([\boldsymbol{1}\oplus R(C)\otimes \mathbb{I}_n]\otimes C^*)^\dagger, 
\end{align}
where $R(C)$ is the adjoint representation of $C$,  $\boldsymbol{1}$ is the identity acting on the no-error label $|0\rangle$, and $\mathbb{I}_n$ is the identity operator acting on the position label. On the other hand, the environment state transforms as 
\begin{align}
    \sigma_E \otimes \frac{\mathbb{I}_R}{d} &\to  ([\boldsymbol{1}\oplus R(C)\otimes \mathbb{I}_n]\otimes C^*)\qty(\sigma_E \otimes \frac{\mathbb{I}_R}{d})([\boldsymbol{1}\oplus R(C)\otimes \mathbb{I}_n]\otimes C^*)^\dagger \nonumber\\
    &= ([\boldsymbol{1}\oplus R(C)\otimes \mathbb{I}_n]\otimes \mathbb{I}_d)\qty(\sigma_E \otimes \frac{\mathbb{I}_R}{d})([\boldsymbol{1}\oplus R(C)\otimes \mathbb{I}_n]\otimes \mathbb{I}_d)^\dagger. 
\end{align}
Now, define 
\begin{align}
    \overline{\sigma}_E  = \int \mathrm{d}C [\boldsymbol{1}\oplus R(C)\otimes \mathbb{I}_n]\sigma_E [\boldsymbol{1}\oplus R(C)\otimes \mathbb{I}_n]^\dagger. 
\end{align}
The concavity of the square root fidelity in the second argument shows that 
\begin{align}
    f_{\mathrm{root}}(\rho_{ER}, \overline{\sigma}_E \otimes \frac{
    \mathbb{I}_R}{d}) &\geq \int \mathrm{d}C f_{\mathrm{root}}(\rho_{ER}, [\boldsymbol{1}\oplus R(C)\otimes \mathbb{I}_n]\sigma_E [\boldsymbol{1}\oplus R(C)\otimes \mathbb{I}_n]^\dagger \otimes \frac{
    \mathbb{I}_R}{d})\nonumber\\
    &= f_{\mathrm{root}}(\rho_{ER}, \sigma_E \otimes \frac{\mathbb{I}_R}{d}). 
\end{align}
Thus, to obtain an upper bound on the fidelity, we can maximize over $\overline{\sigma}_E$. Since 
\begin{align}
    [(\boldsymbol{1}\oplus R(C)\otimes \mathbb{I}_n), \overline{\sigma}_E] = 0, 
\end{align}
one can show that there is no coherence between the no-error sector and the single-error sector and that the $\mathrm{Adj}$ space would be proportional to the identity. Thus, $\overline{\sigma}_E$ may be assumed to be of the form 
\begin{align}
    \overline{\sigma}_E = (1-tz)|0\rangle \langle 0|+ \frac{t\mathbb{I}_{\mathrm{Adj}}}{\Delta}\otimes R, \quad R\geq 0, \, \Tr R=z, \, 0\leq tz \leq 1
\end{align}
\end{proof}
Using Lemma~\ref{lm:truncation}, \ref{lm:rhoER}, and \ref{lm:symmetry-enhancedLP}, one can derive the upper bound on the achievable fidelity. 
\begin{lm}\label{lm:LPsolution}
Let $Q:=(1-\Delta t)^n, \epsilon:= t(1-\Delta t)^{n-1}$, where $\Delta:=d^2-1$, $t:=\frac{p}{d^2}$, and let 
\begin{align}
    \overline{\sigma}_E = (1-tz)|0\rangle \langle 0|+ \frac{t\mathbb{I}_{\mathrm{Adj}}}{\Delta}\otimes R, \quad R\geq 0, \, \Tr R=z. 
\end{align}
The upper bound on the achievable fidelity in Lemma~\ref{lm:symmetry-enhancedLP} has the form 
\begin{align}
    f_{\mathrm{root}}\qty(\rho_{ER}, \bar{\sigma}_E \otimes \frac{\mathbb{I}_R}{d}) &= \sqrt{Q(1-tz)}+\frac{t}{\sqrt{\Delta}}\qty(f_{\mathrm{root}}(A_0^T, R)+ \frac{d_+}{d}f_{\mathrm{root}}(A_+^T, R) + \frac{d_-}{d}f_{\mathrm{root}}(A_-^T, R))+o(t)\nonumber\\
    &= 1-\frac{t}{2}(n\Delta + z) + t\sqrt{\frac{z}{\Delta}}\mathcal{S}(\hat{R}) + o(t)
\end{align}
where we have defined $R = z\hat{R}$ and 
\begin{align}
    \mathcal{S}(\hat{R}) := f_{\mathrm{root}}(A_0^T, \hat{R})+ \frac{d_+}{d}f_{\mathrm{root}}(A_+^T, \hat{R}) + \frac{d_-}{d}f_{\mathrm{root}}(A_-^T, \hat{R}). 
\end{align}
\end{lm}
\begin{proof}
Since
\begin{align}
    \overline{\sigma}_E \otimes \frac{\mathbb{I}_R}{d} &= (1-tz)|0\rangle \langle 0|\otimes \frac{\mathbb{I}_R}{d} + \qty[\frac{t}{d\Delta}R\otimes \mathbb{I}_{d} \oplus \frac{t}{d\Delta}R\otimes \mathbb{I}_{d_+} \oplus \frac{t}{d\Delta}R \otimes \mathbb{I}_{d_-}]\nonumber\\
    & \simeq \sigma_0\oplus \sigma_+ \oplus \sigma_-, 
\end{align}
where 
\begin{align}
\begin{dcases}
     \sigma_0 = \frac{1}{d}\mqty(1-tz & 0\\ 0 & \frac{t}{\Delta}R) \otimes \mathbb{I}_d\\
     \sigma_+ = \frac{t}{d\Delta} R \otimes \mathbb{I}_{d_+}, \\
     \sigma_- = \frac{t}{d\Delta} R \otimes \mathbb{I}_{d_-}. 
\end{dcases}
\end{align}
We get the expression for the square root fidelity
\begin{align}
    f_{\mathrm{root}}\qty(\rho_{ER}, \sigma_E \otimes \frac{\mathbb{I}_R}{d}) &= f_{\mathrm{root}}(LL^\dagger, DD^\dagger) + f_{\mathrm{root}}(\rho_+, \sigma_+) + f_{\mathrm{root}}(\rho_-, \sigma_-)\nonumber\\
    &= \lVert L^\dagger D \rVert_1 + \frac{d_+}{d}\sqrt{\frac{\epsilon t}{\Delta}}f_{\mathrm{root}}(A_+^T, R) + \frac{d_-}{d}\sqrt{\frac{\epsilon t}{\Delta}}f_{\mathrm{root}}(A_-^T, R)
\end{align}
where
\begin{align}
    L = \mqty(\sqrt{Q} & 0\\\sqrt{\Delta \epsilon} c & \sqrt{\epsilon}(A_0^T)^{1/2}), \quad D = \mqty(\sqrt{1-tz} & 0\\ 0  & \sqrt{\frac{tR}{\Delta}}). 
\end{align}
Substituting, one obtains 
\begin{align}
    f_{\mathrm{root}}\qty(\rho_{ER}, \sigma_E \otimes \frac{\mathbb{I}_R}{d}) =& \left \lVert \mqty(\sqrt{Q(1-tz)} & \sqrt{\epsilon t} c^\dagger R^{1/2}\\ 0 & \sqrt{\frac{\epsilon t}{\Delta}} (A_0^T)^{1/2}R^{1/2}) \right \rVert_1+ \frac{d_+}{d}\sqrt{\frac{\epsilon t}{\Delta}}f_{\mathrm{root}}(A_+^T, R) + \frac{d_-}{d}\sqrt{\frac{\epsilon t}{\Delta}}f_{\mathrm{root}}(A_-^T, R).
\end{align}
Now, define $a_t := \sqrt{Q(1-tz)}, b_t := \sqrt{\epsilon t} c^\dagger R^{1/2}, C_t :=\sqrt{\frac{\epsilon t}{\Delta}} (A_0^T)^{1/2}R^{1/2}$, and $x_t:=b_t^\dagger/a_t$. 
As long as $z=O(1)$, we have $a_t \to 1, b_t=O(t), C_t=O(t)$ and 
\begin{align}
    \mqty(a_t & b_t \\ 0 & C_t) \mathrm{exp}\mqty(0 & -x^\dagger\\ x & 0) = \mqty(a_t & 0 \\ 0 & C_t)+O(t^2). 
\end{align}
Therefore, 
\begin{align}
    \left \lVert \mqty(a_t & b_t \\ 0 & C_t) \right \rVert_1 =  \left \lVert\mqty(a_t & b_t \\ 0 & C_t) \mathrm{exp}\mqty(0 & -x^\dagger\\ x & 0)\right \rVert_1 = a_t + \lVert C_t  \rVert_1 + o(t). 
\end{align}
Now, we show $z=O(1)$. Set $y=tz$ and consider pinching $\rho_{ER}, \bar{\sigma}_E$ as 
\begin{align*}
    &\rho_{ER} \to \Pi_0 \rho_{ER} \Pi_0 + \Pi_1 \rho_{ER} \Pi_1\\
    &\bar{\sigma}_E \to \Pi_0 \bar{\sigma}_E \Pi_0 + \Pi_1 \bar{\sigma}_E \Pi_1, 
\end{align*}
where $\Pi_0, \Pi_1$ denote the projector onto the no-error and single-error sector. This shows
\begin{align}
    f_{\mathrm{root}}\qty(\rho_{ER}, \sigma_E \otimes \frac{\mathbb{I}_R}{d}) \leq \sqrt{(1-Q)y}+\sqrt{Q(1-y)}, 
\end{align}
where we have used $f_{\mathrm{root}}(X,Y) \leq \sqrt{\Tr X \Tr Y}$. One can also choose $\sigma_E = |0\rangle \langle 0|$ so we have 
\begin{align}
    \sqrt{Q} \leq f_{\mathrm{root}}\qty(\rho_{ER}, \sigma_E \otimes \frac{\mathbb{I}_R}{d}) \leq \sqrt{(1-Q)y}+\sqrt{Q(1-y)}. 
\end{align}
Expanding the inequality shows for $0\leq y \leq 1$ that $1-\sqrt{1-y} = \frac{y}{1+\sqrt{1-y}}\geq \frac{y}{2}$ and so 
\begin{align}
    y \leq \frac{4(1-Q)}{Q}. 
\end{align}
Substituting $Q= 1-n\Delta t +O(t^2)$, one sees that $y=O(t)$ and so $z=O(1)$, which justifies the earlier assumption. 
\end{proof}
Finally, the $z$ dependent part is a quadratic function of $\sqrt{z}$ and so $z^\star = \frac{\mathcal{S}(\hat{R})^2}{\Delta}$ achieves the maximum value. Noting that the target optimal value is the \textit{squared} fidelity, we get 
\begin{align}
    f_{n,d}(p) \leq 1-t(n\Delta -\frac{1}{\Delta} \max_{\hat{R}\geq 0, \Tr \hat{R} = 1}[\mathcal{S}(\hat{R})]^2)+ o(t). 
\end{align}
Recalling Eq.~\eqref{eq:KSLprime} and Eq.~\eqref{eq:KSLprimeisomorphism}, one can see that $A_0, A_\pm$ are all positive semidefinite. Thus, we can use the inequality $f_{\mathrm{root}}(X,Y) \leq \sqrt{\Tr X \Tr Y}$ together with the Cauchy-Schwarz inequality $s:= \sum_{r=1}^n c_r^2 \geq \frac{1}{n}$ to show 
\begin{align}
    \mathcal{S}(\hat{R}) &\leq \sqrt{\Tr A_0 \Tr \hat{R}} + \frac{d_+}{d} \sqrt{\Tr A_+ \Tr \hat{R}} + \frac{d_-}{d} \sqrt{\Tr A_- \Tr \hat{R}} \nonumber \\
    & = \sqrt{n+d^2-2-\Delta s} + \frac{(d-1)(d+2)}{2}\sqrt{n-1} + \frac{(d-2)(d+1)}{2}\sqrt{n-1} \nonumber \\
    & \leq \sqrt{n+d^2-2-\frac{d^2-1}{n}} + (d^2-2)\sqrt{n-1} \nonumber\\
    &= \sqrt{n-1}\qty(\Delta -1 + \sqrt{1+\frac{\Delta}{n}}). 
\end{align}
Therefore, we have 
\begin{align}
    f_{n,d}(p) &\leq 1- \frac{p}{d^2}\qty(n\Delta-\frac{n-1}{\Delta}\qty(\Delta-1+\sqrt{1+\frac{\Delta}{n}})^2) + o(p)\nonumber\\
    &= 1-\frac{1}{d^2}\qty[1+\frac{(\Delta-1)(\Delta +2+2\sqrt{1+\frac{\Delta}{n}})}{\qty(1+ \sqrt{1+\frac{\Delta}{n}})^2}] \frac{p}{n} + o(p)
\end{align}
which proves Theorem 1. 
\end{proof}

\section{Proof of Theorem 2}
Here we prove Theorem 2. We first give an information-theoretic proof using the Bény-Oreshkov theorem. We then give a more direct evaluation of the achievable fidelity by constructing a concrete encoding and decoding circuit. 
\subsection{Information-theoretic proof}
\begin{thm}[Achievability with parallel strategies, Theorem 2 in the main text]
Let $\Delta = d^2-1$. For $n=kd+1$, the entanglement-assisted $\mathrm{SU}(d)$ covariant code satisfies: 
\begin{align}
C_{n,d}\leq C_{n,d}^{\star},
\qquad
C_{n,d}^{\star}:=
\frac{1}{d^2n}\qty[1+\frac{(\Delta-1)(\Delta +2+2\sqrt{1+\frac{\Delta}{n}})}{\qty(1+ \sqrt{1+\frac{\Delta}{n}})^2}]. 
\end{align}
Combined with Theorem 1, we have $C_{n,d}=C_{n,d}^\star$ for $n=kd+1$. 
\end{thm}
\begin{proof}
Let $n=kd+1$ and choose the partition $\alpha = (k+1, k, \cdots k)$. The Weyl module has dimension $d$, and has the representation 
\begin{align}
    f_{\alpha}(U) = (\mathrm{det} U)^k U. 
\end{align}
Now, define the encoder 
\begin{align}
    E_{n,d}|\psi\rangle = \frac{1}{\sqrt{m_{\alpha}}}\sum_{j=1}^{m_{\alpha}}U_{\mathrm{Sch}}^\dagger(|\psi \rangle_{\mathcal{U}_{\alpha}} \otimes |\alpha\rangle \otimes |j\rangle_{\mathcal{S}_{\alpha}}) \otimes |j\rangle_E. 
\end{align}
This satisfies the covariant condition
\begin{align}\label{eq:Ucovap}
    (U^{\otimes n} \otimes \mathbb{I}_E)E_{n,d} = E_{n,d} U 
\end{align}
for $U\in \mathrm{SU}(d)$. One can then show 
\begin{align}
    E_{n,d}^\dagger \mathcal{O}  E_{n,d} = \frac{1}{m_{\alpha}} \Tr_{\mathcal{S}_{\alpha}}(\Pi_{\alpha}\mathcal{O} \Pi_{\alpha})
\end{align}
for any observable $\mathcal{O}\in \mathcal{L}(\mathbb{C}^d)^{\otimes n}$, where $\Pi_{\alpha}$ is the Young projector. This implies 
\begin{align}
    E_{n,d}^\dagger T_{a,1} E_{n,d}  = E_{n,d}^\dagger T_{a,2} E_{n,d} = \cdots E_{n,d}^\dagger T_{a,n} E_{n,d}. 
\end{align}
Define $S_a := \sum_{r=1}^n T_{a,r}$. Differentiating the covariance condition~\eqref{eq:Ucovap} gives 
\begin{align}
    S_a E_{n,d} = E_{n,d} T_a, \quad E_{n,d}^\dagger T_{a,r} E_{n,d} = \frac{1}{n}T_a. 
\end{align}
 For $r=s$, we have 
\begin{align}
    E_{n,d}^\dagger T_{a,r} T_{b,r} E_{n,d} = E_{n,d}^\dagger (\delta_{ab} + (d_{abc}+if_{abc}) T_{c,r}) E_{n,d} =  \delta_{ab} + \frac{1}{n} (d_{abc}+if_{abc}) T_{c}. 
\end{align}
For $r\neq s$, we have 
\begin{align}
    E_{n,d}^\dagger \qty(\sum_{r}\sum_{s}T_{a,r} T_{b,s})E_{n,d} = E_{n,d}^\dagger S_a S_b E_{n,d} = T_a T_b = \delta_{ab}+(d_{abc}+if_{abc})T_c. 
\end{align}
Since 
\begin{align}
    E_{n,d}^\dagger \qty(\sum_{r}\sum_{s}T_{a,r} T_{b,s}) E_{n,d}= n\delta_{ab} +  (d_{abc}+if_{abc}) T_{c}  +  E_{n,d}^\dagger \qty(\sum_{r\neq s}T_{a,r} T_{b,s}) E_{n,d},
\end{align}
we have 
\begin{align}
    E_{n,d}^\dagger T_{a,r} T_{b,s} E_{n,d} = -\frac{1}{n}\delta_{ab}. \quad (r\neq s)
\end{align}
Comparing with the upper bound argument, we get
\begin{align}
    c_r = \frac{1}{n}, \quad K_{rr} = 1, \, K_{rs} = -\frac{1}{n} \, (r\neq s), \quad S = L = \frac{1}{n} I.
\end{align}
This gives 
\begin{align}
\begin{dcases}
    K' = K -cc^T = (1+\frac{1}{n})I-|u\rangle \langle u|-\frac{1}{n}|u\rangle \langle u| = \qty(1+\frac{1}{n})P_\perp \\
    L' = L - cc^T = \frac{1}{n}I-\frac{1}{n}|u\rangle \langle u| = \frac{1}{n} P_\perp \\
    S' = S - cc^T = \frac{1}{n}I-\frac{1}{n}|u\rangle \langle u| = \frac{1}{n} P_\perp, 
\end{dcases}
\end{align}
where $|u\rangle = \frac{1}{\sqrt{n}} \sum_{r=1}|r\rangle$, and $P_\perp = I -|u\rangle \langle u|$. 
This gives 
\begin{align}
\begin{dcases}
    A_0 = \qty(1+\frac{\Delta}{n})P_\perp\\
    A_+ = P_\perp\\
    A_- = P_\perp.  
\end{dcases}
\end{align}
The choice 
\begin{align}
    \hat{R} = \frac{P_\perp}{n-1}
\end{align}
gives $f_{\mathrm{root}}(A_0,\hat{R}) = \sqrt{\Tr A_0 \Tr \hat{R}},\, f_{\mathrm{root}}(A_+, \hat{R}) = \sqrt{\Tr A_+ \Tr \hat{R}}, \,f_{\mathrm{root}}(A_-, \hat{R}) = \sqrt{\Tr A_- \Tr \hat{R}}$ so the Bény-Oreshkov argument shows that this entanglement-assisted $\mathrm{SU}(d)$ covariant code saturates the upper bound. The Bény–Oreshkov duality thus guarantees the existence of a recovery channel attaining the sequential upper bound to the first order. 
\end{proof}

\subsection{Quantum circuit for noisy unitary purification}
We now provide a concrete encoding and decoding circuit that implements the noisy unitary purification protocol saturating the bound in Theorem 2. The circuit is independent of $p$, and thus the protocol does not require the knowledge of the noise strength in advance. 
\subsubsection{Registers and Young diagrams}
Fix integers
\begin{align}
 d\ge2,\qquad k\ge1,\qquad n=kd+1.
\end{align}
Consider the Bratteli diagram of the partially transposed permutation matrix algebra $\mathcal{A}_{n,1}^d$. Let the $n-1, n, n+1$-th level irreps be 
\begin{align}
    &\lambda^- = (k^d),\, \mu^- = (k+1, k^{d-2}, k-1) \nonumber\\
    &\alpha = (k+1, k^{d-1}), \, \beta = (k+2, k^{d-2}, k-1), \, \gamma = (k+1, k+1, k^{d-3}, k-1) \nonumber\\
    &\hat{\lambda} = (\lambda^-, \phi), \, \hat{\mu} = (\mu^-, \phi).
\end{align}
The Weyl dimension formula gives
\begin{align}
    d_\alpha = d, \, d_\beta = \frac{(d-1)d(d+2)}{2}, \, d_\gamma = \frac{(d-2)d(d+1)}{2}. 
\end{align}
We note that for $d=2$, the irrep $\gamma$ does not exist, but $d_\gamma=0$ so the subsequent discussions are not affected.  

\begin{figure}[h]
\centering
\includegraphics[width=0.45\linewidth]
{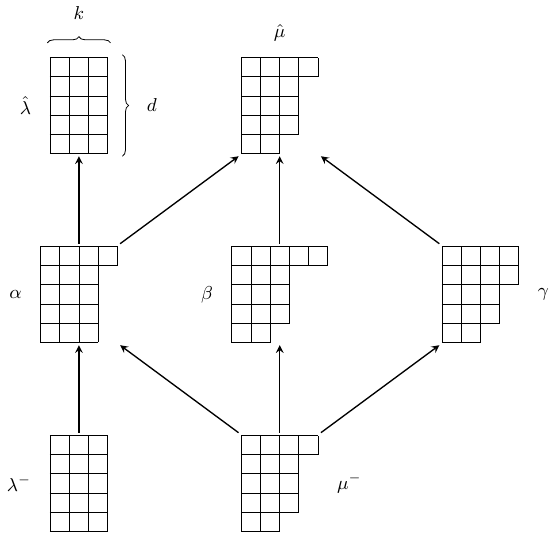}
\caption{The path of Bratteli diagram $\mathcal{A}_{n,1}^d$, used in the analysis of the fidelity. The empty diagram $\phi$ is omitted from the figure.}
\label{fig:young-branching-diagram}
\end{figure}

\subsubsection{Encoder}
Observe that $\mathcal{U}_\alpha \otimes \bar{\Box} = \mathcal{U}_{\hat{\lambda}} \oplus \mathcal{U}_{\hat{\mu}}$. We define the encoding isometry by 
\begin{align}
    E_{n,d}|\psi\rangle &= \frac{1}{\sqrt{m_\alpha}} \sum_{s_\alpha} U_{\mathrm{Sch}}^{(n)\dagger}[|s_\alpha \rangle \otimes |\alpha \rangle \otimes \langle \alpha| \mathrm{CG}_{\lambda^-, \Box}(|u_{\lambda^-} \rangle \otimes |\psi\rangle)] \otimes |s_{\alpha} \to \hat{\lambda} \rangle \nonumber\\
    &= \sqrt{\frac{d}{m_\alpha}} \sum_{s_\alpha} U_{\mathrm{Sch}}^{(n)\dagger}[|s_\alpha \rangle \otimes |\alpha \rangle \otimes (\mathbb{I}_{\mathcal{U}_\alpha}\otimes \langle \overline{\psi}|) \mathrm{CG}_{\alpha, \bar{\Box}}^\dagger|u_\alpha \rangle] \otimes |s_{\alpha} \to \hat{\lambda} \rangle \nonumber\\
    &= \sqrt{\frac{d}{m_\alpha}} \sum_{s_\alpha} (\mathbb{I}_{\mathcal{U}_\alpha}\otimes \langle \overline{\psi}|) U_{\mathrm{mSch}}^{(n,1)\dagger}[|s_{\alpha} \to \hat{\lambda} \rangle \otimes |\hat{\lambda} \rangle \otimes |u_{\hat{\lambda}}\rangle] \otimes |s_\alpha \to \hat{\lambda} \rangle.
\end{align}

\subsubsection{Decoder}
Let $\Path{\delta}$ denote the set of paths in the Bratteli diagram of $\mathcal{A}_{n,1}^d$ ending with the vertex $\delta$ (See Appendix~\ref{ap:preliminaries} for details). 
\begin{align}
    R_\delta^\eta := \sum_{s_\delta \in \Path{\delta}} |s_\delta \to \eta \to \hat{\mu}\rangle \langle s_\delta \to \alpha \to \hat{\lambda}|, 
\end{align}
where $\delta$ runs over the $n-1$-th level, and $\eta$ runs over the $n$-th level. Using this definition, let us further define 
\begin{align}
    B_n = R_{\lambda^-}^\alpha - \frac{1}{\Delta}R_{\mu^-}^\alpha + \frac{\sqrt{dd_\beta}}{\Delta} R_{\mu^-}^\beta + \frac{\sqrt{dd_\gamma}}{\Delta} R_{\mu^-}^\gamma. 
\end{align}
and 
\begin{align}
    B_i = \rho_{\hat{\mu}}(\pi_i^{-1})B_n \rho_{\hat{\lambda}}(\pi_i), 
\end{align}
where $\pi_i=(i,n)$ is the transposition. Define 
\begin{align}
    |\Omega_\alpha \rangle = \frac{1}{\sqrt{m_\alpha}}\sum_{s_\alpha} |s_\alpha \to \hat{\lambda} \rangle \otimes |s_\alpha \rangle 
\end{align}
and 
\begin{align}
    |h_{i,\eta} \rangle = \sum_{s_\eta \in \Path{\eta}} \sqrt{\frac{d \Delta}{m_\alpha d_\eta}} B_i^\dagger |s_\eta \to \hat{\mu} \rangle \otimes |s_\eta\rangle  \quad \eta \in \{ \alpha, \beta, \gamma \}
\end{align}
\begin{align}
    &|g_{i, \alpha}\rangle = (I-|\Omega_\alpha \rangle \langle \Omega_\alpha|)|h_{i, \alpha}\rangle \\
    &|g_{i, \eta}\rangle = |h_{i, \eta}\rangle \quad \eta \in \{ \beta, \gamma \}.
\end{align}
We define the partial isometry 
\begin{align}
    &V_\alpha = |0, \lambda^- \rangle \langle \Omega_\alpha| + \frac{1}{\sqrt{1+\frac{\Delta}{n}}}\sum_i | i, \mu^- \rangle \langle g_{i, \alpha}|\\
    & V_\eta = \sum_i |i, \mu^- \rangle \langle g_{i, \eta}| \quad \eta \in \{\beta, \gamma \}. 
\end{align}
We can then complete $V_\eta$ to produce the full isometry $\widetilde{V}_\eta$. The decoding operation is then 
\begin{align}
    D_{n,d}|o_1, \cdots o_n \rangle |\psi\rangle_E = \mathrm{CG}^\dagger \qty[\bigoplus_\eta \widetilde{V}_\eta](U_{\mathrm{Sch}}^{(n)}\otimes \mathbb{I}_E)|o_1, \cdots o_n \rangle |\psi\rangle_E. 
\end{align}

\subsubsection{Evaluation of the fidelity}
\noindent Let us compute the fidelity gate-by-gate. 
\begin{align}
    T_{a,i}E_{n,d}|\psi\rangle &= \sqrt{\frac{d}{m_\alpha}}\sum_{s_\alpha} (\mathbb{I}_d^{\otimes n}\otimes \langle \psi|)(T_{a,i}\otimes \mathbb{I}_P)U_{\mathrm{mSch}}^{(n,1)^\dagger}|s_\alpha \to \hat{\lambda}, \hat{\lambda}, u_{\hat{\lambda}}\rangle \otimes |s_\alpha \to \hat{\lambda}\rangle_E \nonumber\\ 
    &= \sqrt{\frac{d}{m_\alpha}}\sum_{s_\alpha}\sum_\eta \sum_{t_\eta} (\mathbb{I}_d^{\otimes n}\otimes \langle \psi|)\langle t_\eta \to \hat{\mu}|B_i|s_\alpha \to \hat{\lambda}\rangle U_{\mathrm{mSch}}^{(n,1)^\dagger}(|t_\eta \to \hat{\mu}\rangle \otimes |\hat{\mu}\rangle \otimes |a\rangle) \otimes |s_\alpha \to \hat{\lambda}\rangle_E \nonumber\\ 
     &= \sqrt{\frac{d}{m_\alpha}}\sum_{s_\alpha}\sum_\eta \sum_{t_\eta} (\mathbb{I}_d^{\otimes n}\otimes \langle \psi|)\langle t_\eta \to \hat{\mu}|B_i|s_\alpha \to \hat{\lambda}\rangle (U_{\mathrm{Sch}}^{(n)^\dagger}\otimes \mathbb{I}_{\bar{P}})(|t_\eta\rangle \otimes |\eta\rangle \otimes \mathrm{CG}_{\eta \bar{\Box}}^\dagger|a\rangle) \otimes |s_\alpha \to \hat{\lambda}\rangle_E \nonumber\\ 
     &= \sum_{s_\alpha}\sum_\eta \sum_{t_\eta} \sqrt{\frac{d\Delta}{m_\alpha d_\eta}} \langle t_\eta \to \hat{\mu}|B_i|s_\alpha \to \hat{\lambda}\rangle (U_{\mathrm{Sch}}^{(n)^\dagger}\otimes \mathbb{I}_{\bar{P}})(|t_\eta\rangle \otimes |\eta\rangle \otimes \mathrm{CG}_{\mu^- \Box}(|a\rangle \otimes |\psi\rangle)) \otimes |s_\alpha \to \hat{\lambda}\rangle_E. 
\end{align}
Thus, we have 
\begin{align}
    U_{\mathrm
    {Sch}}^{(n)}T_{a,i}E_{n,d}|\psi\rangle= \sum_\eta |h_{i,\eta}\rangle \otimes |\eta \rangle \otimes \mathrm{CG}_{\mu^- \Box}(|a\rangle \otimes |\psi\rangle), 
\end{align}
where
\begin{align}
    |h_{i,\eta} \rangle =  \sum_{t_\eta \in \Path{\eta}} \sum_{s_\alpha}\sqrt{\frac{d \Delta}{m_\alpha d_\eta}}\langle t_\eta \to \hat{\mu}|B_i|s_\alpha \to \hat{\lambda}\rangle |t_\eta \to \hat{\mu} \rangle \otimes |s_\alpha \to \hat{\lambda}\rangle_E  \quad \eta \in \{ \alpha, \beta, \gamma \}. 
\end{align}
We have 
\begin{align}
    &|h_{i,\alpha}\rangle = |g_{i,\alpha}\rangle + \frac{\sqrt{\Delta}}{n} |\Omega_\alpha \rangle\\
    &|h_{i,\eta} \rangle = |g_{i,\eta}\rangle  \quad \eta \in \{\beta, \gamma \}. 
\end{align}
One can show the following: 
\begin{prop}\label{prop:goverlap}
\begin{align}
    \langle g_{i,\eta}| g_{j,\eta} \rangle = a_\eta (\delta_{ij}-\frac{1}{n}), 
\end{align}
where 
\begin{align}
    a_\alpha = 1 + \frac{\Delta}{n}, \quad a_\beta = a_\gamma = 1. 
\end{align}
\end{prop}
\noindent The proof is deferred to the end of this section. Using the above proposition, one sees that the isometry 
\begin{align}
    &V_\alpha = |0, \lambda^- \rangle \langle \Omega_\alpha| + \frac{1}{\sqrt{1+\frac{\Delta}{n}}}\sum_i | i, \mu^- \rangle \langle g_{i, \alpha}|\\
    & V_\eta = \sum_i |i, \mu^- \rangle \langle g_{i, \eta}| \quad \eta \in \{\beta, \gamma \}. 
\end{align}
 acts as 
 \begin{align}
     &V_\alpha|h_{i,\alpha} \rangle = \frac{\sqrt{\Delta}}{n} |0, \lambda^-\rangle + \sqrt{1+\frac{\Delta}{n}} |e_i, \mu^-\rangle \\
     &V_\eta |h_{i,\eta}\rangle = |e_i, \mu^-\rangle \quad \eta \in \{\beta, \gamma \}, 
 \end{align}
where $e_i$ is the basis for the standard representation of $\mathfrak{S}_n$, $|e_i\rangle = |i\rangle - \frac{1}{n}\sum_j |j\rangle$. Therefore, the full action of the encoding, noise, decoding is given in the following form: 
\begin{align}
    D_{n,d} T_{a,i} E_{n,d}|\psi\rangle &=\frac{\sqrt{\Delta}}{n}|0, \lambda^-\rangle \otimes \mathrm{CG}_{\lambda^- \Box}^\dagger[|\alpha \rangle \otimes \mathrm{CG}_{\mu^- \Box}(|a\rangle \otimes |\psi\rangle)] + |e_i, \mu^-\rangle \otimes \qty[\sqrt{1+\frac{\Delta}{n}} \Pi_\alpha+ \Pi_\beta + \Pi_\gamma](|a\rangle \otimes |\psi\rangle) \nonumber\\
    &= |0, \lambda^-\rangle \otimes \frac{1}{n} |u_{\lambda^-}\rangle \otimes T_a |\psi\rangle + |e_i, \mu^-\rangle \otimes \qty[\sqrt{1+\frac{\Delta}{n}} \Pi_\alpha+ \Pi_\beta + \Pi_\gamma](|a\rangle \otimes |\psi\rangle)
\end{align}
Thus, the fidelity can be calculated as 
\begin{align}
    f_{n,d}(p) = 1- np(1-f_1)+O(p^2), 
\end{align}
where 
\begin{align}
    f_1 = \frac{1}{d^2} + \frac{1}{d^4} \cdot \frac{n-1}{n} \cdot \frac{\qty[\sqrt{1+\frac{\Delta}{n}} d_\alpha+ d_\beta + d_\gamma]^2}{\Delta}.
\end{align}
\subsubsection{Proof of Proposition~\ref{prop:goverlap}}
\begin{proof}
From the definition of $B_n$, one has $B_{\sigma(i)} = \rho_{\hat{\mu}}(\sigma^{-1})B_i \rho_{\hat{\lambda}}(\sigma)$ for $\sigma \in \mathfrak{S}_{n-1}\subset \mathfrak{S}_n$. Using the branching isometry $\mathcal{J}_\eta: \mathcal{S}_\eta \to \mathcal{S}_{\hat{\mu}}$ defined by $\mathcal{J}_\eta|s_\eta\rangle \to |s_\eta \to \hat{\mu}\rangle$, one has the expression
\begin{align}
    |h_{i,\eta}\rangle = \sum_{s_\eta} \sqrt{\frac{d \Delta}{m_\alpha d_\eta}} B_i^\dagger J_\eta|s_\eta\rangle_E \otimes |s_\eta \rangle. 
\end{align}
Since $\rho_{\hat{\mu}}(\sigma^{-1})J_\eta = J_\eta \rho_\eta(\sigma^{-1})$, and 
\begin{align*}
    \sum_{s_\eta} [\rho_\eta(\sigma)\otimes \overline{\rho_\eta(\sigma)}]|s_\eta\rangle \otimes |s_\eta\rangle  = \sum_{s_\eta} |s_\eta\rangle \otimes |s_\eta\rangle, 
\end{align*}
one gets 
\begin{align}
    |h_{\sigma(i),\eta} \rangle = [\rho_{\hat{\lambda}}(\sigma) \otimes \rho_{\eta}(\sigma)] | h_{i,\eta}\rangle 
\end{align}
Now, the restriction of $\mathcal{S}_{\hat{\lambda}}$ to the $n$-th level is $\mathcal{S}_\alpha$, and 
\begin{align*}
    |\Omega_\alpha \rangle = \frac{1}{\sqrt{m_\alpha}}\sum_{s_\alpha} |s_\alpha \to \hat{\lambda} \rangle \otimes |s_\alpha \rangle 
\end{align*}
is the unique vector that is invariant under $\rho_{\hat{\lambda}}(\sigma) \otimes \rho_{\eta}(\sigma)$. Thus, one also gets 
\begin{align}\label{eq:ginvariance}
    |g_{\sigma(i),\eta} \rangle = [\rho_{\hat{\lambda}}(\sigma) \otimes \rho_{\eta}(\sigma)] |g_{i,\eta}\rangle. 
\end{align}
Next, define 
\begin{align}
    \Pi(\sigma) = \sum_i |\sigma(i)\rangle \langle i| 
\end{align}
and 
\begin{align}
    A_{\eta} := \sum_i |g_{i,\eta}\rangle \langle i|, \quad G_\eta = A_\eta^\dagger A_\eta. 
\end{align}
Using $A_\eta$, one can compactly write the invariance in Eq.~\eqref{eq:ginvariance} as 
\begin{align}
    A_\eta \Pi(\sigma) = [\rho_{\hat{\lambda}}(\sigma) \otimes \rho_{\eta}(\sigma)] A_\eta. 
\end{align}
Therefore, for any $\sigma \in \mathfrak{S}_n$, one has 
\begin{align}
    [G_\eta, \Pi(\sigma)] = 0. 
\end{align}
The permutation representation decomposes multiplicity-freely as $\Pi \simeq \mathrm{triv} \oplus \mathrm{std}$, with the trivial space spanned by $|+\rangle = \frac{1}{\sqrt{n}}\sum_i |i\rangle$ and the standard space is the orthogonal complement of $|+\rangle$. Therefore, 
\begin{align}
    G_\eta = a_\eta^{\mathrm{triv}}|+\rangle \langle +| + a_\eta P_\perp, \quad a_\eta^{\mathrm{triv}}, a_\eta \geq 0, 
\end{align}
where $P_\perp  = \mathbb{I}-|+\rangle \langle +|$. 

We first determine the trivial eigenvalue. From Eq.~\eqref{eq:ginvariance}, $A_\eta\ket{+}$ is invariant under $\rho_{\hat{\lambda}}\otimes\rho_\eta$.
The restriction of $\rho_{\hat{\lambda}}$ from the $n+1$-level to the $n$-th level contains only $\rho_\alpha$, and hence Schur's lemma gives
\begin{align}
  A_\eta\ket{+}
  \propto
  \begin{cases}
    \ket{\Omega_{\lambda\alpha}} & (\eta=\alpha)\\
    0 & (\eta = \beta, \gamma)
  \end{cases}.
  \label{eq:direct-trivial-range}
\end{align}
However, when $\eta = \alpha$, every $|g_{i,\alpha}\rangle$ is orthogonal to $|\Omega_{\alpha}\rangle$ by definition. Therefore, $a_\eta^{\mathrm{triv}}=0$ and it remains to compute $a_\eta$. First, the hook length formula gives 
\begin{align}
    \frac{m_{\lambda^-}}{m_\alpha} = \frac{1}{n} \prod_{j=1}^k \frac{k-j+d+1}{k-j+d} = \frac{k+d}{dn} = \frac{n+\Delta}{d^2 n}, 
\end{align}
and 
\begin{align}
    \frac{m_{\mu^-}}{m_\alpha} = 1- \frac{m_{\lambda^-}}{m_\alpha} = \frac{(n-1)\Delta}{nd^2}. 
\end{align}
Now, consider the $\eta=\beta, \gamma$ branch. Since
\begin{align}
    \langle g_{i,\eta}|g_{i,\eta} \rangle = \langle g_{n,\eta}|g_{n,\eta} \rangle = \langle h_{n,\eta}|h_{n,\eta} \rangle= \frac{d \Delta}{m_\alpha d_\eta} \qty(\frac{\sqrt{d d_\eta}}{\Delta})^2 m_{\mu^-} = \frac{d^2}{\Delta} \frac{m_{\mu^-}}{m_\alpha} = \frac{n-1}{n}
\end{align}
we have 
\begin{align}
    (n-1)a_\eta = \Tr G_\eta \Leftrightarrow a_\eta = 1 \quad \eta \in \{\beta, \gamma\}. 
\end{align}
Similarly, we have 
\begin{align}
    \langle h_{i,\alpha}|h_{i,\alpha} \rangle = \langle h_{n,\alpha}|h_{n,\alpha} \rangle = \frac{d \Delta}{m_\alpha d_\alpha} \qty(m_{\lambda^-}+\frac{m_{\mu^-}}{\Delta^2}) 
    = {\Delta} \frac{m_{\lambda^-}}{m_\alpha}+\frac{1}{\Delta}\frac{m_{\mu^-}}{m_\alpha} = 1+ \frac{\Delta-1}{n}.
\end{align}
The overlap between $|\Omega_\alpha\rangle$ is
\begin{align}
    \langle \Omega_\alpha|h_{i,\alpha} \rangle = \langle \Omega_\alpha|h_{n,\alpha} \rangle = \frac{1}{\sqrt{m_\alpha}} \sqrt{\frac{d \Delta}{m_\alpha d_\alpha}} \qty(m_{\lambda^-}-\frac{m_{\mu^-}}{\Delta})= \sqrt{\Delta} \frac{m_{\lambda^-}}{m_\alpha}-\frac{1}{\sqrt{\Delta}}\frac{m_{\mu^-}}{m_\alpha} = \frac{\sqrt{\Delta}}{n}.
\end{align}
Therefore, 
\begin{align}
    |h_{i,\alpha} \rangle = |g_{i,\alpha} \rangle + \frac{\sqrt{\Delta}}{n} |\Omega_\alpha\rangle
\end{align}
and 
\begin{align}
    \langle g_{i,\alpha}|g_{i,\alpha}\rangle = \langle h_{i,\alpha}|(\mathbb{I}-|\Omega_\alpha \rangle \langle \Omega_\alpha|)|h_{i,\alpha}\rangle = \langle h_{i,\alpha}|h_{i,\alpha} \rangle - |\langle h_{i,\alpha}|\Omega_\alpha\rangle|^2 = \frac{n-1}{n}\qty(1+\frac{\Delta}{n}).
\end{align}
This implies 
\begin{align}
    a_\alpha = 1 + \frac{\Delta}{n}. 
\end{align}
\end{proof}

\section{Proof of Theorem 3}
Here we prove Theorem 3. Analogous to noisy unitary purification, we first prove the upper bound on the achievable fidelity within sequential strategies for all $n\geq d-1$. 
We then prove a complementary lower bound that provides a concrete $\mathrm{SU}(d)$-covariant parallel protocol attaining the fidelity upper bound. 

\begin{thm}[Asymptotically optimal fidelity for noisy unitary conjugation]
The leading order coefficient $C'_{n,d}$ satisfies 
\begin{align}
C'_{n,d} \geq  \frac{1}{d^2}\qty[n(d^2-1) - \frac{\mathfrak{B}_{n,d}^2}{d^2-1}], 
\end{align}
where 
\begin{align}
\mathfrak{B}_{n,d} = \sqrt{n+2-\frac{d^2-1}{n}} +\frac{(d-1)(d+2)}{2}\sqrt{n-d+1} + \frac{(d-2)(d+1)}{2}\sqrt{n+d+1}
\end{align}
for all $n\geq d-1, \, n\in \mathbb{N}$. If $n=kd-1, k\in \mathbb{N}$, there is an $\mathrm{SU}(d)$-covariant parallel strategy that saturates the upper bound. Consequently, the optimal fidelity of the noisy unitary conjugation problem has the following form: 
\begin{align}
    g_{n,d}(p) = 1- \qty[\frac{(d^2-1)(d^2+2)}{4d^2n} + o(n^{-1})]p + o(p).
\end{align}
\end{thm}

\subsection{Upper bound on the achievable fidelity}
\begin{lm}[Sequential upper bound] 
For all $n\geq d-1$, the first-order infidelity coefficient satisfies 
\begin{align}
C'_{n,d} \geq  \frac{1}{d^2}\qty[n(d^2-1) - \frac{\mathfrak{B}_{n,d}^2}{d^2-1}], 
\end{align}
where 
\begin{align}
\mathfrak{B}_{n,d} = \sqrt{n+2-\frac{d^2-1}{n}} +\frac{(d-1)(d+2)}{2}\sqrt{n-d+1} + \frac{(d-2)(d+1)}{2}\sqrt{n+d+1}. 
\end{align}
\end{lm}
\begin{proof}
The performance operator satisfies 
\begin{align}
    [A^{\otimes {n+1}} \otimes B^{\otimes n+1}, \Omega_p] = 0 \quad \forall A,B\in \mathrm{SU}(d). 
\end{align}
Thus, we can assume without loss of generality that $\Xi$ is twirled.  
\begin{align}
    [A^{\otimes {n+1}} \otimes B^{\otimes n+1}, \mathcal{J}_\Xi] = 0 \quad \forall A,B\in \mathrm{SU}(d). 
\end{align}
Now, consider the Stinespring dilation of a noisy conjugation superchannel $\Xi$ and denote it $\Xi^{\mathrm{Pur}}$. Here, each memory operation and the final decoder of $\Xi$ is Stinespring-dilated to give $\Xi^{\mathrm{Pur}}$. Now, define $V_0 = \widetilde{\Xi}_{\mathrm{Pur}}(I, I, \cdots I), \, V_{a,r} = \widetilde{\Xi}_{\mathrm{Pur}}(I, I, \cdots T_a, \cdots I)$. Using covariance, one can derive 
\begin{align}
    &C^*(V_0^\dagger V_{a,r})C^T = \sum_b R_{ba}(C)V_0^\dagger V_{b,r}\\
    &C^*(V_{a,r}^\dagger V_{b,s})C^T = \sum_{c,e} R_{ca}(C)R_{eb}(C)V_{c,r}^\dagger V_{e,s}, 
\end{align}
where $R(C)$ is the adjoint representation 
\begin{align}
    CT_aC^\dagger = \sum_{b}R_{ba}(C)T_b. 
\end{align}
Define $\Phi_r: \overline{\mathrm{Adj}} \to \mathcal{L}(\mathbb{C}^d), \, \Phi_r(T_a^*) = V_0^\dagger V_{a,r}$. This is an $\mathrm{SU}(d)$-intertwiner, and we have $\mathcal{L}(\mathbb{C}^d) \simeq \mathcal{L}(\bar{\mathbb{C}}^d) \simeq \mathbb{I}\oplus \overline{\mathrm{Adj}}$, $\dim \mathrm{Hom}(\overline{\mathrm{Adj}}, \mathcal{L}(\mathbb{C}^d)) = 1$, so 
\begin{align}
    V_0^\dagger V_{a,r} = c_r T_a^*
\end{align}
for some $c_r \in \mathbb{C}$. $V_{a,r}$ is obtained by inserting a Hermitian traceless operator to the $r$-th slot of $V_0$, so we have $c_r \in \mathbb{R}$. We can similarly define $\Psi_{r,s}: \overline{\mathrm{Adj}}\otimes \overline{\mathrm{Adj}} \to \mathcal{L}(\mathbb{C}^d), \, \Psi_{r,s}(T_a^* \otimes T_b^*) =  V_{a,r}^\dagger V_{b,s}$, and $\Psi_{r,s}$ is an $\mathrm{SU}(d)$-intertwiner. Now, define the tensors 
\begin{align}
    T_a T_b = \delta_{ab} \mathbb{I} +(d_{abc}+if_{abc})T_c, 
\end{align}
where $d_{abc}$ is totally symmetric and $f_{abc}$ is totally antisymmetric. This gives 
\begin{align}
    T_a^* T_b^* = \delta_{ab} \mathbb{I} +(d_{abc}-if_{abc})T_c^*, 
\end{align}
and so we have 
\begin{align}
    V_{a,r}^\dagger V_{b,s} = K_{rs}\delta_{ab}\mathbb{I} + S_{rs}d_{abc}T_c^*+iL_{rs}f_{abc}T_c^*, 
\end{align}
with $K_{rs}^* = K_{sr}, \, S_{rs}^* = S_{sr},\, L_{rs}^* = L_{sr}$, derived from the fact $(V_{a,r}^\dagger V_{b,s})^\dagger = V_{b,s}^\dagger V_{a,r}$. Thus, $K,S,L$ are Hermitian matrices. 

Now, it suffices to optimize over the noiseless face $\Xi(\mathcal{U}^{\otimes n})=\mathcal{U}^*$ (see Corollary~\ref{cor:noiselessconjugate}), so 
we have $V_{\mathrm{Pur}}(U, \cdots U) = U^* \otimes |\eta_\theta\rangle$. Choosing $U=e^{iT_a\theta}$ and differentiate with respect to $\theta$ gives 
\begin{align}
    \sum_r c_r = -1. 
\end{align}
We also have 
\begin{align}
    V^\dagger_{a,r} V_{b,r} = \delta_{ab} \mathbb{I} + (d_{abc}+if_{abc})c_rT_c^*. 
\end{align}
This gives the constraint
\begin{align}
    K_{rr} = 1, \quad S_{rr} = L_{rr} = c_r. 
\end{align}
Let us use the Bény-Oreshkov condition to obtain the upper bound on the fidelity. 
\begin{align}
    \max_\mathcal{R}f_{\mathrm{Choi}}(\mathcal{R}\circ \Xi(\mathcal{N}_{\mathcal{U},p}^{\otimes n}), \mathcal{U}^*) &\leq \max_{\mathcal{R}'} f_{\mathrm{Choi}}(\mathcal{R}'\circ \Xi_{\mathrm{Pur}}(\mathcal{N}_{\mathcal{U},p}^{\otimes n}), \mathcal{U}^*) \nonumber\\
    &= \max_{\mathcal{\sigma_E}} f_{\mathrm{Choi}}( [\Xi_{\mathrm{Pur}}(\mathcal{N}_{\mathrm{id},p}^{\otimes n})]^c, \mathcal{T}_{\sigma_E}) \nonumber\\
    &= \max_{\mathcal{\sigma_E}} f\qty(\frac{1}{d}\sum_{\alpha, \beta}\sqrt{w_\alpha w_\beta}|\alpha\rangle \langle \beta|_E\otimes (V_\beta^\dagger V_\alpha)^T_R, \sigma_E\otimes \frac{\mathbb{I}_R}{d}) \nonumber\\
    &\leq \max_{\mathcal{\sigma_E}} f\qty(\frac{1}{d}\sum_{\alpha, \beta}\sqrt{w_\alpha w_\beta}\mathcal{P}(|\alpha\rangle \langle \beta|)_E\otimes (V_\beta^\dagger V_\alpha)^T_R, \mathcal{P}(\sigma_E)\otimes \frac{\mathbb{I}_R}{d}) \nonumber\\
    &= \max_{0 \leq q \leq 1} \left[\sqrt{q\max_{\sigma_{\leq 1}} f\qty(\frac{1}{d}\sum_{\alpha, \beta}\sqrt{w_\alpha w_\beta}(\Pi|\alpha\rangle \langle \beta|\Pi)_E\otimes (V_\beta^\dagger V_\alpha)^T_R, \sigma_{\leq 1}\otimes \frac{\mathbb{I}_R}{d})}\right.\nonumber\\
    &\left. + \sqrt{(1-q)\max_{\sigma_{\geq 2}} f\qty(\frac{1}{d}\sum_{\alpha, \beta}\sqrt{w_\alpha w_\beta}(\mathbb{I}-\Pi)|\alpha\rangle \langle \beta|(\mathbb{I}-\Pi)_E\otimes (V_\beta^\dagger V_\alpha)^T_R, \sigma_{\geq 2} \otimes \frac{\mathbb{I}_R}{d})}\right]^2 \nonumber\\
    &\leq \max_{\sigma_{\leq 1}} f\qty(\frac{1}{d}\sum_{\alpha, \beta}\sqrt{w_\alpha w_\beta}(\Pi|\alpha\rangle \langle \beta|\Pi)_E\otimes (V_\beta^\dagger V_\alpha)^T_R, \sigma_{\leq 1}\otimes \frac{\mathbb{I}_R}{d}) 
    \nonumber\\
    &+ \Tr\qty[\frac{1}{d}\sum_{\alpha, \beta}\sqrt{w_\alpha w_\beta}(\mathbb{I}-\Pi)|\alpha\rangle \langle \beta|(\mathbb{I}-\Pi)_E\otimes (V_\beta^\dagger V_\alpha)^T_R]\nonumber\\
    &= \max_{\sigma_E} f\qty(\frac{1}{d}\sum_{\alpha, \beta}\sqrt{w_\alpha w_\beta}(\Pi|\alpha\rangle \langle \beta|\Pi)_E\otimes (V_\beta^\dagger V_\alpha)^T_R, \sigma_E\otimes \frac{\mathbb{I}_R}{d}) + O_{n,d}(p^2)
\end{align}
where $\alpha, \beta$ denote the error labels, $w_\alpha, w_\beta$ are the corresponding weights, $\sigma_{\leq 1} := \frac{\Pi(\sigma_E)\Pi}{\Tr [\Pi(\sigma_E)]}, \sigma_{\geq 2} := \frac{(\mathbb{I}-\Pi)(\sigma_E)(\mathbb{I}-\Pi)}{{\Tr [(\mathbb{I}-\Pi)(\sigma_E)]}}, q= \Tr [\Pi(\sigma_E)]$, $R$ denotes the reference system, $(\cdot)^c$ denotes the complementary channel, and $\mathcal{P}$ is the pinching channel $\mathcal{P}(\cdot) = \Pi(\cdot) \Pi+(\mathbb{I}-\Pi)(\cdot)(\mathbb{I}-\Pi)$ with $\Pi$ being the projector onto the zero- and one-error sector.

Next, define $Q:=(1-\Delta t)^n, \epsilon:= t(1-\Delta t)^{n-1}$, where $\Delta:=d^2-1$, $t:=\frac{p}{d^2}$.
\begin{align}
    \rho_{ER} &:= \frac{1}{d}\left[Q|0\rangle \langle 0|_E\otimes \mathbb{I}_R+\sum_{a,r}c_r\sqrt{\epsilon Q}(|0\rangle \langle a,r|+|a,r \rangle \langle 0|)\otimes (T_a^*)^T+\sum_{a,r,b,s}\epsilon |a,r\rangle \langle b,s|\otimes (V_{b,s}^\dagger V_{a,r})^T\right]. 
\end{align}
Define
\begin{align}
    J|\bar{\psi}\rangle = \sum_a |a\rangle \otimes T_a|\bar{\psi} \rangle 
\end{align}
with $J^\dagger J = \Delta \mathbb{I}$. $U_0=J/\sqrt{\Delta}$ is an isometric embedding. Define 
$c:= \sum_r c_r |r\rangle \langle 0|$, and we have 
\begin{align}
    \rho_{ER} &:= \frac{1}{d}\left[Q|0\rangle \langle 0|\otimes \mathbb{I}_R+\sqrt{\epsilon Q}(c\otimes J + c^\dagger \otimes J^\dagger)+\sum_{a,r,b,s}\epsilon |a,r\rangle \langle b,s|\otimes (V_{b,s}^\dagger V_{a,r})^T\right]. 
\end{align}
Now, define 
\begin{align}
    W_{a,r}:= V_{a,r}-c_rV_0 T_a^*. 
\end{align}
This gives 
\begin{align}
\begin{dcases}
    V_0^\dagger W_{a,r} = 0\\
    V_{a,r}^\dagger V_{b,s} = (W_{a,r}+c_rV_0 T_a^*)^\dagger (W_{b,s}+c_sV_0 T_b^*) = W_{a,r}^\dagger W_{b,s} + c_rc_sT_a^*T_b^*. 
\end{dcases}
\end{align}
Therefore, 
\begin{align}
    W_{a,r}^\dagger W_{b,s} = K'_{rs}\delta_{ab} + (S'_{rs}d_{abc}+iL'_{rs}f_{abc})T_c^*, 
\end{align}
where $K',S',L'$ are Hermitian operators defined by 
\begin{align}
    K' = K -cc^T, \, S'=S-cc^T, \, L'= L+cc^T. 
\end{align}
This allows us to write $\rho_{ER}$ as 
\begin{align}
    \rho_{ER} &= \frac{1}{d}\left[Q|0\rangle \langle 0|\otimes \mathbb{I}_R+\sqrt{\epsilon Q}(c\otimes J + c^\dagger \otimes J^\dagger)+\epsilon cc^\dagger \otimes JJ^\dagger \right.\nonumber\\
    &\left.\quad \quad \quad \quad\quad \quad \quad \quad \quad \quad\quad \quad \quad \quad \quad \quad\quad \quad \quad + \epsilon[(K')^T \otimes I + (S')^T\otimes X_d - (L')^T\otimes X_f]\right], 
\end{align}
where we have defined 
\begin{align}
\begin{dcases}
    X_d = \sum_{a,b,c}d_{abc}|a\rangle \langle b| \otimes T_c\\
    X_f = i\sum_{a,b,c}f_{abc}|a\rangle \langle b| \otimes T_c. 
\end{dcases}
\end{align}
These operators are invariant under $R(U) \otimes U$, where $R(U)$ is the adjoint representation. This leads us to consider the decomposition
\begin{align}
    \mathrm{Adj} \otimes \Box \simeq (\raisebox{0em}{\scalebox{0.4}{\ydiagram{1}}}, \raisebox{0em}{\scalebox{0.4}{\ydiagram{1}}}) \otimes \scalebox{0.4}{\ydiagram{1}} \simeq (\raisebox{0em}{\scalebox{0.4}{\ydiagram{1}}}, \phi)\oplus (\raisebox{0em}{\scalebox{0.4}{\ydiagram{2}}}, \raisebox{0em}{\scalebox{0.4}{\ydiagram{1}}})\oplus (\raisebox{0.3em}{\scalebox{0.4}{\ydiagram{1,1}}}, \raisebox{0em}{\scalebox{0.4}{\ydiagram{1}}}).  
\end{align}
with $\dim R_+ = \frac{(d-1)d(d+2)}{2}, \dim R_- = \frac{(d-2)d(d+1)}{2}$. Therefore, $X_d, X_f$ both act as constants on each irrep $\Box, R_+, R_-$. Let's compute the eigenvalues of $X_d, X_f$. To do this, we construct the total Casimir operator
\begin{align}
    C_{\mathrm{tot}} = C_{\mathrm{Adj}} \otimes \mathbb{I} + \mathbb{I} \otimes C_{\Box} + 2 \sum_a F_a \otimes T_a, 
\end{align}
where $F_a = [T_a, \cdot ]$. We have 
\begin{align}
    X_f = i\sum_{a,b,c}f_{abc}|a\rangle \langle b| \otimes T_c = -\frac{1}{2}\sum_{a,b,c}[F_{c}]_{ab}|a\rangle \langle b| \otimes T_c = -\frac{1}{2} \sum_a F_a \otimes T_a. 
\end{align}
Some computation shows $C_{\mathrm{Adj}} = 2d^2 \mathbb{I}$, and $C_\Box = \Delta \mathbb{I}$. Now, if an irrep is represented by a Young diagram $\lambda=(\lambda_1, \lambda_2, \cdots \lambda_d)$ and $\sum_i \lambda_i = m$, one has
\begin{align}
    C(\lambda) = d^2m + d\sum_{i=1}^{d} \lambda_i(\lambda_i+1-2i)-m^2
\end{align}
Thus, 
\begin{align}
\begin{dcases}
    C_{\mathrm{tot}}(\Box) = d^2-1\\
    C_{\mathrm{tot}}(R_+) = 3d^2+2d-1\\
    C_{\mathrm{tot}}(R_-) = d^2(d-1)+d(-d^2+5d-4)-(d-1)^2 = 3d^2-2d-1. 
\end{dcases}
\end{align}
The eigenvalues of $X_f$ can be computed as 
\begin{align}
\begin{dcases}
\frac{d^2}{2} \quad (\Box, \textrm{multiplicity:}d)\\
-\frac{d}{2} \quad (R_+, \textrm{multiplicity:}\frac{(d-1)d(d+2)}{2})\\
+\frac{d}{2} \quad (R_-, \textrm{multiplicity:}\frac
{(d-2)d(d+1)}{2}).
\end{dcases}
\end{align}
Now, 
\begin{align}
    JJ^\dagger = \sum_{ab} |a\rangle \langle b| \otimes T_aT_b = \sum_{ab}|a\rangle \langle b| \otimes (\delta_{ab}\mathbb{I} + (d_{abc}+if_{abc}) T_c). 
\end{align}
Thus, $1+x_d+x_f = d^2-1$ on $d$ and $1+x_d+x_f = 0$ for $R_+, R_-$. Thus, the eigenvalues of $x_d$ can be computed as 
\begin{align}
\begin{dcases}
\frac{d^2}{2}-2 \quad (\Box, \textrm{multiplicity:}d)\\
\frac{d}{2}-1 \quad (R_+, \textrm{multiplicity:}\frac{(d-1)d(d+2)}{2})\\
-\frac{d}{2}-1 \quad (R_-, \textrm{multiplicity:}\frac{(d-2)d(d+1)}{2}). 
\end{dcases}
\end{align}
Therefore, we have 
\begin{align}
    (K')^T \otimes I + (S')^T\otimes X_d - (L')^T\otimes X_f \simeq (A_0)^T \otimes I_d \oplus (A_+)^T \otimes I_{d_+} \oplus (A_-)^T \otimes I_{d_-}, 
\end{align}
where $d_+=\dim R_+, \,d_- = \dim R_-$ and 
\begin{align}
\begin{dcases}
    A_0 = K' + \qty(\frac{d^2}{2}-2) S' - \frac{d^2}{2} L'\\
    A_+ = K' + \qty(\frac{d}{2}-1) S' + \frac{d}{2} L'\\
    A_- = K' - \qty(\frac{d}{2}+1) S' - \frac{d}{2} L'. 
\end{dcases}
\end{align}
This allows us to write $\rho_{ER}$ as  
\begin{align}
    \rho_{ER} \simeq  \rho_0 \oplus \rho_+ \oplus \rho_-, 
\end{align}
with 
\begin{align}
\begin{dcases}
    \rho_0 = \frac{1}{d}\mqty(Q & \sqrt{\Delta \epsilon Q} c^\dagger\\ \sqrt{\Delta \epsilon Q} c & \epsilon(\Delta cc^\dagger + A_0^T))\otimes I_d,\\
    \rho_+ = \frac{\epsilon}{d} A_+^T\otimes I_{d_+}, \quad \rho_- = \frac{\epsilon}{d}A_-^T\otimes I_{d_-}. 
\end{dcases}
\end{align}
Now, $\Xi_{\mathrm{Pur}}(\mathcal{C} \circ \mathcal{N}_{\mathcal{U},p} \circ \mathcal{C}^\dagger, \mathcal{C} \circ \mathcal{N}_{\mathcal{U},p} \circ \mathcal{C}^\dagger, \cdots, \mathcal{C} \circ \mathcal{N}_{\mathcal{U},p} \circ \mathcal{C}^\dagger) = \mathcal{C}^* \circ \Xi_{\mathrm{Pur}}(\mathcal{N}_{\mathcal{U},p}, \mathcal{N}_{\mathcal{U},p}, \cdots \mathcal{N}_{\mathcal{U},p}) \circ \mathcal{C}^T$, so 
\begin{align}
    \rho_{ER} = ([\boldsymbol{1}\oplus R(C)\otimes \mathbb{I}_n]\otimes C)\rho_{ER}([\boldsymbol{1}\oplus R(C)\otimes \mathbb{I}_n]\otimes C)^\dagger, 
\end{align}
where $R(C)$ is the adjoint representation of $C$,  $\boldsymbol{1}$ is the identity acting on the no-error label $|0\rangle$, and $\mathbb{I}_n$ is the identity operator acting on the position label. On the other hand, 
\begin{align}
    \sigma_E \otimes \frac{\mathbb{I}_R}{d} &\to  ([\boldsymbol{1}\oplus R(C)\otimes \mathbb{I}_n]\otimes C)\qty(\sigma_E \otimes \frac{\mathbb{I}_R}{d})([\boldsymbol{1}\oplus R(C)\otimes \mathbb{I}_n]\otimes C)^\dagger \nonumber\\
    &= ([\boldsymbol{1}\oplus R(C)\otimes \mathbb{I}_n]\otimes \mathbb{I}_d)\qty(\sigma_E \otimes \frac{\mathbb{I}_R}{d})([\boldsymbol{1}\oplus R(C)\otimes \mathbb{I}_n]\otimes \mathbb{I}_d)^\dagger. 
\end{align}
Now, define 
\begin{align}
    \overline{\sigma}_E  = \int \mathrm{d}C [\boldsymbol{1}\oplus R(C)\otimes \mathbb{I}_n]\sigma_E [\boldsymbol{1}\oplus R(C)\otimes \mathbb{I}_n]^\dagger. 
\end{align}
The concavity of the square root fidelity in the second argument shows that 
\begin{align}
    f_{\mathrm{root}}(\rho_{ER}, \overline{\sigma}_E \otimes \frac{
    \mathbb{I}_R}{d}) &\geq \int \mathrm{d}C f_{\mathrm{root}}(\rho_{ER}, [\boldsymbol{1}\oplus R(C)\otimes \mathbb{I}_n]\sigma_E [\boldsymbol{1}\oplus R(C)\otimes \mathbb{I}_n]^\dagger \otimes \frac{
    \mathbb{I}_R}{d})\nonumber\\
    &= f_{\mathrm{root}}(\rho_{ER}, \sigma_E \otimes \frac{\mathbb{I}_R}{d}). 
\end{align} 
Thus, to obtain an upper bound on the fidelity, we can maximize over $\overline{\sigma}_E$. Since 
\begin{align}
    [(\boldsymbol{1}\oplus R(C)\otimes \mathbb{I}_n), \overline{\sigma}_E] = 0, 
\end{align}
$\overline{\sigma}_E$ may be assumed to be of the form 
\begin{align}
    \overline{\sigma}_E = (1-tz)|0\rangle \langle 0|+ \frac{t\mathbb{I}_{\mathrm{Adj}}}{\Delta}\otimes R, \quad R\geq 0, \, \Tr R=z. 
\end{align}
and decompose 
\begin{align}
    \overline{\sigma}_E \otimes \frac{\mathbb{I}_R}{d} &= (1-tz)|0\rangle \langle 0|\otimes \frac{\mathbb{I}_R}{d} + \qty[\frac{t}{d\Delta}R\otimes \mathbb{I}_{d} \oplus \frac{t}{d\Delta}R\otimes \mathbb{I}_{d_+} \oplus \frac{t}{d\Delta}R \otimes \mathbb{I}_{d_-}] \simeq \sigma_0 \oplus \sigma_+ \oplus \sigma_-, 
\end{align}
where 
\begin{align}
\begin{dcases}
     \sigma_0 = \frac{1}{d}\mqty(1-tz & 0\\ 0 & \frac{t}{\Delta}R) \otimes \mathbb{I}_d\\
     \sigma_+ = \frac{t}{d\Delta} R \otimes \mathbb{I}_{d_+}, \\
     \sigma_- = \frac{t}{d\Delta} R \otimes \mathbb{I}_{d_-}. 
\end{dcases}
\end{align}
Thus, we get the expression for the square root fidelity
\begin{align}
    f_{\mathrm{root}}\qty(\rho_{ER}, \sigma_E \otimes \frac{\mathbb{I}_R}{d}) &= f_{\mathrm{root}}(LL^\dagger, DD^\dagger) + f_{\mathrm{root}}(\rho_+, \sigma_+) + f_{\mathrm{root}}(\rho_-, \sigma_-)\nonumber\\
    &= \lVert L^\dagger D \rVert_1 + \frac{d_+}{d}\sqrt{\frac{\epsilon t}{\Delta}}f_{\mathrm{root}}(A_+^T, R) + \frac{d_-}{d}\sqrt{\frac{\epsilon t}{\Delta}}f_{\mathrm{root}}(A_-^T, R)
\end{align}
where
\begin{align}
    L = \mqty(\sqrt{Q} & 0\\\sqrt{\Delta \epsilon} c & \sqrt{\epsilon}(A_0^T)^{1/2}), \quad D = \mqty(\sqrt{1-tz} & 0\\ 0  & \sqrt{\frac{tR}{\Delta}}). 
\end{align}
Substituting, one obtains 
\begin{align}
    f_{\mathrm{root}}\qty(\rho_{ER}, \sigma_E \otimes \frac{\mathbb{I}_R}{d}) =& \left \lVert \mqty(\sqrt{Q(1-tz)} & \sqrt{\epsilon t} c^\dagger R^{1/2}\\ 0 & \sqrt{\frac{\epsilon t}{\Delta}} (A_0^T)^{1/2}R^{1/2}) \right \rVert_1\nonumber\\
    &+ \frac{d_+}{d}\sqrt{\frac{\epsilon t}{\Delta}}f_{\mathrm{root}}(A_+^T, R) + \frac{d_-}{d}\sqrt{\frac{\epsilon t}{\Delta}}f_{\mathrm{root}}(A_-^T, R).
\end{align}
From this we have 
\begin{align}
    f_{\mathrm{root}}\qty(\rho_{ER}, \sigma_E \otimes \frac{\mathbb{I}_R}{d}) &= \sqrt{Q(1-tz)}+\frac{t}{\sqrt{\Delta}}\qty(f_{\mathrm{root}}(A_0^T, R)+ \frac{d_+}{d}f_{\mathrm{root}}(A_+^T, R) + \frac{d_-}{d}f_{\mathrm{root}}(A_-^T, R))+o(t)\nonumber\\
    &= 1-\frac{t}{2}(n\Delta + z) + t\sqrt{\frac{z}{\Delta}}\mathcal{S}(\hat{R}) + o(t)
\end{align}
where we have redefined $R = z\hat{R}$ and 
\begin{align}
    \mathcal{S}(\hat{R}) := f_{\mathrm{root}}(A_0^T, \hat{R})+ \frac{d_+}{d}f_{\mathrm{root}}(A_+^T, \hat{R}) + \frac{d_-}{d}f_{\mathrm{root}}(A_-^T, \hat{R}). 
\end{align}
The $z$ dependent part is a quadratic function of $\sqrt{z}$ and so $z^\star = \frac{\mathcal{S}(\hat{R})^2}{\Delta}$ achieves the maximum value. 
Noting that the target optimal value is the squared fidelity, we get 
\begin{align}
    g_{n,d}(p) \leq 1-t(n\Delta -\frac{1}{\Delta} \max_{\hat{R}\geq 0, \Tr \hat{R} = 1}[\mathcal{S}(\hat{R})]^2)+ o(t). 
\end{align}
Using $f_{\mathrm{root}}(X,Y) \leq \sqrt{\Tr X \Tr Y}$ and $s:= \sum_{r=1}^n c_r^2 \geq \frac{1}{n}$, $\Tr K' = n-s, \Tr S' = -1-s, \Tr L' = -1+s$, one can show 
\begin{align*}
    \Tr(A_0) = n+2 - \Delta s\\
    \Tr(A_+) = n-d+1\\
    \Tr(A_-) = n+d+1
\end{align*}
and 
\begin{align}
    \mathcal{S}(\hat{R}) &\leq \sqrt{\Tr A_0 \Tr \hat{R}} + \frac{d_+}{d} \sqrt{\Tr A_+ \Tr \hat{R}} + \frac{d_-}{d} \sqrt{\Tr A_- \Tr \hat{R}} \nonumber \\
    & = \sqrt{n+2-\Delta s} + \frac{(d-1)(d+2)}{2}\sqrt{n-d+1} + \frac{(d-2)(d+1)}{2}\sqrt{n+d+1} \nonumber \\
    & \leq \sqrt{n+2-\frac{d^2-1}{n}} +\frac{(d-1)(d+2)}{2}\sqrt{n-d+1} + \frac{(d-2)(d+1)}{2}\sqrt{n+d+1}
\end{align}
Therefore, we get 
\begin{align}
    g_{n,d}(p) \leq 1- \qty[\frac{(d^2-1)(d^2+2)}{4d^2n}+o(n^{-1})]p + o(p). 
\end{align}
\end{proof}
\subsection{Information-theoretical achievability}
We now prove the achievability. We first provide an information-theoretical proof. After that, we provide a concrete $\mathrm{SU}(d)$-covariant parallel protocol attaining the bound. 
\begin{lm}[Parallel achievability]
For $n=kd-1$, 
\begin{align}
C'_{n,d} \leq  \frac{1}{d^2}\qty[n(d^2-1) - \frac{\mathfrak{B}_{n,d}^2}{d^2-1}], 
\end{align}
where 
\begin{align}
\mathfrak{B}_{n,d} = \sqrt{n+2-\frac{d^2-1}{n}} +\frac{(d-1)(d+2)}{2}\sqrt{n-d+1} + \frac{(d-2)(d+1)}{2}\sqrt{n+d+1}. 
\end{align}
\end{lm}
\begin{proof}
Now, we move on to the lower bound. Set the irrep $\alpha = (k, k, \cdots k-1)$. Then, this irrep has dimension $d$. Now, define the encoder 
\begin{align}
    V_{n,d}|\psi\rangle = \frac{1}{\sqrt{m_{\alpha}}}\sum_{j=1}^{m_{\alpha}}U_{\mathrm{Sch}}^\dagger(|\alpha\rangle \otimes |\psi\rangle _{\mathcal{U}_{\alpha}} \otimes |j\rangle_{\mathcal{S}_{\alpha}}) \otimes |j\rangle_E. 
\end{align}
This satisfies the covariant condition
\begin{align}\label{eq:Ucovap2}
    (U^{\otimes n} \otimes \mathbb{I}_E)V_{n,d} = V_{n,d} U^* 
\end{align}
One can then show 
\begin{align}
    V_{n,d}^\dagger \mathcal{O}  V_{n,d} = \frac{1}{m_{\alpha}} \Tr_{\mathcal{S}_{\alpha}}(\Pi_{\alpha}\mathcal{O} \Pi_{\alpha})
\end{align}
for any observable $\mathcal{O}\in \mathcal{L}(\mathbb{C}^d)^{\otimes n}$, where $\Pi_{\alpha}$ is the Young projector. This implies 
\begin{align}
    V_{n,d}^\dagger T_{a,1} V_{n,d}  = V_{n,d}^\dagger T_{a,2} V_{n,d} = \cdots V_{n,d}^\dagger T_{a,n} V_{n,d}. 
\end{align}
Define $S_a := \sum_{r=1}^n T_{a,r}$. Differentiating the covariance condition~\eqref{eq:Ucovap2} gives 
\begin{align}
     S_a V_{n,d} = -V_{n,d}T_a^*, \quad V_{n,d}^\dagger T_{a,r} V_{n,d} = -\frac{1}{n}T_a^*. 
\end{align}
 For $r=s$, we have 
\begin{align}
    V_{n,d}^\dagger T_{a,r} T_{b,r} V_{n,d} = V_{n,d}^\dagger (\delta_{ab} + (d_{abc}+if_{abc}) T_{c,r}) V_{n,d} =  \delta_{ab} - \frac{1}{n} (d_{abc}+if_{abc}) T_{c}^*. 
\end{align}
For $r\neq s$, we have 
\begin{align}
    V_{n,d}^\dagger \qty(\sum_{r}\sum_{s}T_{a,r} T_{b,s})V_{n,d} = V_{n,d}^\dagger S_a S_b V_{n,d} = T_a^* T_b^* = \delta_{ab}+(d_{abc}-if_{abc})T_c^*. 
\end{align}
Since 
\begin{align}
    V_{n,d}^\dagger \qty(\sum_{r}\sum_{s}T_{a,r} T_{b,s}) V_{n,d}= n\delta_{ab} -  (d_{abc}+if_{abc}) T_{c}^*  +  V_{n,d}^\dagger \qty(\sum_{r\neq s}T_{a,r} T_{b,s}) V_{n,d},
\end{align}
we have 
\begin{align}
    V_{n,d}^\dagger T_{a,r} T_{b,s} V_{n,d} = -\frac{1}{n}\delta_{ab}+\frac{2d_{abc}}{n(n-1)}T_c^*. \quad (r\neq s)
\end{align}
Comparing with the upper bound argument, we get 
\begin{align}
\begin{dcases}
    K_{rr}=1, \quad c_r = S_{rr}=L_{rr}=-\frac{1}{n}\\
    K_{rs} = -\frac{1}{n}, \quad S_{rs} = \frac{2}{n(n-1)}, \quad L_{rs}=0 \quad (r\neq s)
\end{dcases}
\end{align}
and
\begin{align}
\begin{dcases}
    K' = K -cc^T = (1+\frac{1}{n})I-|u\rangle \langle u|-\frac{1}{n}|u\rangle \langle u| = \qty(1+\frac{1}{n})P_\perp \\
    S' = S - cc^T = -\frac{n+1}{n(n-1)}I + \frac{2}{n-1}|u\rangle \langle u| -\frac{1}{n}|u\rangle \langle u| = -\frac{n+1}{n(n-1)} P_\perp, \\
    L' = L + cc^T = -\frac{1}{n}I+\frac{1}{n}|u\rangle \langle u| = -\frac{1}{n} P_\perp 
\end{dcases}
\end{align}
where $|u\rangle = \frac{1}{\sqrt{n}} \sum_{r=1}|r\rangle$, and $P_\perp = I -|u\rangle \langle u|$. The choice 
\begin{align}
    \hat{R} = \frac{P_\perp}{n-1}
\end{align}
gives the saturation of the upper bound inequality, because  $f_{\mathrm{root}}(A_0,\hat{R}) = \sqrt{\Tr A_0 \Tr \hat{R}},\, f_{\mathrm{root}}(A_+, \hat{R}) = \sqrt{\Tr A_+ \Tr \hat{R}}, \,f_{\mathrm{root}}(A_-, \hat{R}) = \sqrt{\Tr A_- \Tr \hat{R}}$ and 
$s= \frac{1}{n}$. 
\end{proof}

\subsection{Quantum circuit for noisy unitary conjugation}
Here we prove the achievability of the $\mathrm{SU}(d)$-covariant protocol by constructing a concrete encoder and decoder. The arguments here are for the noisy unitary conjugation. 
\subsubsection{Registers and Young diagrams}
Fix integers
\begin{align}
 d\ge2,\qquad q\ge1,\qquad n=qd-1,\qquad N=n+1=qd.
 \label{eq:regime}
\end{align}
The logical input and final output are
$P\simeq F\simeq\C^d$.  The joint inputs and outputs of the $n$
parallel queries are
\[
 \bfI\coloneqq I_1 \otimes \cdots \otimes I_n\simeq(\C^d)^{\otimes n},
 \qquad
 \bfO\coloneqq O_1 \otimes \cdots \otimes O_n\simeq(\C^d)^{\otimes n}.
\]
For a Young diagram $\eta\vdash_d r$, let $\cU_\eta$ and
$\cS_\eta$ denote the corresponding $\mathrm U(d)$ and symmetric-group
irrep spaces, and put
\[
 d_\eta\coloneqq \dim\cU_\eta,\qquad
 m_\eta\coloneqq \dim\cS_\eta=|\Path{\eta}|.
\]

The balanced diagram and its rectangular successor are
\begin{align}
 \alpha=(q^{d-1},q-1)\vdash n,
 \qquad
 \lambda=(q^d)\vdash N.
 \label{eq:diagrams}
\end{align}
We also define Young diagrams
\begin{align}
  &\mu=(q+1,q^{d-2},q-1), \nonumber \\
  &\beta=(q+1,q^{d-3},(q-1)^2),\quad \gamma=(q+1,q^{d-2},q-2), \nonumber \\
  &\rho = (q^{d-2},(q-1)^2),\quad \sigma=(q^{d-1},q-2).
\end{align}
The one-box branching relations among these diagrams are summarized in Fig.~\ref{fig:young-branching-diagram}.
They obey
\begin{align}
 &d_\alpha=d,\qquad d_\beta = {d(d+1)(d-2)\over 2},\qquad d_\gamma={d(d-1)(d+2)\over 2} \nonumber\\
 &d_\lambda=1, \qquad d_\mu=d^2-1.
 \label{eq:dimensions}
\end{align}
We note that for $d=2$, the irrep $\beta$ does not exist, but the final results in the subsequent discussions are unaffected because $d_\beta = 0$ in the above formula. A fixed unitary
$J_\alpha:P\to\cU_\alpha$ specifies the logical basis.
\begin{figure}[h]
\centering
\includegraphics[width=0.5\linewidth]
{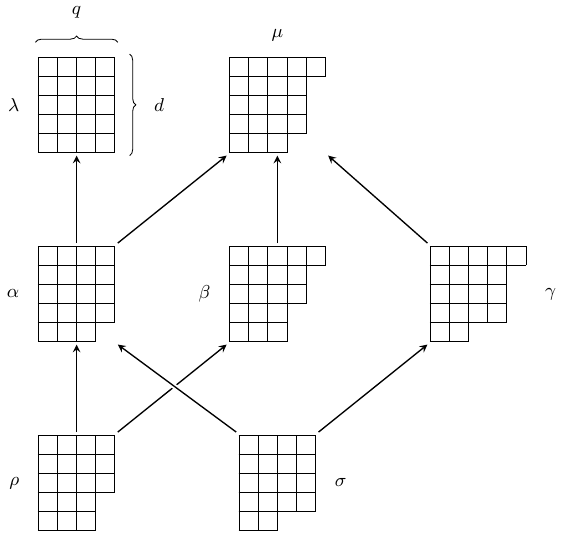}
\caption{The relationships among the Young diagrams used in the analysis of the fidelity.}
\label{fig:young-branching-diagram}
\end{figure}

Let $\cN_{U,p}:\mathcal{L}(\C^d)\to
\mathcal{L}(\C^d)$ be the queried channel.  The construction
below applies to any $\cN_{U,p}$; for the noisy-unitary setting one may
take
\[
 \cN_{U,p}=\mathcal D_p\circ\mathcal U,
 \qquad
 \mathcal U(\rho)=U\rho U^\dagger,\qquad
 \mathcal D_0=\operatorname{id}.
\]
Thus $p=0$ gives the noiseless unitary query.

\subsubsection{Encoder}
We use the Schur transform with output-register order
\begin{align}
 U_{\mathrm{Sch}}^{(n)}:
 (\C^d)^{\otimes n}
 \longrightarrow
 \bigoplus_{\eta\vdash_dn}
 \cS_\eta\otimes\ket{\eta}\otimes\cU_\eta.
 \label{eq:schur-order}
\end{align}
Since $\cU_\alpha\cong \overline{\C^d}$ holds, we can take an isomorphism $J_\alpha:P\to\cU_\alpha$.
Then, the encoder isometry $E_{n,d}:P\to\bfI\otimes E$ for $E\simeq \cS_\lambda$ is defined by
\begin{align}
  E_{n,d}\ket{\psi}_P \coloneqq \frac{1}{\sqrt{m_\alpha}} \sum_{s_\alpha \in \Path{\alpha}} U_{\mathrm{Sch}}^{(n)\dagger} \qty[\ket{s_\alpha}_{\cS_\alpha} \otimes \ket{\alpha} \otimes (J_\alpha \ket{\psi})_{\cU_\alpha}] \otimes \ket{s_\alpha \to \lambda}_E,
  \label{eq:encoder}
\end{align}
which defines the encoder channel $\cE_{n,d}\coloneqq E_{n,d}(\cdot) E_{n,d}^\dagger$.

\subsubsection{Decoder}

For $\eta\in \mu-\square = \{\alpha, \beta, \gamma\}$, we define the path frame vectors $\ket{g_{i,\eta}}$ by
\begin{align}
  \ket{g_{i,\eta}}
  \coloneqq
  \sum_{s_\eta\in\Path{\eta}}
  \rho_\lambda(\tau_i)
  R_{\lambda\to\mu}^\dagger
  \rho_\mu(\tau_i)
  \ket{s_\eta\to\mu}\otimes\ket{s_\eta},
  \label{eq:direct-frame-vectors}
\end{align}
where $\tau_i$ is the transposition $(i,n+1)$ in $\mathfrak S_{n+1}$, $\rho_\nu$ is Young's orthogonal representation of $\mfS_{n+1}$ on $\cS_\nu$ for $\nu\in \{\lambda, \mu\}$, and $R_{\lambda\to\mu}$ is defined by
\begin{align}
  R_{\lambda\to \mu} \coloneqq \sum_{s_\alpha\in \Path{\alpha}} \ketbra{s_{\alpha}\to \mu}{s_{\alpha}\to \lambda}
\end{align}
We define the normalized invariant vector by
\begin{align}
  \ket{\Omega_{\lambda\alpha}}
  \coloneqq
  \frac{1}{\sqrt{m_\alpha}}
  \sum_{s_\alpha\in\Path{\alpha}}
  \ket{s_\alpha\to\lambda}\otimes\ket{s_\alpha}.
  \label{eq:direct-omega}
\end{align}
We then define the centered frame vectors $\ket{\widetilde{g}_{i, \eta}}$ by
\begin{align}
  \ket{\widetilde{g}_{i,\alpha}} &\coloneqq (\mathbb{I} - \ketbra{\Omega_{\lambda\alpha}}) \ket{g_{i,\alpha}}, \\
  \ket{\widetilde{g}_{i,\eta}} &\coloneqq \ket{g_{i,\eta}} \quad \forall \eta\in\{\beta, \gamma\}
\end{align}
and the corresponding frame operators $A_\eta$, $\widetilde{A}_\eta$ and Gram matrices $G_\eta$, $\widetilde{G}_\eta$ by
\begin{align}
  A_\eta \coloneqq \sum_{i=1}^n \ketbra{g_{i,\eta}}{i}, \quad G_\eta \coloneqq A_\eta^\dagger A_\eta,
  \label{eq:direct-frame-and-gram-operators}\\
  \widetilde{A}_\eta \coloneqq \sum_{i=1}^n \ketbra{\widetilde{g}_{i,\eta}}{i}, \quad \widetilde{G}_\eta \coloneqq \widetilde{A}_\eta^\dagger \widetilde{A}_\eta.
\end{align}
As shown in the following proposition, the Gram matrices $\widetilde{G}_\eta$ have a single nonzero eigenvalue $g_\eta$ with multiplicity $n-1$:
\begin{prop}
\label{prop:direct-gram}
The Gram matrix $\widetilde{G}_{\eta}$ is given by
\begin{align}
  \widetilde{G}_{\eta} = g_\eta (\mathbb{I}-\ketbra{+}),
\end{align}
where $\ket{+} \coloneqq \frac{1}{\sqrt{n}} \sum_{i=1}^n \ket{i}$ and $g_\eta$ is the unique nonzero eigenvalue of $\widetilde{G}_\eta$ given by
\begin{align}
 \frac{g_\alpha}{m_\alpha}
 &=\frac{d^2(q^2-1)}{n(n-1)(d^2-1)},
 \label{eq:direct-g-alpha}\\
 \frac{g_\beta}{m_\alpha}
 &=\frac{d(d-2)(q+1)}{2(n-1)(d-1)},
 \label{eq:direct-g-beta}\\
 \frac{g_\gamma}{m_\alpha}
 &=\frac{d(d+2)(q-1)}{2(n-1)(d+1)}.
 \label{eq:direct-g-gamma}
\end{align}
\end{prop}
From this proposition, we can define partial isometries $V_\eta^{\mathrm{act}}$ for $\eta\in\{\alpha, \beta, \gamma\}$ as
\begin{align}
  V_\alpha^{\mathrm{act}}
  &\coloneqq 
  \ketbra{0, \lambda}{\Omega_{\lambda\alpha}}
  +
  \frac{1}{\sqrt{g_\alpha}}
  \sum_{i=1}^{n}\ketbra{i, \mu}{\widetilde{g}_{i,\alpha}},
  \label{eq:direct-active-V-alpha}\\
  V_\eta^{\mathrm{act}}
  &\coloneqq 
  \frac{1}{\sqrt{g_\eta}}
  \sum_{i=1}^{n}\ketbra{i, \mu}{\widetilde{g}_{i,\eta}}
  \quad \forall \eta\in\{\beta,\gamma\},
  \label{eq:direct-active-V-other}
\end{align}
which can be completed to isometries $V_\eta$ by adding orthonormal vectors in the orthogonal complement of the support of $V_\eta^{\mathrm{act}}$.
We also fix isometries $V_\eta$ for $\eta\notin\{\alpha, \beta, \gamma\}$ arbitrarily.
Then, we define the decoder isometry $D_{n,d}: E\otimes \bfO\to R\otimes Y\otimes F\otimes \cU_Y$ by
\begin{align}
  &D_{n,d}(\ket{s_\lambda}_E \otimes \ket{o_1\cdots o_n}_{\bfO}) \notag\\
  &\coloneqq \sum_{\eta\vdash_d n} \sum_{s_\eta\in \Path{\eta}} (\mathbb{I}_R \otimes \mathrm{dCG}^\dagger)V_\eta(\ket{s_\lambda} \otimes \ket{s_\eta}) \otimes (\bra{s_\eta} \otimes \bra{\eta}\otimes \mathbb{I}_{\cU_\eta})U_\mathrm{Sch}^{(n)}\ket{o_1\cdots o_n},
\end{align}
where $R = \mathrm{span}\{\ket{r}\}_{r=0}^{n}$, $Y = \mathrm{span}\{\ket{\nu}\}_{\nu\vdash_d n+1}$, and $\cU_Y = \bigoplus_{\nu\vdash_d n+1} \cU_\nu$, and we define the joint Hilbert space $Z \coloneqq R\otimes Y\otimes\cU_Y$.
We define the decoder channel $\cD_{n,d}$ by
\begin{align}
  \cD_{n,d}(\cdot)\coloneqq \Tr_{Z}\!\qty[D_{n,d} (\cdot) D_{n,d}^\dagger].
  \label{eq:decoder-channel-definition}
\end{align}
The output channel implemented by the full circuit is
\begin{align}
 \Phi_{U,p}
 =
 \cD_{n,d}\circ
 \left(\cN_{U,p}^{\otimes n}\otimes\operatorname{id}_E\right)
 \circ \cE_{n,d}.
 \label{eq:full-channel}
\end{align}
The complete circuit is given by Fig.~\ref{fig:full-circuit}.

\begin{figure}[htbp]
 \centering
 \includegraphics[width=0.75\linewidth]
 {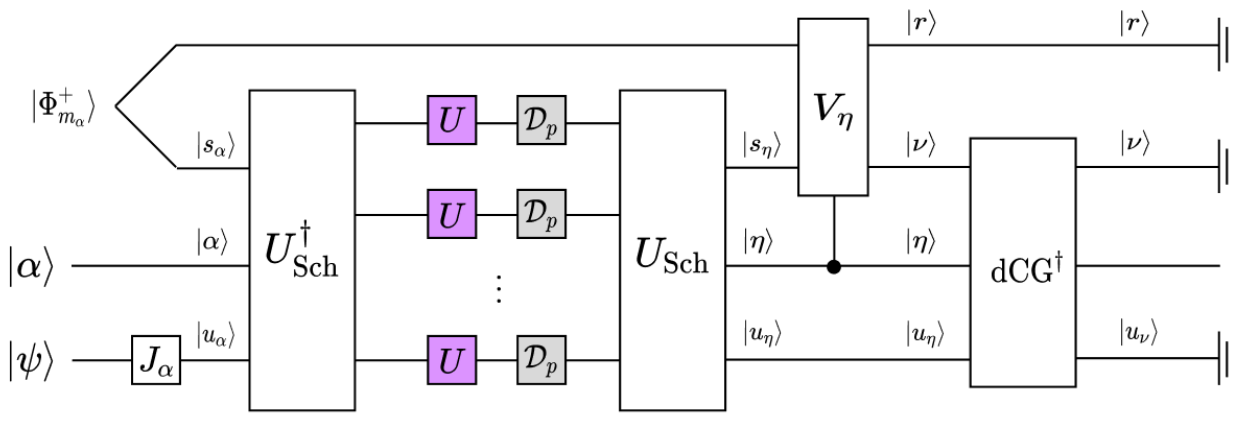}
 \caption{Full Stinespring circuit. Here, $|\Phi_{m_\alpha}^+\rangle := \frac{1}{\sqrt{m_\alpha}} \sum_{s_\alpha} |s_\alpha\rangle_{\mathcal{S}_\alpha} \otimes |s_\alpha \to \lambda\rangle_E$. The middle boxes represent the $n$ parallel calls
 $\cN_{U,p}^{\otimes n}:\bfI\to\bfO$.
 The Schur multiplicity register is retained coherently until
 $V_\eta$; it must not be measured or discarded.
 The final registers $r$, $\nu$, and $u_\nu$ comprise the decoder environment.}
 \label{fig:full-circuit}
\end{figure}

\subsubsection{Evaluation of the fidelity}
\label{sec:one-error-output}
The input noisy query is given by
\begin{gather}
\mathcal U(\rho)\coloneqq U\rho U^\dagger,\qquad \cT(\rho)\coloneqq \Tr(\rho)\frac{\mathbb{I}_d}{d},\qquad \cN_{U,p}=(1-p)\mathcal U+p\cT.
\label{eq:one-error-noise-model}
\end{gather}
We choose a Hermitian operator basis $\{T_a\}_{a=0}^{\Delta}$ with $\Delta\coloneqq d^2-1$ satisfying
\begin{gather}
T_0=\mathbb{I}_d,\qquad \Tr(T_aT_b)=d\delta_{ab}.
\label{eq:one-error-operator-basis}
\end{gather}
Here, $\{T_a\}_{a=1}^{\Delta}$ forms a basis of $\mathrm{Adj}$, and $T_0$ is newly introduced to form a basis of $\mathbb{C}^d \otimes \bar{\mathbb{C}}^d$. Then, the completely depolarizing channel can be written as
\begin{gather}
\cT(\rho)=\sum_{a=0}^{\Delta}{T_a \over d} (\cdot) {T_a\over d}.
\label{eq:one-error-replacer-kraus}
\end{gather}
Let $Z\coloneqq R\otimes Y\otimes\cU_Y$ denote all decoder registers that are traced out in Eq.~\eqref{eq:decoder-channel-definition}.
For an orthonormal basis $\{\ket{z}\}_z$ of $Z$, we define the decoder Kraus operators $\{D_z\}_z$ by
\begin{gather*}
D_z\coloneqq(\mathbb{I}_F\otimes\bra{z}_Z)D_{n,d}:\bfO\otimes E\longrightarrow F.
\label{eq:one-error-decoder-kraus}
\end{gather*}
We can evaluate the fidelity for the case of $U = \mathbb{I}_d$ due to the covariance of the protocol.
The output channel for $U = \mathbb{I}_d$ is given by
\begin{align}
  \Phi_p(\rho) = (1-p)^n \Phi^{(0)}(\rho) + p(1-p)^{n-1} \Phi^{(1)}(\rho) + O(p^2),
\end{align}
where $\Phi^{(0)}(\rho)$ and $\Phi^{(1)}(\rho)$ are defined by
\begin{align}
  \Phi^{(0)} &\coloneqq \cD_{n,d} \circ \qty[\mathbb{I}_d^{\otimes n} \otimes \mathbb{I}_E] \circ \cE_{n,d},\\
  \Phi^{(1)} &\coloneqq \sum_{i=1}^{n} \cD_{n,d} \circ \qty[\qty(\mathbb{I}^{\otimes (i-1)} \otimes \cT \otimes \mathbb{I}^{\otimes (n-i)}) \otimes \mathbb{I}_E] \circ \cE_{n,d}.
\end{align}
We define
\begin{align}
  K_{a,i,z}\coloneqq D_z \qty(\mathbb{I}^{\otimes (i-1)} \otimes {T_a\over d} \otimes \mathbb{I}^{\otimes (n-i)})E_{n,d},
\end{align}
and we denote $K_{0,z}\coloneqq K_{0,i,z}$ for any $i\in [n]$ since it does not depend on $i$.
Then, $\Phi^{(0)}$ and $\Phi^{(1)}$ can be expressed in terms of Kraus operators as
\begin{align}
  \Phi^{(0)}(\cdot) &\coloneqq d^2 \sum_z K_{0,z}(\cdot)K_{0,z}^\dagger,\\
  \Phi^{(1)}(\cdot) &\coloneqq \sum_{i=1}^{n}\sum_{a=0}^{\Delta}\sum_z K_{a,i,z}(\cdot)K_{a,i,z}^\dagger.
\end{align}
The fidelity between the output channel $\Phi_p$ and the ideal channel $\mathbb{I}_d$ is given by
\begin{align}
  F = (1-p)^n \sum_z|f_{0,z}|^2 + {p(1-p)^{n-1} \over d^2} \sum_{i=1}^{n}\sum_{a=0}^{\Delta}\sum_z |f_{a,i,z}|^2 + O(p^2),
\end{align}
where $f_{0,z}$ and $f_{a,i,z}$ are defined by
\begin{align}
  f_{0,z} \coloneqq \Tr[K_{0,z}], \quad f_{a,i,z} \coloneqq \Tr[K_{a,i,z}].
\end{align}
We evaluate $f_{0,z}$ and $f_{a,i,z}$ as follows.
First, the encoder can be written as
\begin{align}
  E_{n,d}\ket{\psi}_P
  &\coloneqq \frac{1}{\sqrt{m_\alpha}} \sum_{s_\alpha \in \Path{\alpha}} U_{\mathrm{Sch}}^{(n)\dagger} \qty[\ket{s_\alpha}_{\cS_\alpha} \otimes \ket{\alpha} \otimes (J_\alpha \ket{\psi})_{\cU_\alpha}] \otimes \ket{s_\alpha \to \lambda}_E \nonumber\\
  &= \frac{1}{\sqrt{m_\alpha}} \sum_{s_\alpha \in \Path{\alpha}} U_{\mathrm{Sch}}^{(n)\dagger}\qty[\ket{s_\alpha} \otimes \ket{\alpha} \otimes \bra{\alpha}\mathrm{CG}_{\lambda, \overline{\square}} (\ket{u_\lambda} \otimes \ket{\psi})] \otimes \ket{s_\alpha \to \lambda}_E \nonumber\\
  &= \sqrt{d \over m_\alpha} \sum_{s_\alpha \in \Path{\alpha}} U_{\mathrm{Sch}}^{(n)\dagger}\qty[\ket{s_\alpha} \otimes \ket{\alpha} \otimes (\mathbb{I}_{\cU_\alpha} \otimes \bra{\overline{\psi}})\mathrm{CG}_{\alpha, \square}^\dagger \ket{u_\lambda}] \otimes \ket{s_\alpha \to \lambda}_E \nonumber\\
  &= \sqrt{d \over m_\alpha} \sum_{s_\alpha \in \Path{\alpha}} (\mathbb{I}_d^{\otimes n} \otimes \bra{\overline{\psi}})U_\mathrm{Sch}^{(n+1)\dagger} \qty[\ket{s_\alpha\to\lambda} \otimes \ket{\lambda} \otimes \ket{u_\lambda}] \otimes \ket{s_\alpha \to \lambda}_E \nonumber\\
  &= \sqrt{d \over m_\alpha} \sum_{s_\lambda \in \Path{\lambda}} (\mathbb{I}_d^{\otimes n} \otimes \bra{\overline{\psi}})U_\mathrm{Sch}^{(n+1)\dagger} \qty[\ket{s_\lambda} \otimes \ket{\lambda} \otimes \ket{u_\lambda}] \otimes \ket{s_\lambda}_E
\end{align}
using a basis vector $\ket{u_\lambda}\in \cU_\lambda\cong \C$, where we use $\Path{\lambda} = \{\ket{s_\alpha\to\lambda} \mid s_\alpha\in\Path{\alpha}\}$ which comes from $\lambda-\square = \{\alpha\}$.
The decoder can be written as
\begin{align}
  &\bra{r}\bra{\nu}\bra{\overline{\psi}}\bra{u_\nu} D_{n,d}(\ket{s_\lambda} \otimes \ket{o_1\cdots o_n})\notag\\
  &= \sum_{\eta\vdash_d n} \sum_{s_\eta\in \Path{\eta}} \bra{r, \nu}V_\eta(\ket{s_\lambda} \otimes \ket{s_\eta}) \otimes (\bra{s_\eta} \otimes \bra{\eta}\otimes \bra{u_\nu} \bra{\overline{\psi}} \mathrm{CG}^\dagger_{\nu, \overline{\square}})U_\mathrm{Sch}^{(n)}\ket{o_1\cdots o_n} \nonumber\\
  &= \sum_{\eta\vdash_d n} \sqrt{d_\eta\over d_\nu} \sum_{s_\eta\in \Path{\eta}} \bra{r, \nu}V_\eta(\ket{s_\lambda} \otimes \ket{s_\eta}) \otimes (\bra{s_\eta} \otimes \bra{\eta}\otimes \bra{u_\nu}\mathrm{CG}_{\eta, \square} \ket{\psi})U_\mathrm{Sch}^{(n)}\ket{o_1\cdots o_n} \nonumber\\
  &= \sum_{\eta\vdash_d n} \sqrt{d_\eta\over d_\nu} \sum_{s_\eta\in \Path{\eta}}\bra{r, \nu}V_\eta(\ket{s_\lambda} \otimes \ket{s_\eta}) \bra{s_\eta\to\nu}\bra{\nu}\bra{u_\nu}U_{\mathrm{Sch}}^{(n+1)} \ket{o_1\cdots o_n\psi}.
\end{align}
Thus, we have ($z = (r, \nu, u_\nu)$)
\begin{align}
  f_{a,i,z}&\coloneqq \Tr\!\left[D_z\left(\left(\mathbb{I}_d^{\otimes(i-1)}\otimes\frac{T_a}{d}\otimes\mathbb{I}_d^{\otimes(n-i)}\right)\otimes\mathbb{I}_E\right)E_{n,d}\right]\notag\\
  &= \sum_{\eta\vdash_d n} \sum_{s_\eta\in \Path{\eta}} \sum_{s_\lambda\in\Path{\lambda}} \sqrt{d_\eta \over d m_\alpha d_\nu}\bra{r, \nu}V_\eta(\ket{s_\lambda} \otimes \ket{s_\eta}) \notag\\
  &\times \bra{s_\eta \to \nu} \bra{\nu} \bra{u_\nu} U_\mathrm{Sch}^{(n+1)} \left(\mathbb{I}_d^{\otimes(i-1)}\otimes T_a \otimes\mathbb{I}_d^{\otimes(n+1-i)}\right)U_\mathrm{Sch}^{(n+1)}(\ket{s_\lambda}\otimes \ket{\lambda}\otimes\ket{u_\lambda}) \nonumber\\
  &= \sum_{\eta\vdash_d n} \sum_{s_\eta\in \Path{\eta}} \sum_{s_\lambda\in\Path{\lambda}} \sqrt{d_\eta \over d m_\alpha d_\nu} \bra{r, \nu}V_\eta(\ket{s_\lambda} \otimes \ket{s_\eta}) \notag\\
  &\times \bra{s_\eta \to \nu} \rho_\nu(\tau_i) \bra{\nu} \bra{u_\nu} U_\mathrm{Sch}^{(n+1)} \left(\mathbb{I}_d^{\otimes n}\otimes T_a \right)U_\mathrm{Sch}^{(n+1)}(\rho_\lambda(\tau_i)\ket{s_\lambda}\otimes \ket{\lambda}\otimes\ket{u_\lambda}).
\end{align}
Here, $\cU_\alpha\otimes \cU_\square \cong \cU_\lambda\oplus \cU_\mu$ holds and $\cU_\lambda$ corresponds to the one-dimensional space spanned by $\ket{\Omega_{\alpha\square}} \coloneqq {1\over \sqrt{d}} \sum_{i=1}^{d} \ket{i} \otimes \ket{i} \in \cU_\alpha\otimes \cU_\square$ using the computational basis $\{\ket{i}\}_{i=1}^{d}$ of $\cU_\alpha\cong \C^d$ and $\cU_\square\cong \overline{\C^d}$.
Then, we have
\begin{align}
  (U_\mathrm{Sch}^{(n)} \otimes T_a)U_\mathrm{Sch}^{(n+1)\dagger}(\ket{s_\lambda} \otimes \ket{\lambda} \otimes \ket{u_\lambda})
  &= (\mathbb{I}_{\cU_\alpha} \otimes T_a) \ket{\Omega_{\alpha, \square}}.
\end{align}
Since $(\mathbb{I}_{\cU_\alpha} \otimes T_a) \ket{\Omega_{\alpha, \square}}$ is orthogonal to $\ket{\Omega_{\alpha, \square}}$ for $a\geq 1$, we have
\begin{align}
  \mathrm{CG}_{\alpha, \square}(\mathbb{I}_{\cU_\alpha} \otimes T_a)\ket{\Omega_{\alpha, \square}}
  \in
  \begin{cases}
    \cU_\lambda & (a=0)\\
    \cU_\mu & (a\geq 1)
  \end{cases}.
\end{align}
Thus, we can define a unitary operator $J_{\mathrm{ad}}: \mathrm{span}\{\ket{a}\}_{a=1}^{\Delta} \to \cU_\mu$ by
\begin{align}
  J_{\mathrm{ad}}\ket{a} \coloneqq \mathrm{CG}_{\alpha, \square}(\mathbb{I}_{\cU_\alpha} \otimes T_a) \ket{\Omega_{\alpha, \square}}.
\end{align}
Then, we have
\begin{align}
  U_\mathrm{Sch}^{(n+1)}(\mathbb{I}_d^{\otimes n} \otimes T_a)U_\mathrm{Sch}^{(n+1)\dagger}(\ket{s_\alpha\to\lambda} \otimes \ket{\lambda} \otimes \ket{u_\lambda})
  =
  \begin{cases}
    \ket{s_\alpha\to\lambda} \otimes \ket{\lambda} \otimes \ket{u_\lambda} & (a=0)\\
    \ket{s_\alpha \to \mu}\otimes \ket{\mu} \otimes J_{\mathrm{ad}} \ket{a} & (a\geq 1)
  \end{cases},
\end{align}
which shows that
\begin{align}
  &\bra{s_\eta \to \nu} \rho_\nu(\tau_i) \bra{\nu} \bra{u_\nu} U_\mathrm{Sch}^{(n+1)} \left(\mathbb{I}_d^{\otimes n}\otimes T_a \right)U_\mathrm{Sch}^{(n+1)}(\rho_\lambda(\tau_i)\ket{s_\lambda}\otimes \ket{\lambda}\otimes\ket{u_\lambda})\notag\\
  &=\begin{cases}
    \delta_{\eta\alpha} \delta_{s_\eta s_\alpha} \delta_{\nu\lambda} \delta_{u_\nu u_\lambda} & (a=0)\\
    \delta_{\nu\mu} \bra{s_\eta\to\mu}\rho_{\mu}(\tau_i) R_{\lambda\to\mu} \rho_{\lambda}(\tau_i) \ket{s_\lambda} \cdot \bra{u_\nu} J_{\mathrm{ad}} \ket{a} & (a\geq 1)
  \end{cases}.
\end{align}
Using this, we can further calculate $f_{a,i,z}$ as follows.
When $a=0$, we write $f_{0,z}\coloneqq f_{0,i,z}$ for any $i\in [n]$ since it is independent of $i$:
\begin{align}
  f_{0,z}
  &= \sqrt{1 \over m_\alpha} \delta_{\nu\lambda}\delta_{u_\nu u_\lambda} \sum_{s_\alpha\in\Path{\alpha}} \bra{r, \nu}V_\alpha(\ket{s_\alpha\to\lambda} \otimes \ket{s_\alpha}) \nonumber\\
  &= \delta_{\nu\lambda}\delta_{u_\nu u_\lambda} \bra{r,\nu} V_\alpha \ket{\Omega_{\lambda\alpha}} \nonumber\\
  &= \delta_{\nu\lambda} \delta_{u_\nu u_\lambda} \delta_{r,0}.
\end{align}
When $a\geq 1$,
\begin{align}
  f_{a,i,z}
  &=\sum_{\eta\vdash_d n} \sqrt{d_\eta \over d m_\alpha d_\mu} \delta_{\nu\mu} \sum_{\substack{s_\eta\in \Path{\eta}\\s_\lambda\in\Path{\lambda}}} \bra{r, \nu}V_\eta(\ket{s_\lambda} \otimes \ket{s_\eta}) \cdot \bra{s_\eta\to\mu} \rho_{\mu}(\tau_i) R_{\lambda\to\mu} \rho_{\lambda}(\tau_i) \ket{s_\lambda} \cdot \bra{u_\nu} J_{\mathrm{ad}} \ket{a} \nonumber\\
  &= \sum_{\eta\vdash_d n} \sqrt{d_\eta \over d m_\alpha d_\mu} \delta_{\nu\mu} \sum_{\substack{s_\lambda\in\Path{\lambda}\\s_\eta\in\Path{\eta}}} \bra{r, \nu}V_\eta(\ket{s_\lambda} \otimes \ket{s_\eta}) \cdot \bra{g_{i,\eta}} (\ket{s_\lambda} \otimes \ket{s_\eta}) \cdot \bra{u_\nu} J_{\mathrm{ad}} \ket{a} \nonumber\\
  &= \sum_{\eta\vdash_d n} \sqrt{d_\eta \over d m_\alpha d_\mu} \delta_{\nu\mu} \bra{r, \nu}V_\eta\ket{g_{i,\eta}} \cdot \bra{u_\nu} J_{\mathrm{ad}} \ket{a} \nonumber\\
  &= \sum_{\eta\in \mu-\square} \sqrt{d_\eta g_\eta\over d m_\alpha d_\mu} \delta_{\nu\mu} \qty(\delta_{r,i}-{1\over n}) \bra{u_\nu} J_{\mathrm{ad}} \ket{a}.
\end{align}
Then, the fidelity is evaluated as follows. Since
\begin{align}
  \sum_z |f_{0,z}|^2 &= \sum_z (\delta_{\nu\lambda} \delta_{u_\nu u_\lambda} \delta_{r,0})^2 = 1,\\
  \sum_{a=1}^{\Delta} \sum_z |f_{a,i,z}|^2
  &= \sum_{a=1}^{\Delta} \sum_z \qty(\sum_{\eta\in \mu-\square} \sqrt{d_\eta g_\eta \over d m_\alpha d_\mu g_\eta} \delta_{\nu\mu} \qty(\delta_{r,i}-{1\over n}) \bra{u_\nu} J_{\mathrm{ad}} \ket{a})^2 \nonumber\\
  &= {n-1\over n}\qty(\sum_{\eta\in \mu-\square} \sqrt{d_\eta g_\eta \over d m_\alpha d_\mu g_\eta})^2 \Tr(J_\mathrm{ad}^\dagger J_{\mathrm{ad}}) \nonumber\\
  &= {n-1\over n} \qty(\sum_{\eta\in \mu-\square} \sqrt{d_\eta g_\eta \over d m_\alpha})^2, 
\end{align}
we obtain
\begin{align}
  F = 1-n\left[1-f_1(d,q)\right]p+O(p^2),
\end{align}
with
\begin{gather}
f_1(d,q)=\frac1{d^2}+{n-1 \over nd^3}\left(\sum_{\eta\in\mu-\square}\sqrt{\frac{d_\eta g_\eta}{m_\alpha}}\right)^2.
\label{eq:one-error-fidelity-branch}
\end{gather}
By substituting Eqs.~\eqref{eq:direct-g-alpha}, \eqref{eq:direct-g-beta}, and \eqref{eq:direct-g-gamma} into Eq.~\eqref{eq:one-error-fidelity-branch}, we obtain
\begin{align}
  f_1(d,q)
  &=\frac1{d^2}+\frac1n\left[\sqrt{\frac{q^2-1}{n(d^2-1)}}+\frac{d-2}{2}\sqrt{\frac{(d+1)(q+1)}{d(d-1)}}+\frac{d+2}{2}\sqrt{\frac{(d-1)(q-1)}{d(d+1)}}\right]^2 \nonumber\\
  &= 1-{(d^2-1)(d^2+2) \over 4d^2 n^2} + O(n^{-3}),
\end{align}
i.e.,
\begin{align}
  F = 1-{(d^2-1)(d^2+2) \over 4d^2 n}p + O(n^{-2}p, p^2).
\end{align}

\subsubsection{Proof of Proposition~\ref{prop:direct-gram}}
\begin{proof}
Recall the frame operator $A_\eta$ and Gram matrix $G_\eta$ from Eq.~\eqref{eq:direct-frame-and-gram-operators}.
For $\pi\in\mathfrak S_n$, let
\begin{align}
  \Pi(\pi)\coloneqq \sum_{i=1}^n\ketbra{\pi(i)}{i}
  \label{eq:direct-position-representation}
\end{align}
be the permutation representation on the query-position space.
The identity $\tau_{\pi(i)}=\pi\tau_i\pi^{-1}$ and the Young branching rule imply
\begin{align}
  \ket{g_{\pi(i),\eta}}
  =
  \bigl(\rho_\lambda(\pi)\otimes
  \overline{\rho_\eta(\pi)}\bigr)\ket{g_{i,\eta}}.
  \label{eq:direct-gram-symmetry}
\end{align}
Indeed, on the path subspaces ending at $\alpha$, the restrictions of $\rho_\lambda(\pi)$ and $\rho_\mu(\pi)$ both equal $\rho_\alpha(\pi)$, while on the path subspace ending at $\eta$, the restriction of $\rho_\mu(\pi)$ equals $\rho_\eta(\pi)$.
Equation~\eqref{eq:direct-gram-symmetry} is equivalent to the intertwining relation
\begin{align}
  A_\eta\Pi(\pi)
  =
  \bigl(\rho_\lambda(\pi)\otimes
  \overline{\rho_\eta(\pi)}\bigr)A_\eta.
  \label{eq:direct-frame-intertwiner}
\end{align}
Consequently,
\begin{align}
  G_\eta\Pi(\pi)=\Pi(\pi)G_\eta.
  \label{eq:direct-gram-commutant}
\end{align}

The permutation representation decomposes multiplicity-freely as $\Pi\simeq\mathrm{triv}\oplus\mathrm{std}$, with the trivial space spanned by $\ket{+}$ and the standard space given by the orthogonal complement of $\ket{+}$.
Schur's lemma therefore gives
\begin{align}
  G_\eta
  =
  g_\eta^{\mathrm{triv}}\ketbra{+}
  +g_\eta (\mathbb{I}_n - \ketbra{+}),
  \qquad
  g_\eta^{\mathrm{triv}},g_\eta\geq 0.
  \label{eq:direct-gram-decomposition}
\end{align}

We first determine the trivial eigenvalue.
Because $\Pi(\pi)\ket{+}=\ket{+}$, Eq.~\eqref{eq:direct-frame-intertwiner} shows that $A_\eta\ket{+}$ is invariant under $\rho_\lambda\otimes\overline{\rho_\eta}$.
The restriction of $\rho_\lambda$ from $\mathfrak S_N$ to $\mathfrak S_n$ contains only $\rho_\alpha$, and hence Schur's lemma gives
\begin{align}
  A_\eta\ket{+}
  \propto
  \begin{cases}
    \ket{\Omega_{\lambda\alpha}} & (\eta=\alpha)\\
    0 & (\eta = \beta, \gamma)
  \end{cases}.
  \label{eq:direct-trivial-range}
\end{align}
The covariance in Eq.~\eqref{eq:direct-gram-symmetry} also shows that $\braket{\Omega_{\lambda\alpha}}{g_{i,\alpha}}$ is independent of $i$.
Therefore
\begin{align}
  \bra{\Omega_{\lambda\alpha}}A_\alpha\ket{+}
  =
  \sqrt n\,
  \braket{\Omega_{\lambda\alpha}}{g_{n,\alpha}}.
  \label{eq:direct-trivial-overlap-reduction}
\end{align}
It remains to evaluate the overlap on the right-hand side.
Equations~\eqref{eq:direct-frame-vectors} and \eqref{eq:direct-omega} give
\begin{align}
  \braket{\Omega_{\lambda\alpha}}{g_{n,\alpha}}
  =
  \frac{1}{\sqrt{m_\alpha}}
  \sum_{s_\alpha,s'_\alpha\in\Path{\alpha}}
  \bra{s'_\alpha\to\lambda}\rho_\lambda(\tau_n)
  \ketbra{s_\alpha\to\lambda}{s_\alpha\to\mu}
  \rho_\mu(\tau_n)\ket{s'_\alpha\to\mu}.
  \label{eq:direct-trivial-overlap-expanded}
\end{align}
Now, the action of the transposition $\tau_n = (n,n+1)$ in the Young--Yamanouchi basis is given as follows~\cite{ceccherini2010representation}: 
\begin{lm}
\label{lem:young-orthogonal-last-swap}
Let $\xi\vdash n-1$, let $e_i$ denote the unit vector that adds one box to row $i$, and suppose
\begin{align}
  s_\nu
  =
  s_\xi\to\xi+e_i\to\xi+e_i+e_j,
  \notag\\
  \xi+e_i\vdash n, \quad \nu=\xi+e_i+e_j\vdash n+1.
  \label{eq:young-path-last-two-boxes}
\end{align}
Define the axial distance
\begin{align}
  r_{s_\nu}\coloneqq (\xi_j-j)-(\xi_i-i)+\delta_{ij}.
  \label{eq:young-axial-distance}
\end{align}
If $i\ne j$ and $\xi+e_j\vdash n$, define
\begin{align}
  \ket{\tau_n\cdot s_\nu}
  \coloneqq 
  \ket{s_\xi\to\xi+e_j\to\xi+e_j+e_i}.
  \label{eq:young-swapped-path-valid}
\end{align}
Otherwise, set
\begin{align}
  \ket{\tau_n\cdot s_\nu}\coloneqq 0.
  \label{eq:young-swapped-path-zero}
\end{align}
Young's orthogonal representation then satisfies
\begin{align}
  \rho_\nu(\tau_n)\ket{s_\nu}
  =
  \frac{1}{r_{s_\nu}}\ket{s_\nu}
  +\sqrt{1-\frac{1}{r_{s_\nu}^2}}\,
  \ket{\tau_n\cdot s_\nu}.
  \label{eq:young-orthogonal-last-swap}
\end{align}
\end{lm}
Lemma~\ref{lem:young-orthogonal-last-swap} makes the first matrix element diagonal because $\lambda$ has the unique parent $\alpha$.
For the two possible penultimate diagrams $\rho$ and $\sigma$, it gives
\begin{align}
  \bra{s'_\alpha\to\lambda}\rho_\lambda(\tau_n)
  \ket{s_\sigma\to\alpha\to\lambda}
  &=
  \delta_{s'_\alpha,s_\sigma\to\alpha}, \nonumber\\
  \bra{s'_\alpha\to\lambda}\rho_\lambda(\tau_n)
  \ket{s_\rho\to\alpha\to\lambda}
  &=
  -\delta_{s'_\alpha,s_\rho\to\alpha}.
  \label{eq:direct-lambda-diagonal-elements}
\end{align}
The same lemma in the $\mu$ representation gives
\begin{align}
  \bra{s_\sigma\to\alpha\to\mu}\rho_\mu(\tau_n)
  \ket{s_\sigma\to\alpha\to\mu}
  &=
  \frac{1}{d+1},
  \notag\\
  \bra{s_\rho\to\alpha\to\mu}\rho_\mu(\tau_n)
  \ket{s_\rho\to\alpha\to\mu}
  &=
  \frac{1}{d-1}.
  \label{eq:direct-mu-diagonal-elements}
\end{align}
Here the path variables range over $s_\rho\in\Path{\rho}$ and $s_\sigma\in\Path{\sigma}$, respectively.
If a boundary diagram is absent, its path set is empty and its contribution below is zero.

The hook-length formula yields
\begin{align}
  \frac{m_\rho}{m_\alpha}
  =
  \frac1n
  \prod_{j=1}^{q-1}\frac{q-j+2}{q-j+1}
  \prod_{a=1}^{d-2}\frac{d-a}{d-a-1}
  =
  \frac{(d-1)(q+1)}{2n},
  \label{eq:direct-rho-dimension-ratio}\\
  \frac{m_\sigma}{m_\alpha}
  =
  \frac1n
  \prod_{j=1}^{q-2}\frac{q-j}{q-j-1}
  \prod_{a=1}^{d-1}\frac{d-a+2}{d-a+1}
  =
  \frac{(d+1)(q-1)}{2n}.
  \label{eq:direct-sigma-dimension-ratio}
\end{align}
All products with an empty index range are understood to equal one.
Substituting Eqs.~\eqref{eq:direct-lambda-diagonal-elements}--\eqref{eq:direct-sigma-dimension-ratio} into Eq.~\eqref{eq:direct-trivial-overlap-expanded} gives
\begin{align}
  &\braket{\Omega_{\lambda\alpha}}{g_{n,\alpha}}
  =
  \frac{1}{\sqrt{m_\alpha}}
  \left(\frac{m_\sigma}{d+1}
  -\frac{m_\rho}{d-1}\right),
  \notag\\
  &\braket{\Omega_{\lambda\alpha}}{g_{n,\alpha}}
  =
  \sqrt{m_\alpha}
  \left(\frac{q-1}{2n}-\frac{q+1}{2n}\right)
  =
  -\frac{\sqrt{m_\alpha}}{n}.
  \label{eq:direct-trivial-overlap-value}
\end{align}
Equations~\eqref{eq:direct-trivial-range}, \eqref{eq:direct-trivial-overlap-reduction}, and \eqref{eq:direct-trivial-overlap-value} imply
\begin{align}
  A_\alpha\ket{+}
  =
  \begin{cases}
    -\sqrt{\frac{m_\alpha}{n}}\,
    \ket{\Omega_{\lambda\alpha}} & (\eta=\alpha)\\
    0 & (\eta = \beta, \gamma)
  \end{cases}.
  \label{eq:direct-frame-trivial-action}
\end{align}
Thus
\begin{align}
  g_\eta^{\mathrm{triv}}
  =
  \begin{cases}
    \frac{m_\alpha}{n} & (\eta=\alpha)\\
    0 & (\eta = \beta, \gamma)
  \end{cases}.
  \label{eq:direct-gram-trivial-eigenvalue}
\end{align}
Moreover, Eq.~\eqref{eq:direct-frame-trivial-action} and the definitions of the centered vectors show that
\begin{align}
  \widetilde{A}_\eta&=A_\eta (\mathbb{I}_n-\ketbra{+}),\\
  \widetilde G_\eta
  &=\widetilde{A}_\eta^\dagger\widetilde{A}_\eta
  =(\mathbb{I}_n-\ketbra{+}) G_\eta (\mathbb{I}_n-\ketbra{+})
  =g_\eta (\mathbb{I}_n-\ketbra{+}).
  \label{eq:direct-centered-gram-proof}
\end{align}

We next calculate the standard eigenvalue from the trace.
Equation~\eqref{eq:direct-gram-decomposition} gives
\begin{align}
  g_\eta
  =
  \frac{\Tr G_\eta-g_\eta^{\mathrm{triv}}}{n-1}.
  \label{eq:direct-gram-std-eigenvalue}
\end{align}
By Eq.~\eqref{eq:direct-gram-symmetry}, all frame vectors in a fixed branch have the same norm, and hence
\begin{align}
  \Tr G_\eta
  &=
  \sum_{i=1}^n\braket{g_{i,\eta}}{g_{i,\eta}} \nonumber\\
  &=
  n\braket{g_{n,\eta}}{g_{n,\eta}} \nonumber\\
  &=
  n\sum_{s_\eta\in\Path{\eta}}
  \sum_{s_\alpha\in\Path{\alpha}}
  \left|
  \bra{s_\alpha\to\mu}\rho_\mu(\tau_n)
  \ket{s_\eta\to\mu}
  \right|^2.
  \label{eq:direct-gram-trace}
\end{align}
For the $\beta$ branch, Lemma~\ref{lem:young-orthogonal-last-swap} gives the only nonzero transition amplitudes
\begin{align}
  \bra{s_\alpha\to\mu}\rho_\mu(\tau_n)
  \ket{s_\beta\to\mu}
  =
  \sqrt{1-\frac{1}{(d-1)^2}}
  \sum_{s_\rho\in\Path{\rho}}\delta_{s_\alpha,s_\rho\to\alpha}\delta_{s_\beta,s_\rho\to\beta}.
  \label{eq:direct-beta-transition}
\end{align}
Therefore
\begin{align}
  \Tr G_\beta
  &=
  n\left(1-\frac{1}{(d-1)^2}\right)m_\rho \nonumber\\
  &=
  \frac{d(d-2)(q+1)}{2(d-1)}\,m_\alpha.
  \label{eq:direct-beta-trace}
\end{align}
Likewise, the only nonzero $\gamma$-branch transition amplitudes are
\begin{align}
  \bra{s_\alpha\to\mu}\rho_\mu(\tau_n)
  \ket{s_\gamma\to\mu}
  =
  \sqrt{1-\frac{1}{(d+1)^2}}
  \sum_{s_\sigma\in\Path{\sigma}}\delta_{s_\alpha,s_\sigma\to\alpha}\delta_{s_\gamma,s_\sigma\to\gamma}.
  \label{eq:direct-gamma-transition}
\end{align}
Hence
\begin{align}
  \Tr G_\gamma
  &=
  n\left(1-\frac{1}{(d+1)^2}\right)m_\sigma \nonumber\\
  &=
  \frac{d(d+2)(q-1)}{2(d+1)}\,m_\alpha.
  \label{eq:direct-gamma-trace}
\end{align}
The one-box restriction of $\cS_\mu$ is the orthogonal direct sum over $\eta\in\mu-\square$.
Completeness of these parent subspaces and unitarity of $\rho_\mu(\tau_n)$ imply
\begin{align}
  \sum_{\eta\in\mu-\square}\Tr G_\eta
  =
  n\sum_{s_\alpha\in\Path{\alpha}}
  \bra{s_\alpha\to\mu}
  \rho_\mu(\tau_n)\rho_\mu(\tau_n)^\dagger
  \ket{s_\alpha\to\mu}
  =
  nm_\alpha.
  \label{eq:direct-total-gram-trace}
\end{align}
Using Eqs.~\eqref{eq:direct-beta-trace} and \eqref{eq:direct-gamma-trace}, we obtain
\begin{align}
  \Tr G_\alpha
  =
  nm_\alpha-\Tr G_\beta-\Tr G_\gamma
  =
  \frac{n+2}{d^2-1}\,m_\alpha.
  \label{eq:direct-alpha-trace}
\end{align}
Finally, substituting Eqs.~\eqref{eq:direct-gram-trivial-eigenvalue}, \eqref{eq:direct-beta-trace}, \eqref{eq:direct-gamma-trace}, and \eqref{eq:direct-alpha-trace} into Eq.~\eqref{eq:direct-gram-std-eigenvalue} yields
\begin{align}
  g_\alpha
  =
  \frac{d^2(q^2-1)}
  {n(n-1)(d^2-1)}\,m_\alpha,
  \notag\\
  g_\beta
  =
  \frac{d(d-2)(q+1)}
  {2(n-1)(d-1)}\,m_\alpha,
  \notag\\
  g_\gamma
  =
  \frac{d(d+2)(q-1)}
  {2(n-1)(d+1)}\,m_\alpha.
  \label{eq:direct-standard-eigenvalues}
\end{align}
These are precisely Eqs.~\eqref{eq:direct-g-alpha}--\eqref{eq:direct-g-gamma}.
\end{proof}

\end{document}